\documentclass[10pt,onecolumn]{article}
\usepackage[utf8]{inputenc}  
\usepackage[T1]{fontenc}
\usepackage[final]{graphicx}
\usepackage{amsmath, amsthm, amssymb}
\usepackage{thmtools}
\usepackage{thm-restate}
\usepackage[left=2cm, right=2cm, bottom=2cm,top=2cm]{geometry}
\usepackage{titling}
\usepackage{mleftright}\mleftright

\usepackage{subcaption}

\usepackage{multicol}
\usepackage{nicematrix}
\usepackage{adjustbox}
\usepackage{stmaryrd}
\usepackage{siunitx}
\usepackage{multirow}
\usepackage{makecell}
\usepackage{xstring}
\usepackage{blkarray}
\usepackage{parskip}
\usepackage{setspace}
\usepackage{xurl}

\usepackage[
    style=wikibibstyle,
    citestyle=numeric-comp,
    sorting=none,
    backref=false,
    minnames=3,
    maxnames=15,
    giveninits=true,
    defernumbers=true
    ]{biblatex}
\NewDocumentCommand{\refcite}{O{}m}{%
\IfSubStr{#2}{,}%
    {Refs.~\cite[#1]{#2}}%
    {Ref.~\cite[#1]{#2}}%
}

\usepackage{ifdraft}
\ifdraft{\newcommand{\authnote}[3]{{\color{#3}({\bf  #1:} #2)}}}{\newcommand{\authnote}[3]{}}

\usepackage{enumerate}
\usepackage[shortlabels]{enumitem}
\usepackage{tikz}
\usetikzlibrary{quantikz2}
\usetikzlibrary{decorations.pathmorphing}

\pgfdeclaredecoration{snakeline}{initial}{
    \state{initial}[width=2pt]{ \pgfpathlineto{\pgfpoint{0pt}{0pt}} }
    \state{final}{ \pgfpathlineto{\pgfpointdecoratedpathlast} }
    \state{step}[width=\pgfdecorationsegmentlength,auto corner]{
        \pgfmathsin{\pgfdecorationsegmentaspect*360*\pgfdecoratedratio}
        \let\x=\pgfmathresult
        \pgfpathlineto{\pgfqpoint{\pgfdecorationsegmentlength}{\pgfdecorationsegmentamplitude*\x}}
    }
}

\pgfdeclarelayer{background}
\pgfdeclarelayer{foreground}
\pgfsetlayers{background,main,foreground}

\makeatletter
\DeclareRobustCommand
  \myvdots{\vbox{\baselineskip4\p@ \lineskiplimit\z@
    \hbox{.}\hbox{.}\hbox{.}}}
\makeatother

\makeatletter
\DeclareRobustCommand
  \myhdots{\hbox{\kern2\p@
    .\kern2\p@.\kern2\p@.}}
\makeatother

\makeatletter
\DeclareRobustCommand
  \myddots{\vbox{\baselineskip4\p@ \lineskiplimit\z@
    \hbox{$\kern0\p@.$}\hbox{$\kern4\p@.$}\hbox{$\kern8\p@.$}}}
\makeatother

\newcommand{\Hsquare}{%
  \text{\fboxsep=-.2pt\fbox{\rule{0pt}{0.8ex}\rule{0.8ex}{0pt}}}%
}

\usepackage{booktabs}
\usepackage{ragged2e}

\usepackage{soul}

\usepackage{array}
\newcolumntype{P}[1]{>{\centering\arraybackslash}m{#1}}
\newcolumntype{U}[1]{>{\centering\arraybackslash}p{#1}}
\newcolumntype{L}[1]{>{\centering\arraybackslash}l{#1}}

\usepackage{xr-hyper}
\usepackage[colorlinks=true, citecolor=blue,linkcolor=blue,final]{hyperref}

\usepackage{physics}
\usepackage[capitalise]{cleveref}
\usepackage{float}
\usepackage{cancel}
\usepackage{appendix}
\usepackage[ruled]{algorithm2e}
\usepackage{import}
\usepackage{relsize}
\usepackage[all]{nowidow}
\usepackage{thmtools, thm-restate}
\usepackage{appendix}
\usepackage{balance}

\newlist{stepenum}{enumerate}{1}
\setlist[stepenum]{label=(\roman*)}
\crefname{stepenumi}{step}{steps}
\Crefname{stepenumi}{Step}{Steps}

\def\autorefapp#1{\hyperref[#1]{Appendix~\ref{#1}}}

\newcommand{\killpunct}[1]{}

\renewcommand{\bra}[1]{\langle #1 |}
\renewcommand{\ket}[1]{|#1\rangle}

\def\ketbra#1{ |{#1}\rangle\!\langle{#1}| }

\newcommand{\nrm}[1]{\lVert #1 \rVert}

\DeclareMathOperator*{\EV}{\mathbb{E}}

\makeatletter
\newcommand{\doublewidetilde}[1]{{%
  \mathpalette\double@widetilde{#1}%
}}
\newcommand{\double@widetilde}[2]{%
  \sbox\z@{$\m@th#1\widetilde{#2}$}%
  \ht\z@=.9\ht\z@
  \widetilde{\box\z@}%
}
\makeatother

\newcommand{\CB}{\mathcal{B}}
\newcommand{\CC}{\mathcal{C}}

\newcommand{\CE}{\mathcal{E}}
\newcommand{\CF}{\mathcal{F}}
\newcommand{\CG}{\mathcal{G}}

\newcommand{\CI}{\mathcal{I}}
\newcommand{\CN}{\mathcal{N}}
\newcommand{\CP}{\mathcal{P}}
\newcommand{\CQ}{\mathcal{Q}}
\newcommand{\CR}{\mathcal{R}} 
 
\newcommand{\CT}{\mathcal{T}} 

\newcommand{\CV}{\mathcal{V}}
\newcommand{\CW}{\mathcal{W}}
\newcommand{\CX}{\mathcal{X}}
\newcommand{\CY}{\mathcal{Y}}

\newcommand{\PauliSet}{\mathbb{P}}
\newcommand{\PauliSetNoSign}{\hat{\mathbb{P}}}
\newcommand{\CliffordSet}{\mathbb{T}}
\newcommand{\Id}{\mathbb{I}}
\newcommand{\Swap}{\mathbb{S}}
\newcommand{\CSwap}{\mathrm{C}\mathbb{S}}

\newcommand{\tCE}{\widetilde{\vphantom{\raisebox{0.14em}{:}}\smash[t]{\mathcal{E}}}}
\newcommand{\tCG}{\widetilde{\mathcal{G}}}

\newcommand{\tCP}{\widetilde{\mathcal{P}}}
\newcommand{\tCQ}{\widetilde{\vphantom{\raisebox{0.14em}{:}}\smash[t]{\CQ}}}

\newcommand{\tCV}{\widetilde{\mathcal{V}}}

\newcommand{\e}{\mathrm{e}}
\newcommand{\iUnit}{\mathrm{i}}
\newcommand{\rd}{\mathrm{d}}

\newcommand{\ol}[1]{%
  \mathpalette\olinner{#1}%
}
\newcommand{\olinner}[2]{%
  \overline{\vphantom{\raisebox{%
\ifx#1\scriptstyle-0.2em\else-0.04em\fi}{\ensuremath{#2}}}\smash[t]{#2}}%
}

\newcommand{\UR}{\mathrm{UR}}
\newcommand{\poly}{\mathrm{poly}}
\newcommand{\polylog}{\mathrm{polylog}}

\newcommand{\tpsig}[1]{\widetilde{\vphantom{\raisebox{0.14em}{:}}\smash[t]{\psi}}(#1)}

\newtheorem{lemma}{Lemma}

\newtheorem{definition}{Definition}
\newtheorem{corollary}{Corollary}
\newtheorem{theorem}{Theorem}

\newtheorem{innercustomthm}{}
\newenvironment{customthm}[1]
  {\renewcommand\theinnercustomthm{#1}\innercustomthm}
  {\endinnercustomthm}

\declaretheorem[
  name=Proposition,
]{proposition}

\usepackage{regexpatch}
\makeatletter
\xpatchcmd\thmt@restatable{%
\csname #2\@xa\endcsname\ifx\@nx#1\@nx\else[{#1}]\fi
}{%
\ifthmt@thisistheone
\csname #2\@xa\endcsname\ifx\@nx#1\@nx\else[{#1}]\fi
\else
\csname #2\@xa\endcsname[{Restated}]
\fi}{}{}
\makeatother

\newcounter{circuit}
\newenvironment{circuit}{%
    \setcounter{circuit}{\value{equation}}%
    \begin{equation}%
    \refstepcounter{circuit}%
}{%
    \end{equation}\ignorespacesafterend%
}%

\crefformat{circuit}{circuit~(#2#1#3)}
\Crefformat{circuit}{Circuit~(#2#1#3)}

\usepackage{abstract}
\usepackage{lipsum}   

\usepackage[affil-it]{authblk} 

\begin{document}

\title{\Large \textbf{A distillation--teleportation protocol for fault-tolerant QRAM}} 

\author{ \normalsize Alexander M. Dalzell,${}^1$
András~Gilyén,${}^2$
Connor~T.~Hann,${}^1$
Sam~McArdle,${}^1$
Grant~Salton,${}^{1,3}$
Quynh~T.~Nguyen,${}^4$
Aleksander~Kubica,${}^{1,5}$
Fernando~G.S.L.~Brand\~ao${}^{1,6}$
}

\affil{ \small  ${}^1$ AWS Center for Quantum Computing, Pasadena, CA, USA \\
${}^2$  HUN-REN Alfréd Rényi Institute of Mathematics, Budapest, Hungary \\
${}^3$ Amazon Quantum Solutions Lab, Seattle, WA, USA \\
${}^4$ School of Engineering and Applied Sciences, Harvard University, Cambridge, MA, USA\\
${}^5$  Yale Quantum Institute \& Department of Applied Physics, New Haven, CT, USA \\
${}^6$  Institute for Quantum Information and Matter, Caltech, Pasadena, CA, USA
}

\date{}
\maketitle 
\vspace{-14pt}
\begin{onecolabstract}
\vspace{10pt}
 

We present a protocol for fault-tolerantly implementing the logical quantum random access memory (QRAM) operation,
given access to a specialized, noisy QRAM device.
For coherently accessing classical memories of size $2^n$, our protocol consumes only $\mathrm{poly}(n)$ fault-tolerant quantum resources (logical gates, logical qubits, quantum error correction cycles, etc.), avoiding the need to perform active error correction on all $\Omega(2^n)$ components of the QRAM device.   
 This is the first rigorous conceptual demonstration
 that a specialized, noisy QRAM device could be useful
 for implementing a fault-tolerant quantum algorithm. In fact, the fidelity of the device can be as low as $1/\poly(n)$.
 The protocol queries the noisy QRAM device $\poly(n)$ times to prepare a sequence of $n$-qubit QRAM \textit{resource states}, which are moved to a general-purpose $\poly(n)$-size processor to be encoded into a QEC code, distilled, and fault-tolerantly teleported into the computation. To aid this protocol, we develop a new gate-efficient streaming version of quantum purity amplification that matches the optimal sample complexity in a wide range of parameters and is therefore of independent interest.
 
 The exponential reduction in fault-tolerant quantum resources comes at the expense of an exponential quantity of purely classical complexity---each of the $n$ iterations of the protocol requires adaptively updating the $2^n$-size classical dataset and providing the noisy QRAM device with access to the updated dataset at the next iteration. 
 While our protocol demonstrates that QRAM is more compatible with fault-tolerant quantum computation than previously thought, the need for significant classical computational complexity exposes potentially fundamental limitations to realizing a truly $\mathrm{poly}(n)$-cost fault-tolerant QRAM. 
\vspace{14pt}
\end{onecolabstract}

\tableofcontents

\newpage 

\section{Introduction}\label{sec:introduction}

The development of fast and large-scale random access memory (RAM) has played an indispensable role in the development of conventional computing. Early RAM devices assisted in the first demonstrations of stored-program electronic computers \cite{williams1948electronicDigitalComputers,williams1949storageSystem,manchesterBabyWiki}, and today, the availability of efficient high-speed RAM enables data-intensive computing applications in areas like machine learning \cite{kim2023fullStackOptimization, silvano2024surveydeeplearninghardware}.
 At an abstract level, RAM performs the following operation: take as input an $n$-bit address $x$ specifying a location in memory, and retrieve one of the $2^n$ data items $f(x)$, labeled by $x$. In practice, the access time for RAM is astonishingly fast: modern RAM chips can achieve latency of 10 nanoseconds or faster.\footnote{RAM latency is a complex topic, due to the many kinds of memory that are used in a computer, each type offering benefits and drawbacks on a number of dimensions including speed, physical size per bit, volatility, and price. See \refcite{SRAMexample} and the associated datasheet found there for a concrete example of an asynchronous static RAM chip achieving 10 nanoseconds latency on $2^{20}$ memory locations (each storing 8 bits).} Furthermore, the RAM runtime is independent of the location of the data within the memory, and the latency can remain nearly unchanged even as the overall size of the memory is scaled up. 

The idea of quantum random access memory (QRAM) \cite{giovannetti2007QuantumRAM,jaques2023qram}
is to achieve something similar even when the $n$-qubit address register is in a quantum superposition $\sum_{x} \alpha_x \ket{x}$ of all $2^n$ addresses. For simplicity, we consider the most basic version of QRAM: applying a phase $(-1)^{f(x)}$ onto basis state $\ket{x}$, where the $2^n$ binary values $f(0), f(1), \ldots, f(2^n-1)$ are stored in classical memory.
\begin{align}\label{eq:QRAM_intro}
   \text{QRAM operation}\colon \qquad  \sum_{x \in \{0,1\}^n} \alpha_x \ket{x} \quad \overset{\raisebox{3pt}{$\mathlarger{V(f)}$}}{\longmapsto} \quad  \sum_{x \in \{0,1\}^n} (-1)^{f(x)} \alpha_x \ket{x}.
\end{align}
We refer to the $n$-bit Boolean function $f\colon \{0,1\}^n\rightarrow \{0,1\}$ as the classical \textit{dataset} or \textit{data table} we want to query. For each $f$, the $n$-qubit unitary $V(f)$ that implements the QRAM operation is  diagonal in the computational basis, with diagonal $\pm 1$ entries determined by $f$. We note that the controlled $V(f)$ operation\footnote{\label{footnote:controlled-V(f)}Controlled $V(f)$ is equivalent to (non-controlled) $V(\hat{f})$ for a dataset $\hat{f}$ with $n+1$ address bits; see \cref{app:extension_to_b_bits}.} can be used to implement the more familiar formulation of QRAM, which reads the classical data into an ancilla register as $\ket{x}\ket{0}\mapsto \ket{x}\ket{f(x)}$.\footnote{In some places in the literature \cite{jaques2023qram}, the operation $\ket{x}\ket{0}\mapsto \ket{x}\ket{f(x)}$ is referred to as QRACM, where the additional C emphasizes the fact that the data $f(x)$ are classical and thus the states $\ket{f(x)}$ are computational basis states. This distinguishes QRACM from its generalization, QRAQM, where each $\ket{f(x)}$ can be an arbitrary (possibly multiqubit) quantum state. In this paper, we do not consider QRAQM, and we refer to QRAM interchangeably with QRACM.}

Whereas a single classical RAM query can access at most one entry of a data table, a single QRAM query suffices to create a superposition over all $2^n$ entries.  The assumption that QRAM is cheap and available underlies a number of proposed quantum algorithms (see \refcite{ciliberto2018QMLReview, biamonte2016QuantumMachineLearning,jaques2023qram,dalzell2023quantumAlgorithmsSurvey} for relevant surveys), which leverage this ability to offer up to exponential speedups over their classical counterparts. 
Often, the need for QRAM in these algorithms is contained within an unspecified oracle or data access assumption.
For instance, quantum machine learning algorithms for  support vector machines \cite{rebentrost2014QSVM}, Gaussian process regression \cite{zhao2015QAssisstedGaussProcRegr}, and recommendation systems \cite{kerenidis2016QRecSys} require only a polylogarithmic (in the size of the dataset) number of queries to an oracle that accesses (in superposition) the entries of a classical matrix or vector. Similarly, quantum algorithms for solving differential equations \cite{berry2014high,berry2017diffEQExponentialPrecision,childs2020quantum,krovi2023improved,jennings2023CostDiffEQ,Berry2024quantumalgorithm} discretize the equations and invert the resulting linear systems \cite{harrow2009QLinSysSolver}, in some cases incurring only a polylogarithmic (in the size of the linear system) number of queries to the classical data defining the instance, such as object geometries and boundary conditions. As a final example, quantum algorithms for solving optimization problems like semidefinite and linear programs \cite{brandao2016QSDPSpeedup,brandao2017QSDPSpeedupsLearning,apeldoorn2017QSDPSolvers,apeldoorn2019QAlgorithmsForZeroSumGames,apeldoorn2018ImprovedQSDPSolving,kerenidis2018QIntPoint,kerenidis2019QAlgsSecondOrderConeSVM,augustino2021quantum}, with applications in logistics and finance \cite{kerenidis2019PortfolioOptimization,dalzell2022socp},  require coherent oracle access to the classical matrices defining the optimization problem. 
In all of these areas, the claimed speedup is typically dependent upon the assumption that---at least at an abstract level---the cost of QRAM is similar to that of RAM. 
\vspace{-0.5\baselineskip}
\begin{quotation}
\begin{customthm}{Cheap QRAM assumption}
    For an arbitrary data table $f$, the computational cost of implementing the unitary operation $V(f)$ from \cref{eq:QRAM_intro} is $\poly(n)$.
\end{customthm}
\end{quotation}
\vspace{0.5\baselineskip}
\noindent Here, the term \textit{computational cost} is intentionally vague---depending on the context, it might refer to circuit depth, physical runtime, energy dissipated, or some other metric---one must define it more precisely before justifying the assumption  (see discussion in \refcite{jaques2023qram}). Focusing on physical runtime/latency as a metric, the assumption of $\poly(n)$ cost is roughly valid in the case of RAM: one can write down classical circuits for RAM that have $O(n)$ depth, and in practice actual RAM chips maintain extremely fast latency even at very large scale.  However, for QRAM, the validity of this assumption has been the source of significant controversy  \cite{arunachalam2015RobustnessBuckBrigQRAM,ciliberto2018QMLReview,jaques2023qram,steiger2016RacingInParallel}. 
At the root of the issue is the fact that, unlike RAM, QRAM must be implemented in such a way that information about \textit{which} address is being queried is not leaked to the environment, which would lead to decoherence. Strategies for preventing this decoherence without also reducing QRAM's relative power have so far proved to be elusive. 

One might try to justify the cheap QRAM assumption by writing down an $O(n)$-depth quantum circuit  for the $n$-qubit unitary $V(f)$ \cite{dimatteo2020FaultTolerantQRAM,paler2020parallelizingBucketBrigade,hann2021resilienceofQRAM,mukhopadhyay2024qRAMloglogTdepth}, and then running that circuit on a general-purpose fault-tolerant quantum processor; assuming gates can be implemented in parallel, $O(n)$ latency is achievable. This strategy---referred to as ``circuit QRAM'' in \refcite{jaques2023qram}---has significant drawbacks. In particular, it requires $\Omega(2^n)$ logical ancilla qubits and $\Omega(2^n)$ classical co-processors to control the system and perform active error correction on all its components in parallel. Each logical ancilla may require dozens or hundreds of physical qubits, leading to an extremely large device footprint, a conclusion that is further exacerbated by the presence of a large number of magic state factories for implementing in parallel the non-Clifford $T$ or Toffoli gates in the circuit, of which there must be at least $\Omega(\sqrt{2^n})$ \cite{low2018tradingTgatesforDirtyQubits}. One estimate for a surface code approach found that \textit{quadrillions} of physical qubits would be needed for querying an 8-gigabyte memory  \cite{dimatteo2020FaultTolerantQRAM}.  The opportunity cost of these quantum and classical resources is steep. For example, the $O(2^n)$ classical co-processors can perform complex tasks like sparse matrix-vector multiplication for $2^n \times 2^n$ matrices in $\poly(n)$ time \cite{steiger2016RacingInParallel,ciliberto2018QMLReview,jaques2023qram}. Consequently, for circuit QRAM, the cheap QRAM assumption is only justifiable in a cost model that essentially precludes the possibility of quantum advantage in many proposed applications. 

Ideally, the QRAM operation would instead be carried out by a specialized hardware element, separate from the main general-purpose quantum processor, mirroring how RAM is performed in a different way than computation on the main CPU. While the physical size of the QRAM hardware element would scale as $\Omega(2^n)$---this is necessary simply to store the dataset $f$---the runtime could be as little as $O(n)$. Furthermore, since the device is specialized for QRAM, it could in principle be performed \textit{passively} and \textit{ballistically} \cite{jaques2023qram}, that is, implemented automatically by natural evolution of the system while requiring at most $\poly(n)$ external interventions from classical control and dissipating at most $\poly(n)$ energy.\footnote{Reading from a classical RAM can be viewed as a passive operation: although the circuit for RAM has $\Omega(2^n)$ gates/components, these gates are etched onto the chip and are performed without any external intervention---one simply needs to set the voltages on the $n$ input pins specifying the desired address. It is possible to design a RAM circuit that dissipates only $O(n)$ energy, although since this does not represent a bottleneck in practical systems, actual RAM chips are better modeled as dissipating $O(\sqrt{2^n})$ energy (memory is laid out in 2D and in practice an entire row/column is activated, rather than just a single memory cell) \cite{jaques2023qram,jaeger2016microelectronic}.} Constructing such a device is a formidable engineering challenge; there are currently no fully convincing proposals on how it could be done, but nothing rules it out in theory;\footnote{We ignore speed-of-light constraints, which we expect only to be relevant at large QRAM size \cite{wang2024qramFundamentalCausalBounds}. At large enough scale, the speed of light would prevent both RAM and QRAM from  achieving query latency $O(n)$, since the dataset of size $2^n$ must be embedded in 2 or at most 3 spatial dimensions, and the time needed for information to travel across the device would be at least  $\Omega(2^{n/3})$ (in the case of a 3D embedding).} see \cref{fig:passive_QRAM} for an abstract picture of how such a device might be structured. 

 Yet, even if a passive, physical QRAM device did exist, it is unclear how it could actually be useful to a fault-tolerant quantum computation (FTQC). The ability to apply the physical QRAM operation at computational cost $\poly(n)$ is not sufficient to justify the cheap QRAM assumption, even if there are no errors in the QRAM device itself (which is not realistic anyway). 
The problem is that, in FTQC,  we need to perform the logical QRAM operation, denoted by $\ol{V(f)}$, onto an address register encoded into some quantum error-correcting (QEC) code.  
Naively, we could implement $\ol{V(f)}$ by  un-encoding the $n$ logical qubits into $n$ physical qubits, running the physical qubits through the physical QRAM device, and re-encoding the output. However, the un-encoding and re-encoding processes introduce uncorrectable errors, and any noise in the physical QRAM device will also propagate into logical errors on the re-encoded state.  
An alternative would be to find a QEC code where the logical $\ol{V(f)}$ is a transversal gate, meaning it can be implemented as a tensor product of $\poly(n)$ physical $V(f)$ gates without the need for un-encoding and re-encoding. 
Unfortunately, there are known challenges to finding such codes \cite{jaques2023qram}. A general $n$-qubit QRAM gate is in the $n$-th level of the Clifford hierarchy (see \cref{sec:teleportable_gates} and \cref{app:QRAM_in_Clifford_hierarchy}), and all known examples of codes supporting transversal implementation of a gate in the $n$-th level have $O(2^{n})$ qubits~\cite{zeng2008semi, hastings2018distillation, kubica2015universalTransversalGates}. In fact, there is at least one example of a gate in the $n$-th level---the single-qubit $\pi/2^n$ rotation gate---where a matching lower bound of $\Omega(2^n)$ qubits has been shown for a strong form of transversality \cite{koutsioumpas2022smallest}, leading one to speculate that a similar lower bound may hold for QRAM, as well.

\begin{figure}[htb]
    \centering
        \begin{subfigure}[b]{0.32\textwidth}
        \centering
        \includegraphics[ width=\textwidth]{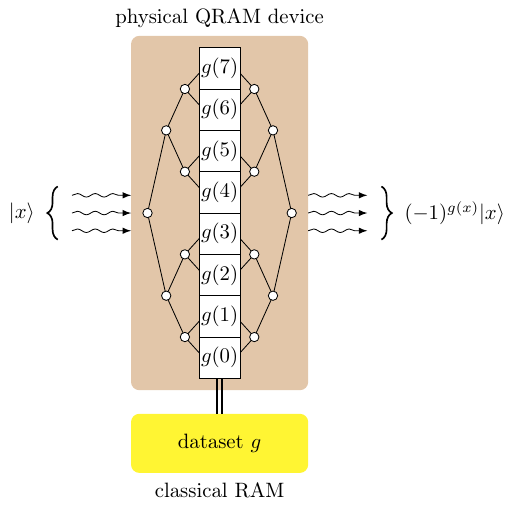}
        \caption{
        Possible structure of physical QRAM device passively implementing $V(g)$ from \cref{eq:QRAM_intro}
        }
        \label{fig:passive_QRAM}
    \end{subfigure}
    \hspace{30pt}  
    \begin{subfigure}[b]{0.6\textwidth}
        \centering
        \includegraphics[ width=\textwidth]{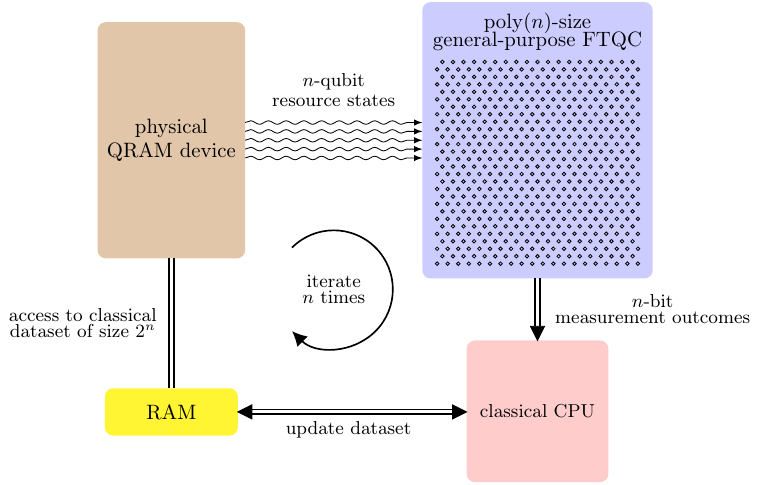}
        \caption{Setup of our protocol}
        \label{fig:setup}
    \end{subfigure}
    \caption{(a)  Ideally, the physical QRAM operation $V(g)$ of \cref{eq:QRAM_intro} is performed passively by a specialized device.  We may imagine, for example,  encoding the address state $\sum_{x}\alpha_x \ket{x}$ into the polarization states of $n$ photons, and then sending them into a pre-manufactured device, where they return having picked up a $-1$ phase only on branches of the superposition where $g(x)=1$ \cite{kuperberg2013anotherSubexponential,jaques2023qram}. This might be accomplished by placing the classical bits $g(0), g(1), \ldots, g(2^{n}-1)$ at the leaves of a binary tree, selectively routing the photons to the correct leaf based on their polarization, picking up a phase if a photon exists at location $x$ and $g(x) = 1$, and then unrouting the $n$ photons. (b) Our protocol utilizes a specialized, physical QRAM device, which is separate from the general-purpose fault-tolerant quantum processor. The QRAM device is used to create QRAM resource states on $n$ physical qubits which are moved onto the main processor. The main processor encodes, distills, and teleports these resource states, generating classical $n$-bit measurement outcomes, which are sent to a classical CPU. The classical CPU performs a calculation to update the dataset stored in classical memory (RAM), which is queried by the physical QRAM device at the next iteration of the protocol. }
    \label{fig:intro_figure}
\end{figure}

Our main contribution is to devise a protocol that implements the logical operation $\ol{V(f)}$ fault tolerantly, using $\poly(n)$ queries to a noisy device that can implement the physical QRAM operation with at least $1/\poly(n)$ fidelity, as well as $\poly(n)$ fault-tolerant operations on a general-purpose quantum processor---exponentially lower than the number of fault-tolerant quantum operations required for circuit QRAM. The protocol generalizes well-known distillation--teleportation protocols for non-Clifford gates like the $T$ gate and the $\mathrm{CCZ}$ gate. First, the physical $V(f)$ gate is used to prepare many copies of a faulty physical $n$-qubit QRAM resource state. Next, the physical resource states are encoded into a QEC code (the protocol is agnostic to which one) and distilled into a single high-fidelity logical resource state. 
Finally, the high-fidelity logical resource state is teleported into the computation to enact the logical QRAM gate, up to a correction which can be computed classically---our protocol outsources this calculation to a classical processor as depicted in \cref{fig:setup}. The required correction is a different logical QRAM gate $\ol{V(f')}$, where $f'$ is determined by $f$ and random measurement outcomes obtained during the teleportation procedure. The correction $\ol{V(f')}$ is then implemented in the same way, requiring a correction of its own, $\ol{V(f'')}$, where $f''$ is again dependent on $f'$ and random measurement outcomes. We show that after iterating this process for $n$ rounds, no further correction is necessary. This is a consequence of the fact that, despite its exponential circuit complexity, the unitary $V(f)$ lies in the $n$-th level of the Clifford hierarchy \cite{cui2017diagonalGatesCliffordHierarchy} for every $f$, which implies that the first correction $V(f')$ is in the $(n-1)$-th level, the second correction $V(f'')$ is in the $(n-2)$-th level, and so on.  
Our insights are (i) to notice that these corrections always lie within the family of QRAM gates  of \cref{eq:QRAM_intro}, allowing for a straightforward recursive implementation, and (ii) to devise a method for preparing the high-fidelity encoded resource states, completing the end-to-end workflow for fault-tolerant $\ol{V(f)}$. A no-go theorem in \refcite{jaques2023qram} ruled out a wide class of QRAM distillation--teleportation protocols; our protocol sidesteps this theorem by being adaptive and querying the physical QRAM on different datasets ($f$, $f'$, $f''$, etc.) in each round. We provide a more complete informal overview of the protocol in \cref{sec:protocol_overview} and a detailed error analysis of each step in \cref{sec:protocol_detailed}.

By showing how to perform logical $\ol{V(f)}$ using $\poly(n)$ calls to physical QRAM, our protocol salvages the potential utility of the specialized, faulty QRAM device, and it encourages a  model of quantum computation where QRAM is performed separately from the main quantum processing unit. 

Our protocol also \textit{partially} justifies the cheap QRAM assumption, provided that a passive QRAM device can be constructed. Indeed, if noisy physical QRAM has computational cost $\poly(n)$, then the quantum resources required to implement fault-tolerant QRAM via our protocol also scales only as $\poly(n)$. The main caveat is that running our protocol requires a non-negligible amount of adaptive \textit{classical} computation of complexity $O(2^n)$ to compute the required correction operations (and ``reload'' the passive QRAM device, so that it has access to the new classical dataset at the next round of the protocol), although this complexity may be amenable to some degree of parallelization. We explore the nuances of this caveat in \cref{sec:complexity_of_classical_update}. In our protocol, this adaptive classical computation and QRAM reloading appears necessary in order to avoid revealing which address is being queried even while using a noisy QRAM device. In \cref{sec:outlook}, we pose the question of whether this reflects an inevitable limitation of fault-tolerant QRAM or whether a stronger justification of the cheap QRAM assumption, where both quantum and classical resources are $\poly(n)$, may be possible.

In any case, our protocol can be viewed as trading $O(2^n)$ quantum resources for $O(2^n)$ classical resources.  That is, our protocol does not require the $O(2^n)$ actively error-corrected quantum resources incurred in circuit QRAM (fault-tolerant quantum gates, ancilla qubits, magic state factories, control wiring, classical co-processors, etc.). Instead, it requires $\poly(n)2^n$ purely classical resources in addition to only $\poly(n)$ fault-tolerant quantum resources and $\poly(n)$ queries to a faulty QRAM device. There may be applications where such a tradeoff is beneficial, since quantum devices will be significantly slower and more expensive than classical devices for the foreseeable future.

\section{Overview of protocol}\label{sec:protocol_overview}

\def\oColor{brown!45}
\def\dColor{teal!25}
\def\tColor{blue!70!cyan}
\def\urColor{yellow!80!white}
\def\sliceColor{green!80!black}
\def\errorColor{red}
\definecolor{lightseagreen}{cmyk}{0.52,0,0.04,0.1}
\def\twirlColor{lightseagreen}
\definecolor{myTan}{cmyk}{0,0.01,0.11,0.07}
\def\eColor{myTan}

\newcommand{\pgfmnc}{\pgfmatrixnextcell}
\def\thickWire{1.1pt}

\newcommand{\oGate}{
\gate[style={fill=\oColor, text height = 0.3em}]{\mbox{\!\tiny$\Psi$\!}}
}
\newcommand{\oGateRed}{
\gate[style={fill=\oColor, draw=red, text height = 0.3em}]{\mbox{\!\tiny$\Psi$\!}}
}
\newcommand{\oGateNormalSize}{
\gate[style={fill=\oColor, text height = 0.3em}]{\Psi}
}
\newcommand{\oGateRedNormalSize}{
\gate[style={fill=\oColor, draw=red, text height = 0.3em}]{\Psi}
}
\newcommand{\eGate}{\gate[style={fill=\eColor}][0em][0em]{{\mbox{\!\tiny$E$}\!}}}
\newcommand{\eGateRed}{\gate[style={fill=\eColor, draw=red}][0em][0em]{{\mbox{\!\tiny$E$}\!}}}
\newcommand{\eGateNormalSize}{\gate[style={fill=\eColor}][0em][0em]{E}}
\newcommand{\eGateRedNormalSize}{\gate[style={fill=\eColor, draw=red}][0em][0em]{E}}

\newcommand{\dGate}{\gate[5,label style = {rotate=270}, style = {text width = 1.8em, fill=\dColor}]{{{\mbox{$\ol{\mathrm{DISTILL}}$}}}}}
\newcommand{\urGate}{\gate[style={fill=\urColor}]{\rm UR}}

\subsection{Warm-up: distillation--teleportation protocol for the \texorpdfstring{$T$}{T} gate}

In many schemes for FTQC, it is relatively cheap to implement logical Clifford gates (e.g., they can often be done transversally). On the other hand, non-Clifford logical gates like the $T$ gate and the $\mathrm{CCZ}$ gate are more expensive; these gates can instead be performed using distillation--teleportation protocols. In this section, we review such a protocol for the $T$ gate, a diagonal gate mapping $\ket{0} \mapsto \ket{0}$ and $\ket{1} \mapsto \e^{\iUnit \pi/4}\ket{1}$. Although the $\mathrm{CCZ}$ gate is actually a special case of the QRAM operation from \cref{eq:QRAM_intro} and thus more directly related to our protocol, the $T$ gate provides a gentler introduction because it is a single-qubit gate. 

We discuss distillation and gate teleportation separately, beginning with gate teleportation. Henceforth, we denote logical states and operations with an overline, for example, we denote the logical $T$ gate by $\ol{T}$. Teleporting $\ol{T}$ into a quantum computation requires the preparation of a resource state, also known as a magic state, which is $\ol{T}$  applied to the equal superposition state:
\begin{equation}\label{eq:T_magic_state}
    \ket{\ol{T}} = \ol{T}\ket{\ol{+}} = \frac{1}{\sqrt{2}}(\ket{\ol{0}} + \e^{\iUnit \pi/4}\ket{\ol{1}})\,,
\end{equation}
where $\ket{\ol{0}}$ and $\ket{\ol{1}}$ denote the encoded computational basis for some QEC code (here we are agnostic to which code), and $\ket{\ol{+}} = \frac{1}{\sqrt{2}}(\ket{\ol{0}} + \ket{\ol{1}})$.  
The $\ol{T}$ gate can then be applied to an arbitrary quantum state $\ket{\ol{\alpha}}$ (on one logical qubit) by entangling $\ket{\ol{\alpha}}$ with $\ket{\ol{T}}$ and making a (logical) measurement, as follows:
\begin{circuit}\label{eq:T_teleportation_figure_uncorrected}
\tikzset{
    operator/.style={draw, filling, inner sep=2pt, thickness, align=center, baseline=(current bounding box.center)},
    internal/.style={thickness, line width=\thickWire}
}
    \scalebox{1.0}{
    \begin{quantikz}[
        row sep={0em,between origins}, 
        column sep={0em, between origins}, 
        align equals at=1.5, 
        wire types = {n,n}
    ]
    \lstick{$\ket{\ol{\alpha}}$} &  \wire[r][4][line width = \thickWire]{q} &[1em]  &[1em] \ctrl[style = {line width = \thickWire, draw=\tColor, fill=\tColor}]{1} &[2.5em]  &[3em] \rstick{
    $
    \begin{cases}
        \ol{T} \ket{\ol{\alpha}} & \text{if } m=0 \\
        \ol{X}\,\ol{T}\, \ol{X}\ket{\ol{\alpha}} & \text{if } m=1
    \end{cases}
    $
    }  \\[3em]
    \lstick{$\ket{\ol{T}}$} & \wire[r][3][line width = \thickWire]{q} &   & \targ[style={line width = \thickWire, draw=\tColor}, internal/.append style={line width=\thickWire}]{} &  \meter[style={line width = \thickWire,draw=\tColor}]{} \wire[r][1]{c}& \rstick{$m$}
    \end{quantikz}
    }
\end{circuit}
The logical CNOT gate (a Clifford gate) and the single-qubit logical measurement are both performed fault-tolerantly within the QEC code to ensure negligible chance of logical error. Direct computation verifies that the single-qubit measurement outcome $m \in \{0,1\}$ is uniformly random, regardless of the state $\ket{\ol{\alpha}}$. If the measurement outcome is $m=0$, the gate $\ol{T}$ is exactly implemented on the top wire. However, if the outcome $m=1$ is obtained, the wrong phase was applied to the state, equivalent to the gate $\ol{X} \, \ol{T}\,  \ol{X}$ instead of $\ol{T}$ (where $X,Y,Z$ denote the Pauli operators). 
To fix this, one must apply a \textit{correction} operation when the measurement outcome is 1. The correction required to undo the erroneous $\ol{X} \, \ol{T} \, \ol{X}$ gate and re-do the  $\ol{T}$ gate is the $\ol{T} \, \ol{X} \,{\smash{\ol{T}}}^\dag \ol{X}$ gate, which is equal to the phase gate $\ol{S} = {\smash{\ol{T}}^2}$, up to a global phase. Crucially, the phase gate is a Clifford gate, and thus the logical $\ol{S}$ can typically be implemented fault tolerantly in a more direct fashion. The full gate teleportation circuit with the correction is then given by:
\begin{circuit}\label{eq:T_teleportation_figure}
\tikzset{
    operator/.style={draw, filling, inner sep=2pt, thickness, align=center, baseline=(current bounding box.center)},
    internal/.style={thickness, line width=\thickWire}
}
    \scalebox{1.0}{
    \begin{quantikz}[
        row sep={0em,between origins}, 
        column sep={0em, between origins}, 
        align equals at=1.5, 
        wire types = {n,n}
    ]
    \lstick{$\ket{\ol{\alpha}}$} &  \wire[r][4][line width = \thickWire]{q} &[1em]  &[1em] \ctrl[style = {line width = \thickWire, draw=\tColor, fill=\tColor}]{1} &[2.5em] \gate{\ol{S}} \wire[d][1]{c} &[3em] \rstick{$\ol{T} \ket{\ol{\alpha}}$}  \\[3em]
    \lstick{$\ket{\ol{T}}$} & \wire[r][3][line width = \thickWire]{q} &   & \targ[style={line width = \thickWire, draw=\tColor}, internal/.append style={line width=\thickWire}]{} &  \meter[style={line width = \thickWire,draw=\tColor}]{} &
    \end{quantikz}
    }
\end{circuit}

The benefit of performing the $\ol{T}$ gate via gate teleportation is that the difficulty is reduced to preparing a high-fidelity $\ket{\ol{T}}$ state.  This state can be prepared through a multistep process of physical preparation, encoding, and then distillation. For concreteness, one can consider magic state injection schemes for the surface code \cite{horsman2012latticeSurgery,li2015magic, lodyga2015simpleSchemeEncoding,litinski2019magicstate}. Here, the first step is to prepare the $\ket{T}$ state on a single physical qubit.
Next, an encoding procdure is performed, that is, $\ket{T}$ is mapped to $ \ket{\ol{T}}$, which is encoded in a $d \times d$ surface code patch. 
This can be realized by, for instance, preparing a product state and performing appropriate stabilizer measurements.
This procedure is generally not fault-tolerant; if the underlying hardware has error rate $p$, the logical error on the prepared state is $O(p)$, but the logical error can be kept independent of how large one makes the code distance $d$. Some of the possible logical errors are heralded---they can be detected by applying certain checks, in which case the procedure can be restarted from scratch, improving the postselected fidelity. The final step is magic state distillation, whereby multiple noisy $\ket{\ol{T}}$ states are consumed to produce a smaller number of higher-fidelity $\ket{\ol{T}}$ states.  
For example, the 15-to-1 magic state distillation protocol uses 15 input magic states of error rate $p_{\rm in}$, succeeds with probability $1-O(p_{\rm in})$,
and conditioned on success, produces a single output magic state of error rate $p_{\rm out} = O(p_{\rm in}^3)$ \cite{bravyi2005UniversalQC,knill2004FTPostSelectedQC,litinski2019magicstate}. By recursively applying this protocol, one can distill $\ket{\ol{T}}$ states with arbitrarily low error rate even when all physical components have noise rate $p = O(1)$, provided that $p$ is below a certain threshold for state distillation. 

In some instances, one may not want to perform the corrective $\ol{S}$ gate directly.\footnote{For some codes, such as the color code~\cite{bombin2006topologicalQuantumDistillation}, the $\ol{S}$ gate is transversal; for the surface code, however, it is fold-transversal~\cite{kubica2015unfoldingColorCode,moussa2016transversalCliffordGates}, and therefore more challenging to implement. One may also prefer to use autocorrected gadgets~\cite{litinski2019gameofsurfacecodes, gidney2019flexibleLayoutSurfaceCode} that avoid direct implementation of $\ol{S}$.}
In this case, another option is to perform the $\ol{S}$ gate also via gate teleportation, using the resource state 
\begin{equation}
    \ket{\ol{S}} = \ol{S}\ket{\ol{+}} = \frac{1}{\sqrt{2}}(\ket{\ol{0}} + \iUnit\ket{\ol{1}})\,.
\end{equation}
The conditional correction required when teleporting $\ol{S}$ is the gate $\ol{S}\, \ol{X}\, {\smash{\ol{S}}}^\dag \ol{X} \propto {\smash{\ol{S}}}^2 = \ol{Z}$. While implementing the Pauli $\ol{Z}$ gate is typically easy for FTQC schemes, it could in principle also be implemented by teleporting the resource state
\begin{align}
    \ket{\ol{-}} = \ol{Z}\ket{\ol{+}} = \frac{1}{\sqrt{2}}( \ket{\ol{0}} - \ket{\ol{1}})\,.
\end{align} 
Teleporting the $\ket{\ol{-}}$ state requires no correction, regardless of the measurement outcome, since $\ol{Z}\,\ol{X}\,\ol{Z}\,\ol{X} \propto \ol{\Id}$, where $\ol{\Id}$ is the (logical) identity operator. 
Following this strategy, we can write the following circuit, which implements the $\ol{T}$ gate via three successive teleportations, where the second and third teleportations are applied only if all prior measurement outcomes are 1. 

\def\firstCorrColor{yellow!30!white}
\def\secondCorrColor{yellow!70!white}
\begin{circuit}\label{eq:T_repeated_teleportation}
\tikzset{
    operator/.style={draw, filling, inner sep=2pt, thickness, align=center, baseline=(current bounding box.center)},
    internal/.style={thickness, line width=\thickWire}
}
    \scalebox{1.0}{
    \begin{quantikz}[
        row sep={0em,between origins}, 
        column sep={0em, between origins}, 
        align equals at=1.5, 
        wire types = {n,n},
        execute at end picture={
            \begin{pgfonlayer}{background}
            \draw[line width=\thickWire, fill = \firstCorrColor, rounded corners, dotted] 
                    ($(\tikzcdmatrixname-1-6)+(0,1)$) 
                    rectangle 
                    ($(\tikzcdmatrixname-2-15)+(1.25,-1)$);
            \draw[line width=\thickWire, fill = \secondCorrColor, rounded corners, dotted ] 
                    ($(\tikzcdmatrixname-1-11)+(0,0.5)$) 
                    rectangle 
                    ($(\tikzcdmatrixname-2-15)+(0.75,-0.5)$);
            \end{pgfonlayer}
        }
    ]
    \lstick{$\ket{\ol{\alpha}}$} 
    &  \wire[r][14][line width = \thickWire]{q} 
    &[1em]  
    &[1em] \ctrl[style = {line width = \thickWire, draw=\tColor, fill=\tColor}]{1} 
    &[2.5em] 
    &[2.8em] 
    &[3.2em]  
    &[1em]  
    &[1.5em]  \ctrl[style = {line width = \thickWire, draw=\tColor, fill=\tColor}]{1} 
    &[2.5em] 
    &[2.8em] 
    &[3.2em] 
    &[1em] 
    &[1.5em]  \ctrl[style = {line width = \thickWire, draw=\tColor, fill=\tColor}]{1}
    &[2.5em] 
    &[5em]\rstick{$\ol{T} \ket{\ol{\alpha}}$}  \\[3em]
    %
    %
    \lstick{$\ket{\ol{T}}$} 
    & \wire[r][3][line width = \thickWire]{q} 
    &   
    & \targ[style={line width = \thickWire, draw=\tColor}, internal/.append style={line width=\thickWire}]{} 
    &  \meter[style={line width = \thickWire,draw=\tColor}]{} \wire[r][1]{c}
    & \phase{} 
    &\lstick{$\ket{\ol{S}}$} \wire[r][3][line width = \thickWire]{q} 
    & 
    & \targ[style={line width = \thickWire, draw=\tColor}, internal/.append style={line width=\thickWire}]{}
    & \meter[style={line width = \thickWire,draw=\tColor}]{} \wire[r][1]{c}
    & \phase{} 
    &\lstick{$\ket{\ol{-}}$} \wire[r][3][line width = \thickWire]{q} 
    & 
    & \targ[style={line width = \thickWire, draw=\tColor}, internal/.append style={line width=\thickWire}]{}
    & \meter[style={line width = \thickWire,draw=\tColor}]{} 
    &
    \end{quantikz}
    }
\end{circuit}
This approach may strike the reader as unnecessary, but designing the procedure in this iterative way will mirror the structure of our full protocol for QRAM.

\subsection{Teleportable gates and the Clifford hierarchy}\label{sec:teleportable_gates}

Not all gates can be teleported in the manner of \cref{eq:T_teleportation_figure}. The key reason it works is that the correction operation, $\ol{S}$, is a Clifford gate. In general, if one attempts to teleport a diagonal single-qubit logical gate $\ol{G}$, the conditional correction is $\ol{G}\,\ol{X}\,\ol{G}^\dag \, \ol{X}$. Early work on gate teleportation \cite{gottesman1999gateTeleportation} characterized the set of teleportable gates. It identified a hierarchy of teleportable gates known as the Clifford hierarchy. Focusing here on logical gates for consistency with the rest of this section, the logical Clifford hierarchy is a sequence of sets $\CC_k$ for $k = 1,2,\ldots$, where $\CC_1$ is the set of logical Pauli gates, and $\CC_k$ is defined recursively by 
\begin{align}\label{eq:Clifford_hierarchy}
  \CC_k = \{\ol{G} \colon \ol{G} \, \ol{P} \, \ol{G}^\dag \in \CC_{k-1} \text{ for all } \ol{P} \in \CC_1\}\,.  
\end{align}
That is, the $k$-th level of the Clifford hierarchy are gates that, under conjugation, transform Pauli gates into gates in the $(k-1)$-th level. We may recognize $\CC_2$ as the set of gates that transform Paulis to Paulis---that is, the set of Clifford gates. The $\ol{T}$ gate lies in $\CC_3$ because $\ol{T}\,\ol{X}\,\ol{T}^\dag = \e^{-\iUnit \pi/4}\ol{S}\,\ol{X}$ is Clifford, $\ol{T}\,\ol{Y}\,\ol{T}^\dag = \e^{-\iUnit \pi/4}\ol{S}\,\ol{Y}$ is Clifford, and $\ol{T}\,\ol{Z}\,\ol{T}^\dag = \ol{Z}$ is Pauli (and therefore Clifford).

Focusing here on single-qubit diagonal gates, if a gate $\ol{G}$ lies in $\CC_k$, then the teleportation procedure calls to use the state $\ol{G}\ket{\ol{+}}$ as a resource state. The conditional correction $\ol{G} \, \ol{X} \, \ol{G}^\dag\,  \ol{X}$ is also diagonal and lies in $\CC_{k-1}$. As pointed out already in \refcite{gottesman1999gateTeleportation}, this immediately yields a recursive procedure for implementing any gate in $\CC_k$: prepare $\ol{G}\ket{\ol{+}}$; teleport; if outcome 1 is obtained, classically compute the required correction $\ol{G}' = \ol{G}\,\ol{X}\,\ol{G}^\dag \, \ol{X} \in \CC_{k-1}$; prepare $\ol{G}' \ket{\ol{+}}$; teleport; if outcome 1 is obtained, compute the required correction $\ol{G}'' = \ol{G}' \, \ol{X} \, \ol{G}'^\dag \, \ol{X} \in \CC_{k-2}$, etc. Each correction is one level lower in the hierarchy than the last. After enough rounds, no further correction will be required, as in \cref{eq:T_repeated_teleportation}.

\subsection{Teleporting the QRAM gate}\label{sec:teleporting_the_QRAM_gate}

The teleportation strategy for single-qubit diagonal gates can also be applied to multi-qubit diagonal gates, such as the QRAM unitary $\ol{V(f)}$ from \cref{eq:QRAM_intro}. We define \textit{QRAM resource states} analogously to the resource state $\ket{T}$~(cf.~\cref{eq:T_magic_state}).
\begin{align}\label{eq:QRAM_resource_state}
   \text{QRAM resource state:} \qquad \quad \ket{\Psi(f)} = V(f) \ket{+}^{\otimes n} = \frac{1}{\sqrt{2^n}} \sum_{x\in\{0,1\}^n} (-1)^{f(x)} \ket{x}\,.
\end{align}
We denote the encoded logical QRAM resource state by $\ket{\ol{\Psi(f)}}$. 

Assuming that we can prepare the encoded resource state $\ket{\ol{\Psi(f)}}$, then we may teleport the QRAM gate into an arbitrary encoded $n$-qubit state $\sum_{x} \alpha_x \ket{\ol{x}}$ by making an entangled measurement, as in the following circuit (cf.~\cref{eq:T_teleportation_figure_uncorrected}). 
\begin{circuit}\label{eq:teleportation_circuit_large}
\tikzset{
    operator/.style={draw, filling, inner sep=2pt, thickness, align=center, baseline=(current bounding box.center)},
    internal/.style={thickness, line width=\thickWire}
}
    \scalebox{1.0}{
    \begin{quantikz}[
        row sep={0em,between origins}, 
        column sep={0em, between origins}, 
        align equals at=3.5, 
        wire types = {n,n,n,n,n,n,n,n},
        execute at end picture={
            \begin{pgfonlayer}{background}
            \draw[thick, line cap=round, line width=\thickWire] 
                (\tikzcdmatrixname-1-2) -- (\tikzcdmatrixname-1-9);
            \draw[thick, line cap=round, line width=\thickWire] 
                (\tikzcdmatrixname-2-2) -- (\tikzcdmatrixname-2-9);
            \draw[thick, line cap=round, line width=\thickWire] 
                (\tikzcdmatrixname-4-2) -- (\tikzcdmatrixname-4-9);
            \draw[thick, line cap=round, line width=\thickWire] 
                (\tikzcdmatrixname-5-2) -- (\tikzcdmatrixname-5-8);
            \draw[thick, line cap=round, line width=\thickWire] 
                (\tikzcdmatrixname-6-2) -- (\tikzcdmatrixname-6-8);
            \draw[thick, line cap=round, line width=\thickWire] 
                (\tikzcdmatrixname-8-2) -- (\tikzcdmatrixname-8-8);
            \end{pgfonlayer}
        }
    ]
    \lstick[4]{$\sum_{x} \alpha_x \ket{\ol{x}}$} 
    & 
    &[1em] 
    &[1em] \ctrl[style = {line width = \thickWire, draw=\tColor, fill=\tColor}]{4} 
    &[1.5em] 
    &[1.5em]
    &[1.5em]
    &[2em]
    &[2em]\rstick[4]{$\sum_x \alpha_x (-1)^{f(x \oplus m)} \ket{\ol{x}}$}  
    &[3em]\\[2em]
     & &  & & \ctrl[style = {line width = \thickWire, draw=\tColor, fill=\tColor}]{4} &  & &&  \\[1.5em]
         & & \gate[style={draw=none, fill=white}]{\myvdots}  & &  & \gate[style={draw=none, fill=white}]{\myddots} & \\[1.5em]
     & &   & & & &\ctrl[style = {line width = \thickWire, draw=\tColor, fill=\tColor}]{4} &  & \\[4em]
    \lstick[4]{$\ket{\ol{\Psi(f)}}$} & &  & \targ[style={line width = \thickWire, draw=\tColor}, internal/.append style={line width=\thickWire}]{} & & & &\meter[style={line width = \thickWire,draw=\tColor}]{} \wire[r][1]{c} &\rstick[4]{$m$}\\[2em]
     & & & & \targ[style={line width = \thickWire, draw=\tColor}, internal/.append style={line width=\thickWire}]{} & & &
     \meter[style={line width = \thickWire,draw=\tColor}]{}\wire[r][1]{c}  &\\[1.5em]
         & & \gate[style={draw=none, fill=white}]{\myvdots}  & &  & \gate[style={draw=none, fill=white}]{\myddots}  & \\[1.5em]
    & &[2em]  & & & & \targ[style={line width = \thickWire, draw=\tColor}, internal/.append style={line width=\thickWire}]{} & \meter[style={line width = \thickWire,draw=\tColor}]{}\wire[r][1]{c}  &
    \end{quantikz}
    }
\end{circuit}
Here we emphasize that in the context of our protocol,  the circuit is a logical circuit: all qubits are logical qubits, and both the upper and lower sets of $n$ logical qubits are constructed out of  $n'>n$ physical qubits using some QEC code. For example, one could choose to encode each logical qubit into its own $d \times d$ surface code patch, giving $n' = nd^2$ for that example.  The $n$ logical CNOT gates and $n$ logical single-qubit measurements in \cref{eq:teleportation_circuit_large} are performed fault tolerantly within this code, allowing us to neglect the chance of logical error. 

Each possible $n$-bit measurement outcome $m \in \{0,1\}^n$ is obtained with uniform probability $1/2^n$, regardless of the state $\sum_{x} \alpha_x \ket{\ol{x}}$. If $m=0^n$ is obtained, then by direct calculation (see \cref{sec:teleportation_detailed}), we can verify that the gate $\ol{V(f)}$ has been correctly applied, yielding the state $\sum_x \alpha_x (-1)^{f(x)} \ket{\ol{x}}$. However, most of the time, we obtain a nonzero measurement outcome $m \neq 0^n$, in which case the phase $(-1)^{f(x)}$ is applied onto the basis state $\ket{\ol{x \oplus m}}$ rather than $\ket{\ol{x}}$, yielding the state $\sum_x \alpha_x (-1)^{f(x\oplus m)} \ket{\ol{x}}$. Here and throughout, $\oplus$ denotes bitwise addition, modulo 2. 

To correct for this, we need to apply the phase $(-1)^{f(x\oplus m) \oplus f(x)}$ onto the basis state $\ket{\ol{x}}$ for each $x$; that is, we need to implement the correction operation $\ol{V(f')}$, 
where $f'$ is a Boolean function defined by the rule
\begin{align}\label{eq:update_rule_overview}
    f'(x) = f(x) \oplus f(x\oplus m)\,.
\end{align}
The function $f'$ depends on $m$, and thus it can only be determined after the teleportation of $\ket{\ol{\Psi(f)}}$ has been performed. To tie back to the case of a single-qubit diagonal gate $\ol{G}$ discussed in \cref{sec:teleportable_gates}, where the conditional correction was $\ol{G}' = \ol{G}\, \ol{X} \, \ol{G}^\dag \, \ol{X}$, we can note that $\ol{V(f)}^\dag = \ol{V(f)}$ and rewrite
\begin{align}
    \ol{V(f')} = \ol{V(f)} \, \ol{X}^m \, \ol{V(f)}^\dag \, \ol{X}^m\,,
\end{align} 
where $X^m$ denotes the $n$-qubit Pauli operator with Pauli-$X$ in positions where $m_i=1$ and identity operator $\Id$ in positions where $m_i = 0$, such that $\ol{X}^m \ket{\ol{x}} = \ket{\ol{x\oplus m}}$.

It now suffices to observe that for any $f$, the $n$-qubit unitary $\ol{V(f)}$ is in the $n$-th level of the logical Clifford hierarchy \cite{jaques2023qram, cui2017diagonalGatesCliffordHierarchy}; we provide a self-contained proof of this in \cref{app:QRAM_in_Clifford_hierarchy}. This guarantees that the correction $\ol{V(f')}$ will be in the $(n-1)$-th level. As explained in \cref{app:QRAM_in_Clifford_hierarchy}, the reason this holds is related to the degree of the Boolean functions $f$ and $f'$, when they are expanded as a polynomial of their $n$ input bits. Specifically, we may observe that the highest-degree monomials in the expansion of $f(x)$ are the same as those in the expansion of $f(x \oplus m)$. Thus, when $f'(x)$ is defined as $f(x) \oplus f(x\oplus m)$, the highest-degree monomials all cancel out, leaving only monomials of a lower degree.    That is, the degree of $f'$ is smaller than the degree of $f$ by at least one. We use this fact to prove in general that if a Boolean function $h$ has degree $d$, then $\ol{V(h)} \in \CC_d$. In particular, since $\ol{V(f)} \in \CC_n$ (the maximum possible degree of any function is $n$), we have that $\ol{V(f')} \in \CC_{n-1}$.

Our protocol proposes to implement the correction $\ol{V(f')}$ in the same fashion as $\ol{V(f)}$: by preparing the resource state $\ket{\ol{\Psi(f')}}$ and teleporting as in \cref{eq:teleportation_circuit_large}. This will also produce a correction, associated with a Boolean function $f''$ of degree $n-2$. As we iterate, we descend the Clifford hierarchy, and the degree of our correction function is reduced. Once we have performed $n$ rounds of teleportation, our correction function has degree zero. If a Boolean function $h$ is constant, this implies that $\ol{V(h)} \propto \ol{\Id}$; thus, once we have reduced the correction function to degree zero, we may cease iterating the protocol.  

Later, in \cref{sec:teleportation_detailed}, we perform a more complete analysis of the teleportation channel; for example, we quantify the error in the teleportation channel when an imperfect resource state is teleported instead of $\ket{\ol{\Psi(f)}}$. 

\subsection{Preparing the encoded QRAM resource state}

The analysis above shows how we can implement the logical $\ol{V(f)}$ gate, provided that we can adaptively prepare the resource states $\ket{\ol{\Psi(g)}}$, up to low error, for any particular Boolean function $g$. At first glance, this seems like a tall task. There are $2^{2^n}$ different states that we may need to prepare. By a simple counting argument, the quantum circuit complexity of at least one of these states is at least $\Omega(2^n/n)$. The innovation of our protocol is to outsource this complexity to a  single-purpose, faulty (and ideally passive) QRAM device, which may be able to exploit the unique structure of QRAM to implement $V(g)$ cheaply, but imperfectly. 

We propose a three-step procedure for preparing these states, analogous to the preparation of the $\ket{\ol{T}}$ state: physical preparation, encoding, and distillation.
\begin{itemize}
    \item \textbf{Physical preparation}: we assume that we have access to a QRAM device that can implement an approximation to $V(g)$ at the physical level, as discussed in \cref{sec:introduction} and \cref{fig:intro_figure}. By running this device on the initial input state $\ket{+}^{\otimes n}$, we produce the physical resource state $\ket{\Psi(g)}$ of \cref{eq:QRAM_resource_state}. The device can be faulty. In fact, our protocol can succeed as long as the device produces states that have at least $1/\poly(n)$ minimum fidelity with respect to $\ket{\Psi(g)}$.
    \item \textbf{Encoding}: The physical $n$-qubit state is not protected by a QEC code, and thus it is vulnerable to error. We immediately encode it into (an approximation of) the logical state $\ket{\ol{\Psi(g)}}$ using some number $n' > n$ of physical qubits on our main quantum processor. This step incurs some additional logical error because encoding arbitrary states is not fully fault tolerant. However, for topological codes like the surface code, there exist effective methods for encoding a physical qubit into a logical qubit \cite{lodyga2015simpleSchemeEncoding}.  The logical error due to encoding is $O(p)$---independent of the code distance---where $p$ is the physical error rate. In \cref{sec:encoding}, we use the general results of \refcite{christandl2024faultTolerantQuantumInputOutput} to formalize the error in this step.  Since the state $\ket{\ol{\Psi(g)}}$ is an $n$-qubit state, we expect the total logical error incurred from encoding to be $O(np)$, although for the case of general codes, we can only show $O(n\sqrt{p})$. The physical error rate must be $p = O(1/n)$ or $p = O(1/n^2)$, so that the total error from encoding remains $O(1)$, but for relevant sizes of $n$ (e.g., $n=43$ already corresponds to one terabyte of QRAM), the $p = O(1/n)$ condition is already met on devices that exist today. 
    \item \textbf{Distillation}: Distillation procedures~\cite{knill2004FTPostSelectedQC, bravyi2005UniversalQC} for the $\ket{\ol{T}}$ (or $\ket{\ol{\mathrm{CCZ}}}$) state leverage the existence of QEC codes where $T$ (or $\mathrm{CCZ}$) is transversal. The overhead, that is, the number of noisy copies of $\ket{\ol{T}}$ needed to distill one $\varepsilon_{\rm dist}$-good copy of $\ket{\ol{T}}$ is $\polylog(1/\varepsilon_{\rm dist})$, and this can be improved to $O(1)$ overhead using high-rate codes~\cite{wills2024constant, nguyen2024good, golowich2024asymptotically}. For the $V(g)$ gate, we do not know of any suitable codes that would enable this kind of approach. However, we can still distill $\ket{\ol{\Psi(g)}}$ using state-agnostic \textit{quantum purity amplification}  methods \cite{cirac1999optimalPurification,keyl2001rateOptimalPurification,fiurasek2004optimalCloningPurification,fu2016quantumStatePurification, childs2025streamingPurification, li2024optimalQuantumPurityAmplification, grier2025StreamingQStatePurification}, which take many copies of an arbitrary mixed state $\ol{\rho}$ and produce one $\varepsilon_{\rm dist}$-good copy of the pure state $\ketbra{\ol{\Xi}}$, where $\ket{\ol{\Xi}}$ is the principal component (i.e., top eigenvector) of $\ol{\rho}$. These methods do not leverage or learn any properties of $\ket{\ol{\Xi}}$, and it is known that the optimal overhead achievable in such settings is $\Theta(1/\varepsilon_{\rm dist})$ \cite{li2024optimalQuantumPurityAmplification}. In \cref{sec:distillation}, we discuss several specific state-agnostic approaches. We first consider the iterated swap test purification method studied in \refcite{fu2016quantumStatePurification,irani2022quantumSearchToDecisionReductions,childs2025streamingPurification, grier2025StreamingQStatePurification}, which is appealing for its simplicity. In the regime where the physical preparation and encoding steps prepare states with high (but still imperfect) fidelity, the iterated swap test approach is nearly optimal. On the other hand, as the fidelity of the undistilled input states decreases, the overhead of the iterated swap test rapidly increases, scaling exponentially in the inverse input fidelity. To alleviate this issue, we propose a new gate-efficient state-agnostic quantum purity amplification procedure based on quantum principal component analysis \cite{lloyd2013QPrincipalCompAnal,kimmel2016hamiltonian}, which achieves nearly optimal sample complexity even in the regime of low input fidelity, while still being compatible with the streaming model (i.e., where the undistilled input states are processed one at a time, rather than all at once as in the known sample-optimal protocol~\cite{grier2025StreamingQStatePurification}). 
    
    To apply these state-agnostic distillation approaches within our protocol, it must be the case that the state that is output by the physical preparation and encoding processes has the ideal resource state $\ket{\ol{\Psi(g)}}$ as its principal component. Evaluating this assertion requires specifying a noise model in our abstract QRAM device and in our main quantum processor. We suppose that our main processor is subject to circuit-level stochastic noise. For the QRAM device, the only assumption we make is that the noise is independent of the dataset, in the sense that, for dataset $g$, it enacts the $n$-qubit quantum channel $\CN_2 \circ \CV(g) \circ \CN_1$, where $\CN_{1,2}$ are $g$-independent noise channels, and $\CV(g) = V(g)[\cdot ]V(g)^\dag$
    is the ideal QRAM channel. In this case, we can ensure that the principal component of the state we prepare is $\ket{\ol{\Psi(g)}}$ by performing a \textit{partial Clifford-twirl} of the unitary $V(g)$. This method leverages the fact that for any Clifford circuit $C$ formed from $Z$, $X$, $\mathrm{CZ}$, and $\mathrm{CX}$ (i.e., CNOT) gates, we have $\ket{\ol{\Psi(g)}} = C\ket{\ol{\Psi(g_C)}}$ for some dataset $g_C$; the idea is to choose a random $C$, compute the dataset $g_C$, query $g_C$ with the QRAM device, and then apply $\ol{C}$ fault-tolerantly to restore $\ket{\ol{\Psi(g)}}$. Partial Clifford twirling is not necessary under the stronger assumption that the noise in the QRAM device naturally guarantees that the ideal resource state is the principal component.  
\end{itemize}
It is important that our protocol can work even when the QRAM device has low (at least inverse polynomial) fidelity. Given the engineering challenges associated with building a reliable physical QRAM device, it is much easier to imagine realizing our protocol in practice, especially as $n$ grows, if the physical QRAM device need only have a small correlation with the correct output.  Along these lines, another key benefit of a distillation--teleportation approach to fault-tolerant QRAM is that one always has the option to restart the preparation, encoding, and distillation procedure if an error is detected. For instance, if the physical QRAM device recognizes certain errors (e.g., photon loss), one can simply postselect on these events not occurring, improving the effective fidelity of the device from the perspective of our protocol.

\subsection{Full summary of protocol and statement of results}

\NewDocumentCommand{\unitCell}{O{}O{}O{}m}{
    \IfStrEq{#4}{dataWire}{%
        \wire[r][9][line width = \thickWire]{q}
        \pgfmnc[0em]         
        \pgfmnc[1.5em]         
        \pgfmnc[1.8em]         
        \pgfmnc[0em]   
        \pgfmnc[2em] \ctrl[style = {line width = \thickWire, draw=\tColor, fill=\tColor}]{3}
        \pgfmnc[2em]         
        \pgfmnc[2.5em] \slice[style={color=\sliceColor,yshift=-2em, shorten >=-1.0em, shorten <=-1.5em, line width = \thickWire},label style={xshift = -6.5em}]{round #3}         
        \pgfmnc[2.5em]         
    }{\IfStrEq{#4}{busWire}{%
        \wire[r][9][line width = \thickWire]{q}
        \pgfmnc[0em]         
        \pgfmnc[1.5em]         
        \pgfmnc[1.8em]         
        \pgfmnc[0em]   
        \pgfmnc[4.9em] \ctrl[style = {line width = \thickWire, draw=\tColor, fill=\tColor}]{3}
        \pgfmnc[2.5em]         
        \pgfmnc[2.5em]         
        \pgfmnc[0em]         
    }{\IfStrEq{#4}{owireTop}{%
        \pgfmnc \oGateRed \wire[r][1][draw=red,thick]{q}\wire[d][5]{c} 
        \pgfmnc \eGateRed \wire[r][2]{q}
        \pgfmnc \dGate
        \pgfmnc 
        \pgfmnc 
        \pgfmnc
        \pgfmnc
        \pgfmnc
    }{\IfStrEq{#4}{owireEmpty}{%
        \pgfmnc \oGateRed \wire[r][1][draw=red,thick]{q} 
        \pgfmnc \eGateRed \wire[r][2]{q}
        \pgfmnc 
        \pgfmnc 
        \pgfmnc 
        \pgfmnc
        \pgfmnc
        \pgfmnc
    }{\IfStrEq{#4}{resourceWire}{%
        \pgfmnc \oGateRed \wire[r][1][draw=red,thick]{q} 
        \pgfmnc \eGateRed \wire[r][2]{q}
        \pgfmnc 
        \pgfmnc \wire[r][2][line width = \thickWire,"{\mbox{$#2$}}"{above,pos=0.4}]{q}
        \pgfmnc \targ[style={line width = \thickWire, draw=\tColor}, internal/.append style={line width=\thickWire}]{}
        \pgfmnc \meter[style={line width = \thickWire,draw=\tColor}]{} \wire[r][1]["{\mbox{$#1$}}"{above,pos=1.0}]{c}
        \pgfmnc  \phase{}
        \pgfmnc 
    }{\IfStrEq{#4}{vdotsWire}{%
        \pgfmnc \gate[style={draw=none, fill=white}]{\myvdots} 
        \pgfmnc \gate[style={draw=none, fill=white}]{\myvdots}
        \pgfmnc
        \pgfmnc 
        \pgfmnc 
        \pgfmnc
        \pgfmnc
        \pgfmnc
    }{\IfStrEq{#4}{classicalWire}{%
    \wire[r][8]{c}
        \pgfmnc \phase{}
        \pgfmnc
        \pgfmnc
        \pgfmnc
        \pgfmnc
        \pgfmnc
        \pgfmnc \urGate \wire[u][3]{c}
        \pgfmnc 
    }{}
    }
    }
    }
    }
    }
    }
}

\begin{figure*}[ht!]
\centering
\tikzset{
    operator/.style={draw, filling, inner sep=2pt, thickness, align=center, baseline=(current bounding box.center)},
    internal/.style={thickness, line width=\thickWire}
}
\scalebox{0.86}{
    \begin{quantikz}[row sep={0em,between origins}, column sep={0em, between origins}, align equals at=4, wire types = {n,n,n,n,n,n,n,n}]
   \lstick{$\ket{\ol{\alpha}}$} &[1.5em] \qwbundle{n} &[1em] \slice[style={color=\sliceColor,yshift=-2em, shorten >=-1.0em, shorten <=-1.5em, line width = \thickWire}]{} &[4em] \wire[l][3][line width = \thickWire]{q} \unitCell[][][1]{dataWire} &[0.8em] \unitCell[][][2]{dataWire} &[1.3em]&[0em] \gate[style={draw=none, fill=white}]{\myhdots} \slice[style={color=\sliceColor,yshift=-2em, shorten >=-1.0em, shorten <=-1.5em, line width = \thickWire}]{}&[2.8em] &[0.8em]\wire[l][2][line width = \thickWire]{q} \unitCell[][][$n$]{dataWire} &[1em] &[0em] \wire[l][1][line width = \thickWire]{q}\rstick{$\ol{V(f)}\ket{\ol{\alpha}}$}\\[2em]
   & & & \unitCell{owireTop} & \unitCell{owireTop}  & & & & \unitCell{owireTop} & &\\[1.8em]
   & & & \unitCell{owireEmpty} & \unitCell{owireEmpty} & & & & \unitCell{owireEmpty} & &\\[1.8em]
    & & & \unitCell{resourceWire} & \unitCell{resourceWire} & \gate[style={draw=none, fill=white}]{\myhdots} & & & \unitCell{resourceWire} & &  \\[1.8em]
   & & & \unitCell{vdotsWire} & \unitCell{vdotsWire}  & & & & \unitCell{vdotsWire}  & &\\[1.8em]
   & & & \unitCell{owireEmpty} & \unitCell{owireEmpty} & & & & \unitCell{owireEmpty} &  &\\[2em]
    \lstick{$f$}& \qwbundle{2^n} & & \wire[l][3]{c}\unitCell{classicalWire} & \wire[l][1]{c} \unitCell{classicalWire}  & \wire[l][1]{c} & \gate[style={draw=none, fill=white}]{\myhdots} & & \wire[l][2]{c}\unitCell{classicalWire}  & & \wire[l][2]{c} 
     \end{quantikz}
     }
     \caption{\label{fig:protocol_overview} Quantum circuit depiction of the protocol for implementing the logical diagonal QRAM operation $\ol{V(f)}$ (see \cref{eq:QRAM_intro}) fault tolerantly for a data table $f$. The protocol cycles through $n$ rounds, and each round has five steps: preparation, encoding, distillation, teleportation, and classical update, depicted in different colors. All gates are fault-tolerant, logical gates, except the query to the noisy QRAM device $\Psi$ and the encoding step $E$, outlined in red. The solid black wires represent encoded logical quantum registers of $n$ logical qubits, the red wires represent unencoded quantum registers of $n$ physical qubits, and the double black lines represent classical registers. For simplicity, we have not depicted the twirling step in the figure (which may not be necessary in practice), where prior to each application of $\Psi$, the dataset is modified by an independently chosen random transformation, which is corrected for after $E$ has been applied; see \cref{eq:twirl_figure}. }
\end{figure*}
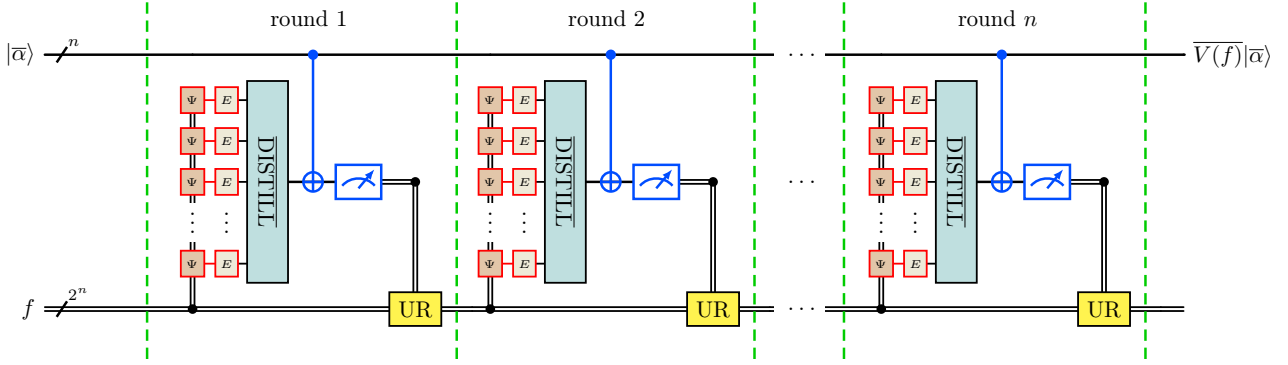

To summarize, our main result is a protocol for implementing the logical QRAM operation $\ol{V(f)}$, up to arbitrarily high fidelity, using many queries to a device that can perform the physical QRAM operation $V(g)$ (for any/all $g$) with a lower nonzero fidelity. As discussed in \cref{sec:introduction}, the physical QRAM operation could be accomplished with a dedicated subcomponent of the larger quantum device specialized for QRAM, which need not be capable of universal fault-tolerant quantum computation.  

The protocol to implement $\ol{V(f)}$ cycles at most $n$ times through five steps discussed in the previous subsections: (i) physical preparation, (ii) encoding, (iii) distillation, (iv) teleportation, and (v) adaptive classical computation of the correction. Step (v) uses measurement outcomes from step (iv) to transform the dataset according to the \textit{update rule} (UR) of \cref{eq:update_rule_overview}, prior to returning to step (i). The entire protocol is depicted in \cref{fig:protocol_overview}, where each of the five steps is shown in a different color. A more detailed specification and formal error analysis of each step is provided in \cref{sec:protocol_detailed}. We arrive at the following statement of the cost of implementing $\ol{V(f)}$. 

\begin{theorem}[Main result (informal)]
For any data table $f$ with $2^n$ entries, and any error parameter $\varepsilon>0$, the protocol performs the logical QRAM operation $\ol{V(f)}$ up to error $\varepsilon$ (in diamond distance), under the assumption that the physical QRAM device implementing physical $V(f)$ has noise independent of $f$. The quantum resources required are:
\begin{itemize}
    \item $\poly(n)/\varepsilon$ calls to a device that performs the physical $V(g)$, for various $g$ (determined adaptively) with any nonzero minimum fidelity $F \geq 1/\poly(n)$. 
    \item $\poly(n)/\varepsilon$ calls to a $\poly(n)$-cost fault-tolerant encoding procedure that encodes $n$-qubit physical states into a suitable QEC code capable of FTQC, while incurring at most $O(1)$ logical error.
    \item $\poly(n)/\varepsilon$ fault-tolerant one- and two-qubit logical gates, single-qubit logical $\ket{\ol{0}}$ state preparations, and single-qubit logical measurements.
\end{itemize}
The classical resources required are:
\begin{itemize}
    \item $n$ applications of the classical update rule, each of which has $O(2^n)$ complexity in a standard RAM model.
    \item $\poly(n)/\varepsilon$ twirling operations on the dataset, each of which has complexity $\poly(n)2^n$ in a standard RAM model.
\end{itemize}
After each update rule and twirling operation, the physical QRAM device must be ``reloaded'' or otherwise given access to the updated classical dataset. 
\end{theorem}

The main implication of this result is the following: suppose that a quantum algorithm calls the QRAM operation $V(f)$ at most $T = \poly(n)$ times, and suppose that one has access to a QRAM device that approximately performs the physical QRAM operation with at least $1/\poly(n)$ fidelity, at computational cost $\poly(n)$ (similar to the cost of RAM). Then, one may take $1/\varepsilon = O(T) = \poly(n)$, and conclude that the algorithm can be implemented fault-tolerantly using only $\poly(n)$ quantum resources.  As a result, our protocol provides a step toward justifying the cheap QRAM assumption, and it provides a method of fault-tolerantly implementing quantum algorithms that depend on QRAM. 

Even in a cost model where noisy physical QRAM incurs computational cost $\Omega(2^n)$---for instance, if the physical QRAM has $\Omega(2^n)$ active gates each requiring $\Omega(1)$ energy input---our protocol still provides the benefit that the exponential quantum complexity is contained entirely to physical quantum operations that can be optimized specifically to perform QRAM. There is no need for an exponential amount of QEC and the associated overheads it incurs. 

\subsubsection{Caveat: classical complexity}
The main caveat of our protocol is that it requires a non-negligible amount of purely classical adaptive computation. In particular, after receiving random measurement outcome $m$, the protocol requires replacing the value $g(x)$ with the value $g(x) \oplus g(x \oplus m)$ for all $2^n$ addresses $x$ of the dataset, as in \cref{eq:update_rule_overview}. While computing the new value is easy for any individual $x$, the sheer number of different $x$ means the complexity---in terms of classical circuit size or RAM calls to the dataset $g$---is at least $\Omega(2^n)$. 

However, we must recall that naively, the logical QRAM requires $\Omega(2^n)$ fault-tolerant quantum resources, if implemented as a fault-tolerant circuit. One unit of fault-tolerant quantum resources, such as one fault-tolerant quantum gate, is expected to be several orders of magnitude more expensive in terms of both financial cost and computational runtime than one unit of classical computation, such as a classical gate or floating point operation \cite{babbush2021FocusBeyondQuadratic}. Thus, trading $\Omega(2^n)$ quantum for $\Omega(2^n)$ classical resources may lead to an overall cheaper and faster computation. 

Furthermore, we expect that although it formally has $\Omega(2^n)$ computational cost, the complexity of the classical update rule has significantly better constant prefactors than the classical computation required to power active QEC of an entire QRAM circuit.  As mentioned previously, circuit QRAM with $\poly(n)$ latency would require a fault-tolerant quantum computer with $\Omega(2^n)$ logical qubits. Such a device would likely require $\Omega(2^n)$ full-fledged \textit{classical} chips to be co-located with the logical qubits, in order to process in parallel the QEC syndrome data generated by the computation in real time. For example, in the surface code operating at a 1 MHz QEC cycle rate, the amount of syndrome data generated by $2^{20}$ logical qubits, each encoded into its own patch at code distance 11, would be more than 15 terabytes per second. Specialized classical decoding algorithms must be run continuously to identify and correct errors as they occur.   In contrast, for a datset of size $2^{20}$ bits, the classical update rule in our protocol is a single structured transformation of a 120 kilobyte dataset. The only interaction between this dataset and the quantum processor is the reloading of the physical QRAM device with the updated dataset.  

Additionally, because of the structure of the classical update rule $g(x) \mapsto g(x) \oplus g(x \oplus m)$, it is conceivable that dedicated classical chips could be built to parallelize the process of performing the update and the reloading of the QRAM device. In \cref{sec:complexity_of_classical_update}, we analyze the complexity of the update rule, and illustrate how it can be implemented with a classical circuit of depth $\poly(n)$, although embedding this circuit into 2 or 3 spatial dimensions leads to asymptotically growing wire density. We show how in a model of parallel computation, the update rule is equivalent (up to $\poly(n)$ factors) to performing sparse matrix-vector multiplication, with a particularly close connection to the (fast) Walsh--Hadamard transform.  This fact helps to understand the expected difficulty of parallelization and it clarifies the opportunity cost of these classical resources. 

\subsubsection{Comments on scalability}

An additional caveat is the fact that our protocol is likely not to be fully scalable for indefinitely large $n$. This stems from two aspects, the physical QRAM device, and the encoding step. 

First,  the physical size of an (ideally passive) physical QRAM device would need to grow as $\Omega(2^n)$, yet we need it to produce physical QRAM resource states with at least $1/\poly(n)$ fidelity. Thus, the fidelity of the individual device components needs to improve as $n$ grows. It is known that certain architectural approaches to QRAM, namely, bucket brigade QRAM, possess a certain noise resilience property: the overall infidelity of the physical QRAM operation scales as $O(qn^2)$ \cite{hann2021resilienceofQRAM}, where $q$ is the error rate of the individual router components that compose the device. This is exponentially better than $O(q2^n)$, which would be the naive expectation, given the exponential number of error-prone routers in the device. This noise resilience is a crucial fact for the possibility of practical QRAM. If this kind of scaling is achieved, then the physical per-component error rate $q$ must decrease asymptotically roughly as $O(1/n^2)$ to be useful for our protocol. 

Second, the physical QRAM resource state is an $n$-qubit state, and in a noise model where each operation on our main quantum processor fails with probability $p$, the encoding of this $n$-qubit physical state into an $n$-qubit logical state necessarily incurs at least $\Omega(np)$ logical error. The formal analysis, later, shows how $O(n\sqrt{p})$ can be achieved regardless of the choice of QEC code. Either way, $p$ must decrease as $O(1/n)$ or as $O(1/n^2)$ to keep this error of total size $O(1)$. 

While any practical implementation of this protocol will certainly need to pay close attention to error rates at every step, this is not a hugely debilitating conceptual issue for QRAM. This is because we do not ever expect to need to build a QRAM device for very large values of $n$. For example, typical RAM devices in classical computers are of size roughly 10 gigabytes, corresponding to only $n=36$. The physical error rate in state-of-the-art quantum devices in several different platforms already achieves $p$ in the range of $10^{-3}$--$10^{-2}$. Improving by roughly an order of magnitude to $p=q=10^{-4}$ would be sufficient to enable our protocol at size $n=36$, assuming that the dominant error contribution scales as $4qn^2$ (where the presumed constant prefactor of 4 is chosen to align with \refcite[Eq.~(28)]{hann2021resilienceofQRAM}).

\subsubsection{Comments on applications}

Our investigation has been primarily motivated by the goal of evaluating the viability of QRAM as a primitive for fault-tolerant quantum computation in an abstract sense. Nonetheless, in \cref{sec:applications}, we consider whether our protocol could provide a practical advantage over alternative methods in several concrete applications. Generally, although our protocol achieves asymptotically polynomial $\poly(n)/\varepsilon$ complexity, we find that this version of the protocol struggles to provide an immediate advantage. For example, in quantum machine learning scenarios, the $\Omega(2^n)$ cost of the classical update rule makes it difficult to find examples where an end-to-end speedup persists over alternative classical methods. The observation in \cref{sec:complexity_of_classical_update} that the update rule is similar in power to a sparse matrix-vector multiplication clarifies that, in our search for superpolynomial quantum speedups, we must only target problems where the ability to perform classical $2^n \times 2^n$ sparse matrix-vector multiplications is not already sufficient to solve the problem in $\poly(n)$ time,  which considerably reduces the set of candidates. See \cref{sec:applications_ML} for comments on possible scenarios where this conclusion may be avoided. 

On the other hand, in scenarios like quantum chemistry and cryptanalysis where QRAM is utilized---in that context often referred to as a ``quantum lookup table'' or ``quantum read-only memory''---the $\Omega(2^n)$ classical cost is tolerable. In fact, in these instances, it is typically already being proposed to implement QRAM with a fully error-corrected quantum circuit  of depth $\Omega(2^n)$ \cite{babbush2018EncodingElectronicSpectraLinearT}. Our protocol could allow this exponential fault-tolerant complexity to be offloaded to a specialized physical QRAM device and a classical computer. However, in our preliminary resource analysis at relevant system sizes in \cref{sec:applications}, the $\poly(n)/\varepsilon$ cost is still too large to provide an actual advantage. Part of the issue is that if the QRAM operation is called $T$ times, one must take $1/\varepsilon = \Omega(T)$, and hence the total cost of implementing all $T$ QRAM queries scales as $T^2$. The discovery of a distillation protocol for QRAM resource states with overhead $\polylog(1/\varepsilon)$ instead of $1/\varepsilon$ would be extremely beneficial in this calculation.

\subsubsection{Extension to multiple output bits}

In \cref{app:extension_to_b_bits}, we explain how our protocol can be straightforwardly extended to the case where $b$ classical bits are stored at each of $2^n$ addresses, and one wishes to coherently read all $b$ bits into a separate bus register. That is, we show how to fault-tolerantly perform the operation $\ol{U(f)}$ that implements $\ket{\ol{x}}\ket{\ol{u}} \mapsto \ket{\ol{x}} \ket{\ol{u \oplus f(x)}}$, with $f \colon \{0,1\}^n \rightarrow \{0,1\}^b$ here denoting a function with $b$ output bits. The strategy is to observe that conjugating $U(f)$ by a Hadamard transform on the bus register yields a diagonal unitary with $\pm 1$ on the diagonal that may be viewed as a generalization of $V(f)$ from \cref{eq:QRAM_intro}. The unitary acts on $n+b$ qubits rather than $n$ qubits, but importantly, the degree of the Boolean function is at most $n+1$, which can be much smaller than $n+b$. The protocol proceeds identically to how it is described in the main text, except that the resource states are larger, requiring $n+b$ qubits, which leads to greater gate complexity overhead when performing distillation and teleportation.

\subsection{Relation to prior work}

The idea of QRAM was first formalized by Giovannetti, Lloyd, and Maccone (GLM) in~\refcite{giovannetti2007QuantumRAM, giovannetti2008architectures} (although some primitive versions of QRAM had been sketched earlier, see e.g.~\refcite[Chapter 6]{nielsen2002QCQI}). These works first introduced the idea of a dedicated QRAM hardware element---a device specially designed for QRAM and separate from the main quantum processor---by proposing implementations based on optical and atomic hardware. Many other proposals have followed, including proposals based on superconducting circuits~\cite{weiss2024quantum,hann2019hardware,wang2024quantum,SalaCadellans2015transmon}, photonic systems~\cite{chen2021scalable,hong2012robust}, and neutral atom arrays~\cite{cesa2025fast} (see~\refcite{jaques2023qram} for a more detailed review). We highlight that notions of teleportation-based QRAM~\cite{chen2021scalable} and QRAM resource states ~\cite{cesa2025fast}---albeit resource states of size exponential in $n$---have previously been proposed. Unfortunately, all of these proposed QRAM implementations face daunting practical challenges, and most are not passive, meaning they require active control over $\Omega(2^n)$ quantum components, which would undermine the cheap QRAM assumption. (As \refcite{jaques2023qram} notes, the proposal of \refcite{SalaCadellans2015transmon} is a noteworthy example of a passive implementation, although it faces challenges of scalability and practicality.) To our knowledge, there has not yet been a proposed implementation of a physical QRAM device that is simultaneously practical, scalable, and passive.

The initial GLM QRAM papers also sparked a long-running debate about the practicality of QRAM and validity of the cheap QRAM assumption, especially in relation to error correction and fault tolerance. In particular, GLM proposed a specific QRAM architecture---the “bucket brigade” QRAM---that they argued was intrinsically robust to errors. This claim was initially met with some skepticism (see, e.g.,~\refcite{arunachalam2015RobustnessBuckBrigQRAM}), but the robustness was later proven in~\refcite{hann2021resilienceofQRAM}, which showed that the overall error of a bucket-brigade QRAM query scaled only with $\poly(n)$, despite the fact that the QRAM itself is comprised of $\Omega(2^n)$ error-prone components. This robustness is key to the viability of our own proposal, since the passive QRAM device in \cref{fig:passive_QRAM} could indeed have only moderate overall error rates compatible with our distillation--teleportation scheme.

Even with some intrinsic robustness against errors, QRAM is still likely to require QEC in most applications. As mentioned in \cref{sec:introduction}, fault-tolerant implementations of QRAM based on circuit decompositions of the unitary $V(f)$ involve $\Omega(2^n)$ qubits and $\Omega(\sqrt{2^n})$ non-Clifford gates, and face serious questions of practicality at large-scales. 
To our knowledge, the survey of QRAM in \refcite{jaques2023qram} was the first to consider the possibility of achieving a fault-tolerant QRAM by a method other than circuit QRAM. They proved several no-go theorems that present barriers to finding a code where QRAM is transversal. They also proved a no-go theorem ruling out certain distillation--teleportation protocols. Specifically, they considered protocols that have a distillation phase that queries the physical QRAM gate $V(f)$ up to $Q$ times to prepare a resource state $\chi(f)$, followed by a teleportation channel where $\chi(f)$  interacts with an arbitrary state $\ketbra{\ol{\alpha}}$ in an attempt to prepare $\ol{V(f)} \ketbra{\ol{\alpha}}\ol{V(f)}$. They showed that for this setup, $Q \geq \Omega(2^{2n})$ queries are  required to succeed with high fidelity on all possible choices of $\ketbra{\ol{\alpha}}$. Translating their logic into our language, they observed that regardless of the protocol and the function $f$, one can always find an $f'$ which differs from $f$ at only a few addresses, but where the resource states $\chi(f)$ and $\chi(f')$ are $O(\sqrt{Q}/2^n)$-close. This implies that if $Q \ll 2^{2n}$, then the teleportation channel cannot produce well-distinguishable outputs when $f$ is queried compared to when $f'$ is queried (data processing inequality). Yet, if we suppose that $f$ and $f'$ differ at even one address $j$ while agreeing on some other address $k$, then when $\ket{\ol{\alpha}} = \frac{1}{\sqrt{2}}(\ket{\ol{j}} + \ket{\ol{k}})$, the $\ol{V(f)}$ and $\ol{V(f')}$ gates should lead to distinguishable orthogonal states, a contradiction.  

Each of the $n$ rounds within our protocol individually fits into the framework of the no-go theorem of \refcite{jaques2023qram}. Our protocol circumvents this result because it adaptively updates the QRAM function being queried in each round, based on measurement outcomes obtained in prior rounds. If two functions $f$ and $f'$ are different, even at a single address, there will be at least one round where the resource states being teleported by our protocol are far away from each other; in fact, if $f$ and $f'$ differ at exactly one address in round $r=1$, then they will differ at exactly $2^{r-1}$ addresses in each round $r=2,3,4,\ldots,n$, provided that all of the $n$-bit random measurement outcomes obtained up until round $r$ form a linearly independent set.

Certain elements of our protocol also connect with prior work outside the context of QRAM. For example, the task of quantum purity amplification has been extensively studied, and we comment more on this in \cref{sec:distillation}. Additionally, while the $n$-qubit states $\ket{\Psi(f)}$ from \cref{eq:QRAM_resource_state} have---to the best of our knowledge---not previously been proposed as QRAM resource states, they have been utilized in other contexts, often by the name of ``phase states.'' For example, phase states have been studied as pseudorandom quantum states in the context of cryptography \cite{ji2018pseudorandomQuantumStates,brakerski2019pseudoRandomBinaryPhase}, as targets for quantum state tomography \cite{arunachalam2023learningPhaseStates}, and as a mechanism for showing search-to-decision reductions in quantum complexity theory \cite{irani2022quantumSearchToDecisionReductions}.

\section{Error models and physical protocol requirements}

The theory of FTQC shows how a quantum processor can perform an arbitrary quantum computation through a sequence of noisy physical operations on a set of physical qubits, provided that the physical noise is sufficiently uncorrelated and its rate $p$ is below a constant threshold~\cite{aharonov1997FTQCconstantError, shor1996FTQC, gottesman2014FTQCconstantOverhead}. 
Our protocol augments this by assuming that we also have access to a physical QRAM device that can perform an approximate $V(g)$ gate on $n$ physical qubits, for any function $g$, as illustrated in \cref{fig:setup}. We require that $g$ can be modified from one query to the next via classical communication with the device. Also, we require that the quantum information contained in the $n$ physical qubits output by the device can be transported into $n$ physical qubits on the main quantum processor without significant degradation of its fidelity, whether by physically moving the qubits output by the device to the main processor, or by some other means.

In this section, we specify the noise models we consider for each of these two components, and we define the setup for the FTQC part of our protocol on the main processor. The protocol is agnostic to many of the details here, including which QEC code is used, and we attempt to keep it as general as possible. 

\subsection{Error model of the physical QRAM device}

Without a more concrete implementation in mind, we cannot fully model the noise in the device. However, we consider a general noise model that assumes only that the noise is independent of $g$, the function being queried. This noise model was also employed for some of the results in \refcite{jaques2023qram}.  
\begin{definition}[Dataset-independent QRAM noise]\label{def:dataset-independent_noise}
A physical QRAM device that implements channel $\tCV(g)$ on input $g$ is said to have dataset-independent noise if there exists $\CN_1$ and $\CN_2$ independent of $g$ for which
\begin{align}
    \tCV(g) = \CN_2 \circ \CV(g) \circ \CN_1\,,
\end{align}
where $\CV(g) = V(g) [\cdot] V(g)^\dag$ is the ideal unitary channel. 
\end{definition}

We defend the plausibility of this noise model by appealing to the presumed structure of shallow-depth physical QRAM implementations. One imagines that the bits of classical memory corresponding to the values $g(0)$, $g(1)$,\ldots, $g(2^n-1)$  are distributed in memory cells over 1D or 2D space---for example, the illustration in \cref{fig:passive_QRAM} distributes them in 1D space. As discussed in \refcite{jaques2023qram}, at a high level, a $\poly(n)$-depth QRAM implementation requires a routing step, a readout step, and an unrouting step. In the routing step, the $n$-qubit address information $\ket{x}$ is used to (coherently) activate a path to the memory cell corresponding to address $x$. In the readout step, a qubit must interact with the classical bit of information $g(x)$ stored in that cell, gaining a $-1$ phase if and only if $g(x)=1$. Then, in the unrouting step, the activated routers must be coherently reset before being traced out to ensure the overall operation maintains coherence between different $\ket{x}$. 

The important fact to notice is that the routing and unrouting steps are not at all dependent on the dataset $g$. Thus, to justify the validity of \cref{def:dataset-independent_noise} in this model, any noise that occurs during the routing step could be propagated backward to the beginning of the circuit and contribute to $\CN_1$, while any noise in the unrouting step could be propagated forward to the end of the circuit and contribute to $\CN_2$. 

The only step that can be dependent on $g$ is the readout step,
but this step is generally considered to be simpler to implement than the routing step \cite{jaques2023qram}. For example, in the bucket-brigade QRAM circuit of \refcite[Figure 10]{hann2021resilienceofQRAM}, the readout step is performed in  a single circuit layer by a set of parallel single-qubit Pauli-$X$ gates: at memory cell $x$, an $X$ gate is applied if $g(x) = 1$, and an identity gate $\Id$ is applied if $g(x) = 0$. (This classically controlled $X$ gate would ideally be applied passively via interaction with a non-volatile memory storing $g(x)$.) Since the identity $\Id$ and Pauli-$X$ gates are simple, it may be plausible that, in some implementations, the noise can be independent of which of them is applied, justifying \cref{def:dataset-independent_noise} for the physical QRAM device. Furthermore, we note that even if the noise is not identical for the $\Id$ and $X$ gates, it is plausible that the dataset-independent noise property \textit{effectively} holds in the context of our protocol, thanks to the protocol's \textit{partial Clifford twirling} (\cref{sec:partial_Clifford_twirl}) that effectively randomizes the data being queried---intuitively this randomization should remove dependence of $g$ from the average noise channel.

We leave to future work the task of more rigorously showing that certain microscopic (i.e., component-level) noise models lead to dataset-independent noise in the form of \cref{def:dataset-independent_noise}. However, we note that the set of noise processes that fall into this category can include some counterintuitive members. For example, we may suppose that a physical QRAM device is constructed from a depth-$n$ binary tree of routing elements, but that one of the routers in the tree is ``dead.'' The QRAM attempts to activate a path through the tree to a particular address $x$ at one of the leaves, and if this path passes through the dead router, it causes catastrophic failure of the device, leading the device to instead output the maximally mixed state $\Id/2^n$. Let $\CX \subset \{0,1\}^n$ denote the set of addresses that cause the dead router to activate when they are queried, and let $\Pi_{\CX} = \sum_{x \in \CX} \ketbra{x}$ be the projector onto these address states. Then, we may write the channel implemented by the noisy device as
\begin{align}
    \tCV(g)[\rho] &= V(g)(\Id - \Pi_{\CX})\rho(\Id - \Pi_{\CX}) V(g)^{\dag} + \tr(\Pi_{\CX}\rho) \frac{\Id}{2^n} \\
    &= (\Id - \Pi_{\CX})V(g)\rho V(g)^{\dag}(\Id - \Pi_{\CX})  + \tr(\Pi_{\CX}\rho) \frac{\Id}{2^n} \,,
\end{align}
where the second equality follows since $V(g)$ is a diagonal unitary, and thus it commutes with the diagonal projector $\Id - \Pi_{\CX}$. We may then rewrite the noisy channel $\tCV(g)$ in a dataset-independent fashion as $
    \tCV(g) = \CN_2 \circ \CV(g) \circ \CN_1$, with $\CN_1 = \CI$ (the identity channel) and 
\begin{align}
    \CN_2[\rho] =  (\Id - \Pi_{\CX})\rho(\Id - \Pi_{\CX}) + \tr(\Pi_{\CX} \rho) \frac{\Id}{2^n}\,.
\end{align}

Furthermore, in the context of our protocol, the device is always queried on the input state $\rho = \ketbra{+}^{\otimes n}$. Thus, this particular noise process has fidelity given by
\begin{align}    \tr(\ketbra{\Psi(g)}\tCV(g)[\ketbra{+}^{\otimes n}] ) = \frac{(2^n-|\CX|)(2^n-1-|\CX|)}{2^{2n}} + \frac{1}{2^n} \geq 1-\frac{2|\CX|}{2^n}\,.
\end{align}
That is, if the dead router is deep in the binary tree and $|\CX|\ll 2^n$, then the fidelity of the device remains close to 1. 

The fact that our protocol can work in such a scenario is counterintuitive because  the dead router would seem to completely block access to the information $g(x)$ for any $x \in \CX$ and thus make it impossible to correctly implement the QRAM when the address register has high overlap with the support of $\Pi_{\CX}$. As we will show in \cref{sec:partial_Clifford_twirl}, our protocol handles this with partial Clifford twirling, a technique that scrambles the dataset $g$ into a new dataset $g_C$ so that the information $g(x)$ is contained in $g_C(y)$, and the location $y$ is uniformly random, and in particular, the probability that $y \in \CX$ is $|\CX|/2^n$. Thus, for every $x$, the dead router only compromises the information $g(x)$ with small probability, even for $x \in \CX$. 

\subsection{Error model of the main quantum processor}

Our main quantum processor acts on a set of noisy physical qubits with a quantum circuit, that is, a sequence of noisy physical quantum operations including initializations, 1- and 2-qubit gates, and measurements, as well as classical computation and adaptive classical feedback. We will assume that our main processor is subject to circuit-level stochastic noise, defined below. 

\begin{definition}[Circuit-level stochastic noise, Section 2.5 of~\refcite{christandl2024faultTolerantQuantumInputOutput}]
 \label{def:stochastic-noise}
    A physical implementation of a quantum circuit $V$ is said to be subject to parameter-$p$ stochastic noise if the following holds.
    \begin{itemize}
        \item Purely classical components are implemented perfectly without any errors.
        \item Each quantum component $\CP$, including (classically controlled-)gates, qubit initializations, measurements, is realized by $\tCP = (1-p) \CP + p \mathcal{N}_{\CP}$, where $\mathcal{N}_{\CP}$ is a quantum channel of the same input and output registers as $\CP$.
    \end{itemize}
\end{definition}
Circuit-level stochastic noise  is a special case of the more standard model called local-stochastic noise model~\cite{gottesman2014FTQCconstantOverhead}, which additionally allows some correlations between the gate faults. We choose circuit-level stochastic noise over local stochastic noise because this allow us to analyze a fault-tolerant logical state preparation procedure using the results of \refcite{christandl2024faultTolerantQuantumInputOutput}, for which the theorems require independent noise. However, we will require that the QEC code and FTQC scheme are able to correct against the more general local stochastic noise. 

\subsection{Fault-tolerant quantum computation}\label{sec:FTQC}

Our protocol is agnostic to which QEC code family and FTQC scheme is utilized, as long as it is capable of a universal set of fault-tolerant logical gates. Specifically, there must be a nonzero threshold $p_0$ such that, if the processor is subject to local stochastic noise with error parameter $p < p_0$, then for any target error rate and any logical quantum circuit, one can choose the code parameters (e.g., distance) large enough to ensure the ideal logical circuit is simulated up to the required error~\cite{aharonov1997FTQCconstantError, gottesman2010introduction, gottesman2014FTQCconstantOverhead,fawzi2018ConstantOverheadExpanderCode, yamasaki2024time, nguyen2024quantum}.

As we wish to enact the logical $n$-qubit QRAM operation of \cref{eq:QRAM_intro}, we assume we have chosen some QEC code family that encodes $n$ logical qubits into some number $n' > n$ of physical qubits, along with a scheme for performing fault-tolerant gates. 
Regardless of our choice, we can identify the following ingredients. 
\begin{itemize}
    \item Encoding: There is an encoding isometry $E$, which maps any $n$-qubit physical state $\ket{\psi}$ to its associated encoded $n$-qubit logical state $\ket{\ol{\psi}}$ (of $n'$ physical qubits). The map $E$ is injective and its image is the codespace of the code. We let $\CE$ denote the corresponding quantum channel $E[\cdot]E^\dagger$ that encodes density matrices. Our protocol will also require a fault-tolerant encoding gadget $\CE_{\rm FT}$, which implements $\CE$ in the absence of noise, but is constructed in such a way that its noisy implementation $\tCE_{\rm FT}$ is resilient to errors (see \cref{sec:encoding}). 
    \item QEC: There is a QEC projector $\CQ$ that maps states outside of the codespace to states in the codespace, that is, a map that detects and corrects physical errors on encoded states. In FTQC, the projector $\CQ$ is implemented with a fault-tolerant QEC gadget, denoted $\CQ_{\rm FT}$, which enacts the map $\CQ$ in the absence of noise and a map $\tCQ_{\rm FT}$ in the presence of noise. The gadget $\CQ_{\rm FT}$ is applied after each location in the logical circuit to prevent the buildup and propagation of physical errors; it must satisfy certain formal properties to guarantee the existence of a threshold; see \refcite{gottesman2010introduction}.
    \item Gates: For each physical gate $G$, we let $\ol{G}$ denote the associated logical operation on the codespace of the QEC code, and we let $\ol{\CG} = \ol{G}[\cdot] \ol{G}^\dag$ denote the associated map.  We let $\ol{\CG}_{\rm FT}$ denote a fault-tolerant gate gadget for $G$. In the absence of noise, the gadget $\ol{\CG}_{\rm FT}$ implements $\ol{\CG}$ when acting on states in the codespace, and in the presence of noise, it implements a map denoted $\tCG_{\rm FT}$, while obeying certain properties related to propagation of physical errors \cite{gottesman2010introduction}. A scheme for universal FTQC requires the specification of a fault-tolerant gate gadget for a universal set of gates. 
\end{itemize}

We now expand more on the formal properties that FTQC guarantees about these maps. First of all, each subset $S \subset [n']$ of physical qubits is either a correctable or an uncorrectable subset, depending on whether there exists an error channel $(\CI_{S^c} \otimes \CW_S)$  acting trivially on $S^c$ and nontrivially only on qubits within $S$ that can induce a logical error, in the sense that $\CQ \neq \CQ \circ (\CI_{S^c} \otimes \CW_S) \circ \CQ$. The assumption that the FTQC scheme has a threshold $p_0$ against local stochastic noise indicates that whenever $p < p_0$,
\begin{align}\label{eq:local_stochastic_error}
    \sum_{S \text{ uncorrectable}} (1-p)^{n'-|S|} p^{|S|} \leq \Gamma(\CE)\,,
\end{align}
where $\Gamma(\CE)$ a quantity that depends on $p$ but has the property that, for fixed $p$, $\Gamma(\CE)$ can be exponentially driven to zero simply by choosing a larger QEC code (indicated by the encoding map $\CE$) from the code family, incurring at most a polylogarithmic overhead (which is defined as the factor by which the number of physical qubits and physical gates must increase). 

Furthermore, abstracting away from physical errors, we can similarly quantify the logical error incurred by performing a fault-tolerant gate as 
\begin{align}\label{eq:fault-tolerant_gate_error}
    \frac{1}{2}\norm{ \ol{\CG} \circ \CQ \circ \tCQ_{\rm FT} - \CQ \circ \tCQ_{\rm FT} \circ \tCG_{\rm FT} \circ \tCQ_{\rm FT}}_{\diamond} \leq \Gamma(\CE)\,,
\end{align}
where again $\Gamma(\CE)$ is used to denote a quantity that vanishes with increasing code size---the overhead of achieving logical error $\Gamma(\CE) \leq \delta$ is at worst $\polylog(1/\delta)$. The equation above states that performing the noisy gate $\tCG_{\rm FT}$ sandwiched between noisy rounds of QEC and then projecting onto the codespace is nearly equivalent to simply performing noisy QEC, projecting onto the codespace, and applying the perfect logical gate.

\section{Distillation--teleportation protocol: details and error analysis}\label{sec:protocol_detailed}

In this section, we examine each step of our protocol in more detail, propagating errors and proving formal statements about the performance of each step. Whereas in previous sections we labeled inputs and outputs of quantum circuits with kets, here we label them with density matrices, since we will be investigating the impact of noise, which may lead states to lose their purity. It should be understood that the unitary circuit elements are applying the associated unitary channel on the density matrices on which they act. 

\subsection{Noisy physical resource state preparation with QRAM device}

As discussed, we assume we have access to a noisy physical QRAM device that attempts to perform the QRAM operation of \cref{eq:QRAM_intro} on the state $\ket{+}^{\otimes n}$ at the physical level, in order to produce the resource state $\ket{\Psi(g)}$ of \cref{eq:QRAM_resource_state}, given as input the dataset $g$. 
We depict this relation pictorially via the following circuit, where the red wire indicates that the outgoing quantum register  is unencoded and the qubits are physical qubits. 
\begin{circuit}
    \scalebox{1.0}{\begin{quantikz}[row sep={0em,between origins}, column sep={0em, between origins}, align equals at=1.5, wire types = {n,n}]
             & \oGateNormalSize \wire[r][2][draw=red,thick]{q} &[2em] \qwbundle[style={draw=red, text=red}]{n} &[1.3em] \rstick{$\textcolor{black}{\ketbra{\Psi(g)}}$}\\[2em]
            & \wire[u][1]["\mbox{\large{$g$}}"{below, pos=-0.15}]{c} & & 
            \end{quantikz}
            }
\end{circuit}
However, it is unrealistic to assume that the QRAM device can be implemented perfectly at the physical level. Rather, the state it produces is a state $\tpsig{g}$ that has some nonzero fidelity with $\ketbra{\Psi(g)}$. We assume that the noise in the device is independent of the dataset, as in \cref{def:dataset-independent_noise}, so that the QRAM channel decomposes as $\CN_2 \circ \CV(g) \circ \CN_1$. Then, we have
\begin{align}
    \tpsig{g} = \CN_2\Big[V(g) \, \CN_1\Big[\ketbra{+}^{\otimes n} \Big] V(g)^\dag \Big ]
\end{align}
In the circuit notation, we indicate the existence of noise by outlining the gate and its output in red. 
    \begin{circuit}
            \scalebox{1.0}{\begin{quantikz}[row sep={0em,between origins}, column sep={0em, between origins}, align equals at=1.5, wire types = {n,n}]
            & \oGateRedNormalSize \wire[r][2][draw=red,thick]{q} &[2em] \qwbundle[style={draw=red, text=red}]{n} &[1.3em] \rstick{$\textcolor{red}{\tpsig{g}}$}\\[2em]
            & \wire[u][1]["\mbox{\large{$g$}}"{below, pos=-0.15}]{c} & & 
            \end{quantikz}
            }
    \end{circuit}
We characterize the noise strength by the infidelity $1-F(g)_{\rm phys}$, where $F(g)_{\rm phys}$ is the fidelity
of the state $\tpsig{g}$ with respect to the ideal state $\ketbra{\Psi(g)}$, given by
\begin{align}
    F(g)_{\rm phys} = \bra{\Psi(g)}\tpsig{g}\ket{\Psi(g)}\,.
\end{align}

\subsection{Encoding physical resource states into logical resource states}\label{sec:encoding}

The noisy resource state $\tpsig{g}$ on $n$ physical qubits is unprotected from errors, and once prepared it must immediately be encoded into a QEC code. Ideally, this encoding process would prepare the state $\CE[\tpsig{g}]$, which is in the codespace of the code, as in the figure below, where the black wire with a black slash-$n$ indicates that the register contains $n$ logical qubits encoded into some number $n'>n$ of physical qubits.
    \begin{circuit}
            \scalebox{1.0}{\begin{quantikz}[row sep={0em,between origins}, column sep={0em, between origins}, align equals at=1, wire types = {n}]
            \lstick{$\textcolor{red}{\tpsig{g}}$}&\wire[r][2][draw=red,thick]{q}&[1.3em] \qwbundle[style={draw=red, text=red}]{n} &[1.8em] \eGateNormalSize \wire[r][2]{q} &[2em]\qwbundle{n} &[1em] \rstick{$\CE[\tpsig{g}]$} 
            \end{quantikz}
            }
    \end{circuit}
However, the encoding process may not be fault tolerant, that is, it may introduce additional errors into the logical output state that are proportional to the physical error rate $p$ of the hardware and the number of gates in the encoding circuit that implements $\CE$; this logical error cannot be suppressed simply by growing the code size. Hence, we need an encoding procedure $\CE_\mathrm{FT}$ that is fault-tolerant against the physical noise model, in the sense specified in~\Cref{prop:encoding} below. 
When we attempt to realize $\CE_\mathrm{FT}$ on faulty hardware, we instead implement a channel $\tCE_\mathrm{FT}$.
After applying $\tCE_\mathrm{FT}$, we apply the (noisy) fault-tolerant QEC gadget---implementing the map $\tCQ_{\rm FT}$---on the now-encoded state to correct for physical errors that occurred during $\tCE_\mathrm{FT}$. The output state is $\tCQ_{\rm FT} \circ \tCE_\mathrm{FT}[\tpsig{g}]$ (we could alternatively think of $\tCQ_{\rm FT}$ as part of $\tCE_\mathrm{FT}$ and omit it from the expression\footnote{However, note that physical errors in $\tCE_\mathrm{FT}$ can combine with physical errors in $\tCQ_{\rm FT}$ that create a logical difference between $\CQ \circ \tCQ_{\rm FT} \circ \tCE_\mathrm{FT}[\tpsig{g}]$ and $\CQ \circ \tCE_\mathrm{FT}[\tpsig{g}]$.}). This state is still not necessarily in the codespace, but the action of the noisy-but-fault-tolerant $\tCQ_{\rm FT}$ brings it close to the codespace, in the robust sense required by formal proofs of fault tolerance, for example, in \refcite{gottesman2010introduction}. The closest codespace state is obtained by applying the QEC projector $\CQ$, resulting in 
\begin{align}\label{eq:noisy_logical_resource_state}
    \ol{\phi(g)} = \CQ \circ \tCQ_{\rm FT} \circ \tCE_\mathrm{FT}[\tpsig{g}]\,.
\end{align}
Of course, we cannot apply noiseless $\CQ$ on actual hardware, but the probability of logical deviation from $\ol{\phi(g)}$ can be suppressed to an arbitrarily small quantity simply by growing the code size (incurring only logarithmic overheads or even constant overhead if one uses constant-rate codes~\cite{gottesman2014FTQCconstantOverhead}) provided that the physical error rate is below the threshold. Since the rest of our protocol is performed using fault-tolerant logical gates, we identify $\ol{\phi(g)}$ as the logical output of the noisy encoding, denoted pictorially by outlining the encoding gate in red. 
    \begin{circuit}
            \scalebox{1.0}{\begin{quantikz}[row sep={0em,between origins}, column sep={0em, between origins}, align equals at=1, wire types = {n}]
            \lstick{$\textcolor{red}{\tpsig{g}}$}&\wire[r][2][draw=red,thick]{q}&[1.3em] \qwbundle[style={draw=red, text=red}]{n} &[1.8em] \eGateRedNormalSize \wire[r][2]{q} &[2em]\qwbundle{n} &[1em] \rstick{$\ol{\phi(g)}$} 
            \end{quantikz}
            }
    \end{circuit}

We define the encoding error as the maximum (over arbitrary $n$-qubit input $\rho$) trace distance between the state
 $\CE[\rho]$ obtained from perfect encoding and the state $\CQ \circ \tCQ_{\rm FT} \circ \tCE_\mathrm{FT}[\rho]$ obtained from noisy encoding, as follows. 

\begin{definition}[Encoding error]
    \label{def:encoding-error}
    Consider a QEC code encoding $n$ logical qubits, specified by an encoding isometry $\CE$. Let $\CQ$ denote a QEC projector for the code $\CE$, and let $\tCQ_{\rm FT}$ denote the noisy implementation of a fault-tolerant QEC gadget for $\CQ$. Given a fault-tolerant encoding procedure $\CE_\mathrm{FT}$ and its noisy implementation $\tCE_\mathrm{FT}$, the encoding error is defined by the following expression
    \begin{align}
        \varepsilon_{\rm enc} = \sup_{\rho} \frac{1}{2} \lVert \CQ \circ\tCQ_{\rm FT}\circ\tCE_\mathrm{FT}[\rho]- \CE[\rho]\rVert_1\,,
    \end{align}
    where the supremum  is taken over all $n$-qubit physical states $\rho$.  Note that $\varepsilon_{\rm enc}$ is expected to have a dependence on $n$, as well as the physical error rate $p$ of the hardware. 
\end{definition}

The goal of the encoding step is to take a physical state $\tpsig{g}$ with some nonzero fidelity $F(g)_{\rm phys} = \bra{\Psi(g)}\tpsig{g}\ket{\Psi(g)}$, and map it to an encoded state $\ol{\phi(g)}$ with a smaller but still nonzero fidelity 
\begin{align}
    F(g)_{\rm enc} = \bra{\ol{\Psi(g)}}\ol{\phi(g)}\ket{\ol{\Psi(g)}}\,.
\end{align}
Using the definition of the encoding error, we can make an additive bound $F(g)_{\rm enc} \geq F(g)_{\rm phys} - \varepsilon_{\rm enc}$, which is sufficient when $F(g)_{\rm phys}$ is large enough that the right-hand side is greater than zero. However, if $F(g)_{\rm phys}$ is smaller than $\varepsilon_{\rm enc}$, then this bound is no longer meaningful. We can instead show a multiplicative bound roughly of the form $F(g)_{\rm enc} \geq (1-O(\varepsilon_{\rm enc}))F(g)_{\rm phys}$, which is more powerful in the small-fidelity regime. 
To show this, we have to manually perform a Pauli twirl of the encoding operation, which allows us to guarantee that the twirled encoding channel is stochastic (i.e., it acts as identity with some probability $1-O(\varepsilon_{\rm enc})$ and as some other CPTP channel otherwise); see, for example, \refcite[Lemma 5.2.4]{dankert2005efficientSimulationRandomQuantum} and \refcite[Lemma 3]{flammia2020efficientEstimationOfPauliChannels}. 
Without Pauli twirling, small encoding error alone is not sufficient to  rule out the possibility that the encoding channel has coherent error, such as small unitary rotations, which can degrade the fidelity in an additive rather than multiplicative way.

The Pauli twirl involves applying a randomly chosen physical Pauli operator, encoding, and then applying the same Pauli operator but on the logical level. The set of physical Pauli operators can cause the fidelity to degrade by a factor $(1-p)^n$ (which is on the order of $1-O(\varepsilon_{\rm enc})$ anyway). We capture this in the following proposition, which states that if we already have an encoding method $\CE_{\rm FT}'$ with small encoding error, we can construct a $\CE_{\rm FT}$ that  degrades the fidelity in this multiplicative way. The formal proof is provided in \cref{sec:delayed_proofs_encoding}.

\begin{restatable}[Pauli twirling the encoding channel]{proposition}{pauliTwirlEncoding}\label{prop:Pauli_twirl_encoding}
    Denote the fidelity of the physical state by $F(g)_{\rm phys} = \bra{\Psi(g)} \tpsig{g} \ket{\Psi(g)}$.  Suppose the processor is subject to circuit-level stochastic noise with strength $p$ (\cref{def:stochastic-noise}), and et $\CE_{\rm FT}'$ be a fault-tolerant encoding channel with encoding error $\varepsilon_{\rm enc}$ (as in \cref{def:encoding-error}). Then, there exists another fault-tolerant encoding channel $\CE_{\rm FT}$ (formed by Pauli-twirling $\CE_{\rm FT}'$), for which
    \begin{align}
        \bra{\ol{\Psi(g)}} \ol{\phi(g)} \ket{\ol{\Psi(g)}} \geq (1-p)^n\Big((1-3\varepsilon_{\rm enc})F(g)_{\rm phys}- 2\Gamma(\CE)\Big)\,,
    \end{align}
    where $\ol{\phi(g)}$ is defined from $\CE_{\rm FT}$ as in \cref{eq:noisy_logical_resource_state}, and $\Gamma(\CE)$ is a quantity that vanishes with increasing code size, provided the physical error rate $p$ is below a constant threshold, as discussed in \cref{sec:FTQC}. 
\end{restatable}

Next, we explain how to achieve an encoding with manageable $\varepsilon_{\rm enc}$. It should be noted that for fixed $p$, it is unavoidable for $\varepsilon_\mathrm{enc}$ to grow at least linearly with the number of logical qubits $n$, simply because we start with an unencoded state $\ket{\psi}$ on $n$ faulty physical qubits. Remarkably, it is possible to construct a fault-tolerant encoding procedure $\CE_\mathrm{FT}$ such that in the presence of noise, the logical encoding error $\varepsilon_\mathrm{enc}$ has no direct dependence on the block size and the distance of the code $\CE$ that we are encoding into. Such fault-tolerant encoding procedures exist for specific families of quantum codes such as the surface code~\cite{lodyga2015simpleSchemeEncoding}. Here, to keep our main results as agnostic to the underlying quantum code $\CE$ as possible, we opt to use a fault-tolerant encoding procedure that works for any quantum code~\cite{christandl2024faultTolerantQuantumInputOutput}. This procedure is based on concatenated-code quantum fault tolerance~\cite{aharonov1997FTQCconstantError}, and its fault tolerance for quantum input--quantum output tasks was recently proven under the circuit-level stochastic noise model defined in \cref{def:stochastic-noise}.
 
The circuit-level stochastic noise model in \cref{def:stochastic-noise} is a special case of a more general model called local-stochastic noise model~\cite{gottesman2014FTQCconstantOverhead}, which additionally allows for some correlations between gate faults. For our purposes, \Cref{prop:encoding} below will be using a result of Ref.~\cite{christandl2024faultTolerantQuantumInputOutput} that is proven under this circuit-level stochastic noise model. However, as is often the case in fault tolerance analysis, we expect the same statement extends to the local-stochastic noise model. Since implementing this extension is beyond the scope of our work, we keep the discussion simple by working with~\Cref{def:stochastic-noise} here.

\begin{restatable}[Error in fault-tolerant encoding channel for general codes]{proposition}{encodingErrorFT}\label{prop:encoding}
  Consider a family of quantum error-correcting codes encoding $n$ logical qubits into a codespace, and suppose that this family has a threshold $p_0$ with respect to local stochastic noise (implying \cref{eq:local_stochastic_error}).  Correspondingly, for a particular instance of the family (labeled by its encoding map $\CE$), let $\CQ$ be the ideal QEC projector, $\CQ_\mathrm{FT}$ be the fault-tolerant QEC gadget, and $\Gamma(\CE)$ be the logical error suppression function (see~\cref{sec:FTQC}). Then, there exists a fault-tolerant encoding procedure $\CE_\mathrm{FT}$ of size $|\CE|\cdot \poly(k)$, such that, when implemented under the circuit-level stochastic noise model (\cref{def:stochastic-noise}), the encoding error as defined in~\Cref{def:encoding-error} satisfies
\[
    \varepsilon_{\rm enc} = \sup_{\rho} \frac{1}{2} \lVert \CQ \circ\tCQ_{\rm FT} \circ\tCE_\mathrm{FT}[\rho]- \CE[\rho]\rVert_1\  \leq \Gamma(\CE) + 2\sqrt{cp} n + 2|\CE|(cp)^k,
\]
where $C$ is an absolute constant and $k$ can take values from a sequence of geometrically increasing integers, provided that the physical error rate $p$ is below some constant threshold. 
\end{restatable}

The formal proof is provided in \cref{sec:delayed_proofs_encoding}. We remark that the final term in the expression above is similar to the $\Gamma(\CE)$ term, in the sense that that for any $\delta$,  one can choose $k = \polylog(n/\delta)$ and ensure that it is smaller than $\delta$. Thus, the overhead is only polylogarithmic and we neglect its contribution in our analysis elsewhere in the protocol. These two propositions together allow us to show that the encoded state maintains substantial fidelity with the ideal resource state, for use in our protocol.

\begin{corollary}\label{cor:encoding_full}
    Suppose the quantum processor is subject to circuit-level stochastic noise with error rate $p$ (defined in \cref{def:stochastic-noise}). Let
    \begin{equation}
        F_{\rm min} = (1-np-6n\sqrt
        {cp}) \min_g \bra{\Psi(g)} \tpsig{g} \ket{\Psi(g)}\,.
    \end{equation}
    Then, there exists a fault-tolerant encoding procedure $\CE_{\rm FT}$, such that for all $g$ we have
    \begin{align}
        \bra{\ol{\Psi(g)}}\ol{\phi(g)} \ket{\ol{\Psi(g)}} \geq F_{\rm min},
    \end{align}
    where $\ol{\phi(g)}$ is defined from $\CE_{\rm FT}$ as in \cref{eq:noisy_logical_resource_state}.
    Moreover, the qubit and gate overhead of applying $\CE_{\rm FT}$ is $\poly(n)$. 
\end{corollary}
\begin{proof}
    This follows by defining $\CE_{\rm FT}'$ to be the encoding procedure shown to exist in \cref{prop:encoding}, and then applying \cref{prop:Pauli_twirl_encoding} to form $\CE_{\rm FT}$. Note that $(1-p)^n > 1-np + \delta$ for $\delta = O(n^2p^2)$. When the value of $k$ is taken to be $\polylog(n/\delta)$, and the QEC code size is taken to be $n \, \polylog(n/\delta)$, then the error terms $\Gamma(\CE)$ and $2|\CE|(cp)^k$ can be made $O(\delta)$, and the stated fidelity $F_{\rm min}$ can be guaranteed. 
\end{proof}

\subsection{Partial Clifford twirling}\label{sec:partial_Clifford_twirl}

The preparation and encoding processes produce a state $\ol{\phi(g)}$, but the distillation process discussed later can only distill $\ketbra{\ol{\Psi(g)}}$ from many copies of $\ol{\phi(g)}$ if the principal eigenvector of $\ol{\phi(g)}$ is $\ket{\ol{\Psi(g)}}$, which we cannot guarantee in general from the noise model we have assumed. 

One solution to this challenge is to use twirling \cite{wallman2016noiseTailoring}, also known as randomized compiling. This technique can convert general noise into better-behaved noise by inserting  random gates from a certain set into a quantum circuit, and compensating for their effect by modifying other gates in the circuit. We already saw an example of this process in the encoding step (\cref{prop:Pauli_twirl_encoding}), where the insertion of Pauli gates led the noise to become stochastic. It is well known that twirling by random Clifford gates can lead to stronger results, converting arbitrary (gate-independent) noise into depolarizing noise. For example, \refcite{mehta2024analysis} showed how randomized compiling, when actively applied within the physical QRAM device, can help mitigate the impacts of coherent errors on fidelity. However, in our setting---where we ideally consider a fully passive QRAM device and thus can only apply twirling ``outside'' the device---we cannot use the full Clifford group since, for example, if we conjugate the QRAM unitary $V(g)$ by the $H$ gate, we do not obtain another QRAM unitary. Our twirling set will instead be the subset of the $n$-qubit Clifford gates generated by $Z$, $X$, $\mathrm{CZ}$ and $\mathrm{CX}$ (CNOT) gates, which do map QRAM unitaries to QRAM unitaries.

\subsubsection{Setup and important lemmas}

In what follows, we make use of the properties of $\{0,1\}^n$ as an $n$-dimensional vector space over the field $\mathbb{F}_2$.

\begin{definition}[Partial Clifford twirling set]\label{def:twirling_set}
Suppose we are given $A$, $B$, $u$, and $v$, where
\begin{itemize}
    \item $A$ is an $n \times n$ invertible matrix over $\mathbb{F}_2$,
    \item $B$ is an upper triangular $n \times n$ matrix (with zeros on the diagonal) with entries in $\mathbb{F}_2$ ,
    \item $u \in \mathbb{F}_2^n$,
    \item $v \in \mathbb{F}_2^n$\,.
\end{itemize}
Define $M_A$ to be the quantum gate that enacts $M_A:\ket{x} \mapsto \ket{Ax}$ for all $x \in \mathbb{F}_2^n$, which can be composed from $O(n^2)$ $\mathrm{CX}$ gates via Gaussian elimination, and define the diagonal unitary
\begin{align}
    Q_B = \prod_{1\leq i < j \leq n}\mathrm{CZ}_{ij}^{B_{ij}}\,,
\end{align}
where $\mathrm{CZ}_{ij}^k$ is the identity gate when $k=0$ and the $\mathrm{CZ}$ gate between qubits $i$ and $j$ when $k=1$. Then, define the $n$-qubit quantum gate that corresponds to $(A,B,u,v)$ as
\begin{align}
    C = Z^v Q_BM_A^\dag X^u ,
\end{align}
 That is, $C$ is a product of $O(n)$ single-qubit Pauli-$Z$ gates determined by the entries of $v$, $O(n^2)$ $\mathrm{CZ}$ gates determined by the entries of $B$, $O(n)$ Pauli-$X$ gates determined by the entries of $X$, and $O(n^2)$ $\mathrm{CX}$ gates determined by the entries of $A$ (via $M_A$). 
 
 Let the \textit{twirling set} $\CliffordSet$ consist of all gates $C$ constructed in this fashion from some choice of $A, B, u, v$. When we say to generate a random gate from $\CliffordSet$, we mean to generate a uniformly random $A$, $B$, $u$, and $v$ and choose the corresponding $C \in \CliffordSet$.
\end{definition}
We call this partial Clifford twirling because the set $\CliffordSet$ is a subset of the $n$-qubit Clifford group. In particular, the set $\CliffordSet$ is generated by $X, Z, \mathrm{CX}, \mathrm{CZ}$, and the Clifford group is obtained by adding the Hadamard $H$ and phase gate $S$ to the generating set. 

\begin{proposition}\label{prop:twirl_consistency}
    Let  $g: \mathbb{F}_2^n \rightarrow \mathbb{F}_2$ be a data table (Boolean function) and let $C \in \CliffordSet$ be a twirling gate corresponding to choice $(A,B,u,v)$. Let $M_A$ and $Q_B$ be defined as in \cref{def:twirling_set}. Then, we have
    \begin{align}
        \ket{\Psi(g)} =   V(g) \ket{+}^{\otimes n} = C\,  V(g_C) \ket{+}^{\otimes n} = C\ket{\Psi(g_C)}
    \end{align}
    where for any $x \in \mathbb{F}_2^n$
    \begin{align}\label{eq:g_C}
        g_C(x) = g(Ax \oplus u) \oplus \left[x\cdot v\right] \oplus \left[x^\top B x\right]
    \end{align}
\end{proposition}
\begin{proof}
    We note that $ X^u M_A \ket{+}^{\otimes n} = \ket{+}^{\otimes n}$ since $X^u$ and $M_A$ are made from $X$ and $\mathrm{CX}$ gates. The gates $X^u$ and $M_A$ simply permute the computational basis states, viewed as vectors in $\mathbb{F}_2^n$, by an affine transformation, and we can further note that $M_A^\dag X^u V(g) X^uM_A = V(g')$ where $g'(x) = g(Ax\oplus u)$. Finally, the operation $Z^v Q_B$ is diagonal, and equal to $V(h)$ where $h(x) = [x \cdot v] \oplus \left[x^\top B x\right]$. We have the composition rule $V(h)V(g') = V(g'\oplus h) = V(g_C)$. The statement follows from these facts as
    \begin{align}
        C\,V(g) \ket{+}^{\otimes n} &= Z^v Q_B M_A^\dag X^uV(g) \ket{+}^{\otimes n}=Z^v Q_B M_A^\dag X^uV(g) X^u M_A \ket{+}^{\otimes n} = Z^v Q_B V(g') \ket{+}^{\otimes n}  \nonumber \\
        &= V(h) V(g')\ket{+}^{\otimes n} = V(g_C)\ket{+}^{\otimes n}
    \end{align}
\end{proof}

Next, we will examine how random choice of $C \in \CliffordSet$ spreads Pauli operators. We will examine the set of signed Pauli operators, and specific subsets of it. 
\begin{definition}[Signed Pauli set and noteworthy subsets]\label{def:signed_Pauli_strings}
    Define $\PauliSet$ to be the set of $2^{2n+1}$ signed Pauli strings written in the canonical form $\iUnit^{a \cdot  b} (-1)^sX^bZ^a$, where $s \in \mathbb{F}_2$, $a,b \in \mathbb{F}_2^n$. The factor of $\iUnit$ ensures that each operator in $\PauliSet$ is Hermitian.  We partition $\PauliSet$ into several nonoverlapping subsets
    \begin{align}
        \PauliSet_{0} &= \{\iUnit^{a \cdot  b}(-1)^sX^bZ^a \in \PauliSet: a = b = 0^n, s = 0\} = \{\Id \} \\
        \PauliSet_{1} &= \{\iUnit^{a \cdot  b}(-1)^sX^bZ^a \in \PauliSet: a = b = 0^n, s = 1\} = \{-\Id \} \\
        \PauliSet_{Z} &= \{\iUnit^{a \cdot  b}(-1)^sX^bZ^a \in \PauliSet: a \neq 0^n, b = 0^n\}  \\
        \PauliSet_{\rm even} &= \{\iUnit^{a \cdot  b}(-1)^sX^bZ^a \in \PauliSet: b \neq 0^n, a \cdot b = 0\} \\
        \PauliSet_{\rm odd} &= \{\iUnit^{a \cdot  b} (-1)^sX^bZ^a \in \PauliSet: a \cdot b = 1\}
    \end{align}
\end{definition}

We argue that the Pauli operators are spread uniformly over the subset of $\PauliSet$ to which it belongs. The full proof of the following proposition is provided in \cref{sec:delayed_proof_uniform_spreading}.

\begin{restatable}[Twirling spreads Paulis uniformly]{proposition}{uniformPauliProp}\label{prop:twirling_uniform}
    Let $C \sim \CliffordSet$ denote choosing $C$ randomly from $\CliffordSet$ as described in \cref{def:twirling_set}.
    Given a fixed $P \in \PauliSet$, for any $C \in \CliffordSet$, $C P C^\dag \in \PauliSet$ since $C$ is Clifford. Furthermore, let $Q \in \PauliSet$ be a random variable formed by choosing $C \sim \CliffordSet$ and defining $Q= CPC^\dag$. Then, the distribution over $Q$ is the uniform distribution over the subset of $\PauliSet$ (i.e., $\PauliSet_0$, $\PauliSet_1$, $\PauliSet_Z$, $\PauliSet_{\rm even}$, or $\PauliSet_{\rm odd}$) to which $P$ belongs. 
\end{restatable}
\begin{proof}[Proof idea]
Consider a Pauli $P \in \PauliSet$ written in canonical form as $P = \iUnit^{a \cdot b} (-1)^{s} X^b Z^a$, and a Clifford $C = Z^v Q_BM_A^\dag X^u \in \CliffordSet$. We explicitly compute the Pauli $CPC^\dag = \iUnit^{a' \cdot b'}(-1)^{s'} X^{b'}Z^{a'}$ and give formulas for $s',a',b'$ in terms of $s,a,b,v,B,A,u$. Then, we use these formulas to verify that for $t \in \{0,1,Z,\mathrm{odd}, \mathrm{even}\}$ if $P \in \PauliSet_{t}$ and $v, B, A, u$ are chosen uniformly at random, then $s',a',b'$ are uniformly random over all values consistent with the definition of $\PauliSet_{t}$. 
\end{proof}

\subsubsection{Modification to circuit}

We now explain precisely how our protocol changes when we implement partial Clifford twirling. We want to implement $V(g)$. Each time we query the physical QRAM device, we generate a $C$ uniformly at random from the Clifford subset $\CliffordSet$, as defined in \cref{def:twirling_set}. We update the function $g$ to be $g_C$ using \cref{eq:g_C}. In particular, the value at each address $x$ may need to be updated, but each one can be computed with a simple $\mathrm{poly}(n)$-time classical computation and a single query to learn $g(y)$ for a particular $y$. We use the QRAM device to produce the noisy physical resource state $\tpsig{g_C}$, and then we encode that resource state, yielding $\ol{\phi(g_C)}$, as in \cref{eq:noisy_logical_resource_state}. Only then do we fault tolerantly apply the gate $\ol{C}$, which consists of $O(n^2)$ logical Clifford gates. This full procedure is depicted in the following circuit. 
\begin{circuit}\label{eq:twirl_figure}
            \scalebox{1.0}{\begin{quantikz}[row sep={0em,between origins}, column sep={0em, between origins}, align equals at=1.5, wire types = {n,n,n}]
            & \oGateRedNormalSize \wire[r][2][draw=red,line width = \thickWire]{q} &[2em] \qwbundle[style={draw=red, text=red}]{n} &[1.2em] \eGateRedNormalSize \wire[r][3][line width = \thickWire] {q} &[2em]\qwbundle{n} &[1.2em] \gate[style={fill=\twirlColor, text height = 0.3em}]{\ol{C}} &[1.5em] \rstick{$\ol{C}\, \ol{\phi(g_C)} \,\ol{C}^\dag$}\\[2em]
            & \gate[style={fill=\twirlColor, text height = 0.3em}]{\mbox{\!\small$g \mapsto g_C$\!}}& & \\[1.5em]
            & \wire[u][2]["\mbox{\large{$g$}}"{below, pos=-0.15}]{c} & & 
            \end{quantikz}
            }
    \end{circuit}
    Each time the QRAM device is called, an independent random $C$ is chosen. Thus, the output state may be modeled as the mixture \begin{align}\label{eq:phi_twirl}
        \ol{\phi(g)}_{\rm twirl} = \EV_{C \sim \CliffordSet} \ol{C} \, \ol{\phi(g_C)} \ol{C}^\dag\,.
    \end{align} 

    We have not included the twirling step in the main circuit of \Cref{fig:protocol_overview}; partly because it would clutter the figure, and partly because we feel that the twirling step may not be necessary in practice if the QRAM device can be constructed in such a way that $\ol{\phi(g)}$ already has $\ket{\ol{\Psi(g)}}$ as its principal eigenvector. 

\subsubsection{Twirling ensures the principal eigenvector is correct}
Now, we are ready to present the main finding of this section. The full proof is provided in \cref{sec:delayed_proof_correct_top_eigenvector}.

\begin{restatable}[Correct top eigenvector]{proposition}{correctTopEigenvector}\label{prop:correct_top_eigenvector_after_twirling}
    Suppose that for every $g$, the state $\ol{\phi(g)}$ defined in \cref{eq:noisy_logical_resource_state} satisfies $\bra{\ol{\Psi(g)}}\ol{\phi(g)}\ket{\ol{\Psi(g)}} \geq F_{\rm min}$, and suppose that the faulty QRAM device is subject to dataset-independent noise (\cref{def:dataset-independent_noise}). Let $C \sim \CliffordSet$ denote drawing $C$ randomly from the twirling set as described in \cref{def:twirling_set}, and let $\EV$ denote expectation value. For each $g$, let $\ol{\phi(g)}_{\rm twirl}$ be defined as in \cref{eq:phi_twirl}. 
    Then $\ol{\phi(g)}_{\rm twirl}$ satisfies the eigenvalue equation
    \begin{align}\label{eq:phi_twirl_eigenvalue_equation}
         \ol{\phi(g)}_{\rm twirl}\ket{\ol{\Psi(g)}} = \lambda_{\rm twirl}\ket{\ol{\Psi(g)}}\,,
    \end{align}
    with $\lambda_{\rm twirl} \geq F_{\rm min}$. Furthermore, all other eigenvalues of $\ol{\phi(g)}_{\rm twirl}$ are no larger than $2^{-n+1}$.
\end{restatable}
\begin{proof}[Proof idea]
Due to the dataset-independent assumption, the noise in the physical QRAM device and the encoding step can be consolidated into a noise matrix $\chi_{P,P'}$ (where $P,P' \in \PauliSet$) for which it is always possible to write
\begin{align}
    \ol{\phi(g_C)} &= \sum_{P,P' \in \PauliSet} \chi_{P,P'} \ol{P} \ketbra{\ol{\Psi(g_C)}} \ol{P}'\,. 
\end{align}
Noting $\ket{\ol{\Psi(g_C)}} = \ol{C}^\dag \ket{\ol{\Psi(g)}}$ (\cref{prop:twirl_consistency}), we can then say 
\begin{align}
    \ol{\phi(g)}_{\rm twirl} =  \sum_{P,P' \in \PauliSet} \chi_{P,P'} \EV_{C \sim \CliffordSet} \ol{C} \, \ol{P} \, \ol{C}^\dag \ketbra{\ol{\Psi(g)}} \ol{C}\, \ol{P}'\, \ol{C}^\dag \,.
\end{align}
We now use the fact that randomly choosing $C$ leads $CPC^\dag$ to uniformly cover a large subset of of $\PauliSet$ (\cref{prop:twirling_uniform}). The offdiagonal terms with $P \neq \pm P'$ vanish due to the uniformly random sign. This conclusion did not require the full size of $\CliffordSet$; it is a consequence of the fact that $\CliffordSet$ contains the Pauli group, and Pauli-twirling leads the effective channel to become a Pauli channel.  If $\CliffordSet$ were the entire Clifford group, then $C P C^\dag$ would be a uniformly random Pauli, and we would immediately be able to use the 1-design property of the Pauli set to say that, for any $g$ and for any case where $P \neq \pm \Id$, the quantity $\EV_{C \sim \CliffordSet} \ol{C} \, \ol{P} \, \ol{C}^\dag \ketbra{\ol{\Psi(g)}} \ol{C}\, \ol{P}\, \ol{C}^\dag $ is equal to the maximally mixed state. This would then immediately imply the statement we seek, and also that all other eigenvalues $\lambda_{\rm other}$ satsify $\lambda_{\rm other} \leq 2^{-n}$. Generally, this would be a manifestation of the fact that Clifford twirling transforms any noise channel into a depolarizing channel. However,  as $\CliffordSet$ is not the full Clifford group, we have to do more work; some subsets of $\PauliSet$ may be underweighted or overweighted. Nevertheless, we show that there is enough uniformity to  recover a similar result, albeit with the bound on $\lambda_{\rm other}$ suffering a factor-of-2 overhead. 
\end{proof}

\subsection{Distillation of logical resource states}\label{sec:distillation}

The encoding step, combined with twirling, produces the logical encoded states $\ol{\phi(g)}_{\rm twirl}$, which are guaranteed to have $\ket{\ol{\Psi(g)}}$ as their principal eigenvector (\cref{prop:correct_top_eigenvector_after_twirling}).
The remaining steps of the protocol act directly on the encoded states (using fault-tolerant gadgets for a universal set of gates): all physical errors created during the distillation and teleportation portions of the protocol can be prevented from turning into logical errors by growing the code distance. Thus, we describe the distillation and teleportation protocol by their logical quantum circuits, and in the description of our protocol, we keep the overline notation to remind the reader that these operations are meant to be performed fault tolerantly on the encoded Hilbert space. 

However, the distillation ideas presented here apply generally, regardless of whether/how the states and operations are encoded with QEC. Later in the section, including some of the proposition statements, we drop the overlines, since the statements could be of independent interest. 

The procedures we consider for distilling the logical QRAM resource state are state agnostic. Namely, given many copies of an arbitrary state $\ol{\rho}_{\rm in}$, the distillation protocol prepares (up to small trace distance error) the pure state $\ketbra{\ol{\Xi}}$, where $\ket{\ol{\Xi}}$ is the principal eigenvector of $\ol{\rho}_{\rm in}$. In other words, our distillation procedure is equivalent to the task of \textit{quantum purity amplification} \cite{li2024optimalQuantumPurityAmplification}. 

Pictorially, the distillation step accomplishes the following operation within our protocol. For simplicity, the circuit below depicts all copies being prepared at the beginning and processed at once; in practice, it is possible to prepare the states in a streaming fashion with less space requirement, as we discuss later. 
\begin{circuit}\label{fig:distillation_figure}
            \scalebox{1}{\begin{quantikz}[row sep={2em,between origins}, column sep={0em, between origins}, align equals at=1, wire types = {n,n,n,n,n}]
            \lstick{$\ol{\phi(g)}_{\rm twirl}$}&[1.2em] \qwbundle{n} &[2.4em] \dGate &[0em] \wire[l][3]{q} &[0em] &[3.5em] & \\
            \lstick{$\ol{\phi(g)}_{\rm twirl}$}& \qwbundle{n} & & \wire[l][3]{q} & & & \\
            \lstick{$\ol{\phi(g)}_{\rm twirl}$}& \qwbundle{n} & & \wire[l][3]{q} & & \wire[l][2]{q}\qwbundle{n} & \rstick{$\ol{\phi(g)}_{\rm dist} \approx\ketbra{\ol{\Psi(g)}}$} \\
            & \gate[style={draw = none, fill=white}]{\myvdots} & & & & &\\
            \lstick{$\ol{\phi(g)}_{\rm twirl}$}& \qwbundle{n} & & \wire[l][3]{q}& & &
            \end{quantikz}
            }
\end{circuit}

Now we present the main result of this section, as applied to our protocol, utilizing the general techniques discussed later in the section. 
\begin{proposition}\label{prop:distillation_general}
    Suppose that the QRAM device is subject to dataset-independent noise, as in \cref{def:dataset-independent_noise}, and that for every $g$, the states $\ol{\phi(g)}$ produced on input $g$ (see \cref{eq:noisy_logical_resource_state}) satisfy $\bra{\ol{\Psi(g)}}\ol{\phi(g)}\ket{\ol{\Psi(g)}} \geq F_{\rm min}$, with $F_{\rm min}\geq 2^{-n+2}$. Then, for any error parameter $\varepsilon_{\rm dist}$, by applying a carefully crafted sequence of subsequent (fractional) swap operations on $O(\frac{1-F_{\rm min}}{F_{\rm min}^2}(\frac{1}{\varepsilon_{\rm dist}}+\frac{1}{F_{\rm min}}))$ copies of $\ol{\phi(g)}_{\rm twirl}$ (from \cref{eq:phi_twirl}) we can distill a state $\ol{\phi(g)}_{\rm dist}$ that satisfies
    \begin{align}\label{eq:prop_distillation_general_equation}
        \frac{1}{2} \nrm{\ketbra{\ol{\Psi(g)}}- \ol{\phi(g)}_{\rm dist}}_1 \leq \varepsilon_{\rm dist}\,.
    \end{align} 
    The protocol requires $O(1)$ single-qubit gates and $O(n)$ controlled swap operations, per copy consumed. 
\end{proposition}
\begin{proof}
    This is based on \cref{prop:correct_top_eigenvector_after_twirling}, which shows that $\ol{\phi(g)}_{\rm twirl}$ has the correct top eigenvector, combined with some state-agnostic quantum purity amplification protocol to distill the top eigenvector.
    When $F_{\rm min}$ is close to $1$ the iterated swap test achieves this with a close-to-optimal $\approx(1-F_{\rm min})/\varepsilon_{\rm dist}$ number of copies of $\ol{\phi(g)}_{\rm twirl}$ and the same order of swap tests, that is, \cref{eq:swap_test}---see \cref{prop:distillation_swap_test_high_fidelity} and \cref{lem:eigenvalue_to_trace_distance}. However, for smaller values of $F_{\rm min}$, the iterated swap test incurs an $\exp(\Theta(1/F_{\rm min}))$ overhead, see the discussion at the end of \cref{sec:distillation_iterated_swap_test}.

    In the general case, we can exploit the fact that the second largest eigenvalue of $\ol{\phi(g)}_{\rm twirl}$ is upper bounded by $2^{-n+1} \leq F_{\rm min}/2$ due to \cref{prop:correct_top_eigenvector_after_twirling}, and use a protocol based on quantum principal component analysis (\cref{prop:QPCADistillation}), achieving the stated copy complexity utilizing a matching number of swap-test-like gadgets from~\cref{fig:LMRDistill}. 

    Both \cref{prop:QPCADistillation} and \cref{prop:distillation_swap_test_high_fidelity} are described as producing a state for which the largest eigenvalue is close to 1. We can turn this into a trace distance bound via \cref{lem:eigenvalue_to_trace_distance}. The stated gate complexity follows from the observation that \cref{fig:LMRDistill} has 3 controlled swap operations (targetting $(n+1)$-qubit registers) and 4 single-qubit gates. 
\end{proof}

To convert to trace distance as in \cref{eq:prop_distillation_general_equation}, our analysis uses the fact that a mixed state is as close to its principal eigenvector as its principal eigenvalue is to $1$, captured in the following lemma.
\begin{lemma}\label{lem:eigenvalue_to_trace_distance}
    A quantum state $\ol{\rho}_{\rm out}$, whose principal eigenvector is $\ket{\ol{\Xi}}$ with eigenvalue $1-\eta$, satisfies
    \begin{align}
        \frac{1}{2}\nrm{\ol{\rho}_{\rm out} - \ketbra{\ol{\Xi}}}_1 = \eta.
    \end{align}
\end{lemma}
\begin{proof}
     Let $\lambda_1,\ldots, \lambda_d$ be the eigenvalues of the state $\ol{\rho}_{\rm out}$, with $\lambda_1 = 1-\eta$ the principal eigenvector. They satisfy $\sum_{j=1}^d \lambda_j = 1$, so $\sum_{j=2}^d \lambda_j = \eta$. The operator $\ol{\rho}_{\rm out} - \ketbra{\ol{\Xi}}$ has the same eigenvectors as $\ol{\rho}_{\rm out}$, and the eigenvalues are $-\eta, \lambda_2,\ldots,\lambda_d$. The sum of the singular values is thus equal to $2\eta$, which completes the proof. 
\end{proof}

\subsubsection{Distillation with the iterated swap test}\label{sec:distillation_iterated_swap_test}
In this subsection, we drop the overlines in our notation, and speak generally about the task of quantum purity amplification. 
We first consider the iterated swap test, a quantum purity amplification procedure studied in detail in \refcite{childs2025streamingPurification} (a similar procedure was discussed in \refcite{irani2022quantumSearchToDecisionReductions}), although there the analysis assumed that the input states $\rho_{\rm in}$ are a mixture of a rank-1 pure state and the maximally mixed state, that is, of the form $\rho_{\rm dep} = (1-\delta)\ketbra{\Xi} + \frac{\delta }{d} \Id$, with $d$ the Hilbert space dimension and $\Id/d$ the maximally mixed state. We do not make this assumption here. See also \refcite{grier2025StreamingQStatePurification} for an analysis of the iterated swap test without this assumption. 

The basic ingredient of this distillation procedure is the swap test (which in our application would be performed fault tolerantly using fault-tolerant gadgets for controlled swap, Hadamard, and measurement). 
\begin{circuit}\label{eq:swap_test}
\scalebox{1.0}{
\begin{quantikz}[row sep={2em,between origins}, column sep=0.75em, align equals at=2]
     \lstick{$\ketbra{{0}}$} & &  \gate{{H}}   & \ctrl{2}  & \gate{{H}} & \meter{} \\
     \lstick{${\rho}_{\rm in}$} & \qwbundle{} &  & \targX{} & & \rstick{${\rho}_{\rm out}$}\\
     \lstick{${\rho}_{\rm in}$} & \qwbundle{} &  & \targX{} & & \trash{\text{\small{trace}}}
     \end{quantikz}
} 
\end{circuit}
We say that the swap test passes if the first qubit is measured in $\ket{{0}}$ at the end of the circuit. Acting on two copies of the trace-1 state ${\rho}_{\rm in}$, the probability of the swap test passing is given by 
\begin{equation}\label{eq:probability_passing_swap_test}
    \Pr[\text{swap test passes}] = \frac{1 + \Tr({\rho}_{\rm in}^2)}{2},
\end{equation}
and the state one obtains conditioned on the swap test passing is
\begin{equation}\label{eq:oneSwap}
    {\rho}_{\rm out} = \frac{{\rho}_{\rm in}+{\rho}_{\rm in}^2}{1 + \Tr({\rho}_{\rm in}^2)}.
\end{equation}
If the principal eigenvalue of ${\rho}_{\rm in}$ is $1-\eta_{\rm in}$, then we can bound 
\begin{align}\label{eq:pSuccLower}
    \Pr[\text{swap test passes}] &\geq \frac{1+(1-\eta_{\rm in})^2}{2} \geq 1-\eta_{\rm in}. 
\end{align}
It is easy to verify that ${\rho}_{\rm out}$ and ${\rho}_{\rm in}$ commute, and they also  have the same eigensubspaces since $[0,1] \ni x\rightarrow x + x^2$ is strictly monotone. Moreover, denoting the principal eigenvalue of ${\rho}_{\rm out}$ by $1-\eta_{\rm out}$, we have
\begin{align}
    \eta_{\rm out} &= 1- \frac{(1-\eta_{\rm in}) + (1-\eta_{\rm in})^2}{1 + \Tr({\rho}_{\rm in}^2)}\nonumber\\&
    \leq 1- \frac{(1-\eta_{\rm in}) + (1-\eta_{\rm in})^2}{1 + (1-\eta_{\rm in})^2+\eta_{\rm in}^2} \nonumber\\&
    = \frac{\eta_{\rm in}}{2}\left(\frac{1+\eta_{\rm in}}{1-\eta_{\rm in} + \eta_{\rm in}^2}\right)\,,\label{eq:etaBound}
\end{align}
which is about $\eta_{\rm in} /2$ for $\eta_{\rm in}\ll 1$. 

The idea of the distillation protocol studied in \refcite{irani2022quantumSearchToDecisionReductions,childs2025streamingPurification,grier2025StreamingQStatePurification} is to iterate the swap test by feeding two copies of ${\rho}_{\rm out}$ that both passed the swap test back into the swap test as ${\rho}_{\rm in}$, and repeating. Each time we successfully pass the swap test, the output state has higher purity and also the probability that the swap test passes at the next level is closer to 1. By continuing this process for sufficiently many iterations and consuming sufficiently many copies of ${\rho}_{\rm in}$, we can produce a state with purity arbitrarily close to 1.

For small $\eta_{\rm in}$ (i.e., the regime of high input fidelity), the number of copies needed scales as $O(\eta_{\rm in}/\varepsilon_{\rm dist})$. \Cref{fig:distillation_figure} depicts creating all copies of the input state at the beginning of the protocol. For the swap test approach, if swap tests on disjoint pairs can be performed in parallel, the overall depth of the protocol could be $O(\log(\eta_{\rm in}/\varepsilon_{\rm dist}))$. However,  as described in \refcite{childs2025streamingPurification}, the swap test approach can also be implemented in a streaming fashion, allowing one to reduce the space requirements to  $O(\log(\eta_{\rm in}/\varepsilon_{\rm dist}))$ at the expense of requiring $O(\eta_{\rm in}/\varepsilon_{\rm dist})$ depth. We now give the formal statement of performance and costs for this type of distillation when $\eta_{\rm in} < 1/4$, for general input states. 

\begin{proposition}\label{prop:distillation_swap_test_high_fidelity}
	Consider a qudit in a mixed state $\rho_{\rm in}$ with its principal eigenvector $\ket{\Xi}$ having eigenvalue $1-\eta_{\rm in}$. After $k$ successful iterations of the swap test procedure we obtain a state $\rho_{k}$ such that $[\rho_{\rm in},\rho_{k}]=0$ and if $\eta_{\rm in} < 1/4$, then  $\bra{\Xi}\rho_{k}\ket{\Xi}\geq 1-\frac{2^{-k}\eta_{\rm in}}{1-4\eta_{\rm in}}$. The expected number of copies of $\rho_{\rm in}$ consumed is upper bounded by $2^k/\sqrt{1-4\eta_{\rm in}}$, and the expected number of required individual swap tests is upper bounded by $(2^k-1)/\sqrt{1-4\eta_{\rm in}}$, moreover the protocol needs to store at most $1$ qubit and $k+1$ qudits for preparing $\rho_{k}$.
\end{proposition}

\begin{proof}
	Let $\rho_0 = \rho_{\rm in}$ and $\eta_0 = \eta_{\rm in}$. The fact that $[\rho_{0},\rho_{k}]=0$ follows directly from \cref{eq:oneSwap}. For the stated space-efficient implementation, observe that until $\rho_{k}$ is prepared it suffices to store at most one copy of each of $\rho_{0},\rho_{1},\ldots,\rho_{k-1}$, except for a single $\rho_i$ where a swap test is performed, see \refcite[Algorithm 3]{childs2025streamingPurification}.
	
	Our proof is inspired by the calculations of \refcite{childs2025streamingPurification}. 
	Let us define $\eta_i:=1-\bra{\Xi}\rho_{i}\ket{\Xi}$. We can verify by induction that $\eta_i\leq\frac{2^{-i}\eta_0}{1-4\eta_0+2^{2-i}\eta_0}\left(\leq\frac{2^{-i}\eta_0}{1-4\eta_0}\right)$. This trivially holds for $i=0$ and the induction step can be verified as follows:

	\begin{align}
		\eta_{i+1}&\leq \frac{\eta_{i}}{2}\left(\frac{1+\eta_{i}}{1-\eta_{i} + \eta_{i}^2}\right) \tag{by \eqref{eq:etaBound}} \\&
		\leq\frac{\eta_i}{2-4\eta_i} \tag{since $(1+\eta_{i})(1-2\eta_i)\leq 1-\eta_{i} + \eta_{i}^2$}\\&
		\leq\frac{2^{-i}\eta_0}{2(1-4\eta_0+2^{2-i}\eta_0)-2^{2-i}\eta_0} \tag{by monotonicity of $\frac{x}{2-4x}$ and the induction hypothesis}\\&	
		=\frac{2^{-1-i}\eta_0}{1-4\eta_0+2^{1-i}\eta_0}.\label{eq:etaBoundInd}
	\end{align}
	Let us denote by $p_{i}^{\rm succ}=\frac{1 + \Tr(\rho_{i-1}^2)}{2}$ the success probability \cref{eq:probability_passing_swap_test} of the swap test on $\rho_{i-1}$. The expected number $c_i$ of copies of $\rho_0$ needed for preparing $\rho_{i}$ is $c_i=2^i\prod_{j=1}^{i}\frac{1}{p_{j}^{\rm succ}}$, which is easy to verify by induction. The $i=0$ case is trivial, and the induction step follows from the observation that to obtain $\rho_{i+1}$ we need to repeat the swap test an expected number of $\frac{1}{p_{i+1}^{\rm succ}}$ many times on two copies of $\rho_{i}$, i.e., $c_{i+1}=2\frac{c_i}{p_{i+1}^{\rm succ}}$.
	
	We can bound $c_i$ by deriving the bound $\prod_{j=1}^{i}p_{j}^{\rm succ}\geq \sqrt{1-4\eta_0}$ as follows
	\begin{align*}
		\Big(\prod_{j=1}^{i}p_{j}^{\rm succ}\Big)^{\!2}&
		\geq \prod_{j=1}^{i}(1-\eta_{j-1})^2 \tag{by \cref{eq:pSuccLower}}\\&
		\geq \prod_{j=1}^{i}(1-2\eta_{j-1})
		\geq \prod_{j=1}^{i}\left(1-\frac{2^{2-j}\eta_0}{1-4\eta_0+2^{3-j}\eta_0}\right) \tag{by \cref{eq:etaBoundInd}}\\&
		= 1-4\eta_0+2^{2-i}\eta_0 \tag*{$\big($by induction: $(1-\frac{2^{1-i}\eta_0}{1-4\eta_0+2^{2-i}\eta_0})(1-4\eta_0+2^{2-i}\eta_0)=(1-4\eta_0+2^{1-i}\eta_0)\big)$}.
	\end{align*}
	The expected number $s_i$ of swap tests used for preparing $\rho_i$ can also be bounded by $(2^i-1)\prod_{j=1}^{i}\frac{1}{p_{j}^{\rm succ}}$ via induction:
	\begin{align*}
		s_{i+1}	= \frac{1+2s_i}{p_{i+1}^{\rm succ}}
		\leq \frac{1+2\cdot(2^i-1)\prod_{j=1}^{i}\frac{1}{p_{j}^{\rm succ}}}{p_{i+1}^{\rm succ}}
		= (2^{i+1}-1)\prod_{j=1}^{i+1}\frac{1}{p_{j}^{\rm succ}} + \frac{1-\prod_{j=1}^{i}\frac{1}{p_{j}^{\rm succ}}}{p_{i+1}^{\rm succ}}
		\leq (2^{i+1}-1)/\sqrt{1-4\eta_0}.\tag*{\qedhere}
	\end{align*}
\end{proof}

\Cref{prop:distillation_swap_test_high_fidelity} only applies when $\eta_{\rm in} < 1/4$, but the iterated swap test can still be successful even when the input fidelity is lower. We now consider what happens when $\eta_{\rm in}\gg 0$. Let $\gamma_{\rm in}:=1-\eta_{\rm in} = \bra{\Xi}\rho_{\rm in}\ket{\Xi}=\lambda_1(\rho_{\rm in})$ denote the principal eigenvalue of $\rho_{\rm in}$, where the notation $\lambda_i(\sigma)$ denotes the $i$-th largest eigenvalue of $\sigma$. Let us assume that $\frac{\lambda_2(\rho_{\rm in})}{\lambda_1(\rho_{\rm in})}\leq \alpha$ for some value $\alpha<1$, then we get the following guarantee on $\gamma_{\rm out}:=\bra{\Xi}\rho_{\rm out}\ket{\Xi}$
\begin{align}
	\gamma_{\rm out} &= \frac{\gamma_{\rm in} + \gamma_{\rm in}^2}{1 + \Tr(\rho_{\rm in}^2)}\nonumber\\&
	\geq \frac{\gamma_{\rm in} + \gamma_{\rm in}^2}{1 + \gamma_{\rm in}^2+\alpha\gamma_{\rm in}(1-\gamma_{\rm in})} \nonumber\\&
	= \frac{\gamma_{\rm in}(1+\gamma_{\rm in})}{1+\gamma_{\rm in}+(\alpha-1)\gamma_{\rm in}(1-\gamma_{\rm in})} \nonumber\\&
	= \frac{\gamma_{\rm in}}{1-(1-\alpha)\gamma_{\rm in}\frac{1-\gamma_{\rm in}}{1+\gamma_{\rm in}}}=:p(\gamma_{\rm in})\label{eq:RatioEst}
\end{align}
which is always greater than $\gamma_{\rm in}$ when $\gamma_{\rm in}\in (0,1)$. Moreover $\frac{\lambda_2(\rho_{\rm out})}{\lambda_1(\rho_{\rm out})}=\frac{\lambda_2(\rho_{\rm in})}{\lambda_1(\rho_{\rm in})}\cdot\frac{1+\lambda_2(\rho_{\rm in})}{1+\lambda_1(\rho_{\rm in})}\leq\alpha$. Finally, by computing the derivative of \cref{eq:RatioEst} in $\gamma_{\rm in}$, we get
\begin{align*}
	\frac{1+\gamma_{\rm in} (2-\gamma_{\rm in}+ 2 \alpha\gamma_{\rm in})}{(1+\gamma_{\rm in} (\alpha +(1-\alpha)\gamma_{\rm in} ))^2}>0,
\end{align*}
which means that $p(\gamma_{\rm in})$ in \cref{eq:RatioEst} is monotonically increasing in $\gamma_{\rm in}$; therefore, if we replace $\gamma_{\rm in}$ with a lower bound on $\gamma_{\rm in}$ we still get a valid lower bound on $\gamma_{\rm out}$. Let's assume that we have an ``easy'' scenario, where $\alpha\leq10^{-3}$; a direct calculation shows that if $\gamma_{\rm in}\geq 0.2$, then $p(\gamma_{\rm in})>0.23$, $p^{\circ 2}(\gamma_{\rm in})=p(p(\gamma_{\rm in}))>0.26,\ldots,p^{\circ 9}(\gamma_{\rm in})>\frac{5}{6}$. By using \cref{eq:pSuccLower}, similarly to the proof of \cref{prop:distillation_swap_test_high_fidelity} we can see that the expected number of copies for successfully completing the $9$ iterations of the swap test is at most
\begin{align}\label{eq:nineIter}
	2^{9}\prod_{j=0}^{8}\frac{2}{1+(p^{\circ j}(\gamma_{\rm in}))^2}
	=2^{18}\prod_{j=0}^{8}\frac{1}{1+(p^{\circ j}(\gamma_{\rm in}))^2}
	<\frac{2^{18}}{22.6}
	< 11600.
\end{align}
Therefore, one can see that in the $\gamma_{\rm in} = 1-\eta_{\rm in} \leq \frac{3}{4}$ 
regime the iterated swap test still works, but its efficiency degrades rapidly~\cite[Theorems 15 \& 30]{grier2025StreamingQStatePurification}. In fact, even if we assume that all non-principal eigenvalues are the same, the protocol incurs an exponential cost in $1/\gamma_{\rm in}$ because the initial swap tests make only a small increase in $\gamma_{\rm out}$ while they succeed with probability about $\frac12$; see, for example, \refcite[Theorem 9]{childs2025streamingPurification}. Nevertheless, when $\gamma_{\rm in}=1-\eta_{\rm in}$ is lower bounded by a constant we get the desired asymptotically optimal complexity (see \cref{subsec:optimalQPA}), by first magnifying $\gamma_{\rm in}$ to at least $\frac{4}{5}$ as in \cref{eq:RatioEst}--\eqref{eq:nineIter}, and then applying \cref{prop:distillation_swap_test_high_fidelity}, stated formally as follows. 

\begin{proposition}\label{prop:distillation_swap_test_arbitrarily_low_fidelity}
    In the setting of \cref{prop:distillation_swap_test_high_fidelity}, suppose that $\rho_{\rm in}$ has principal eigenvalue $1-\eta_{\rm in}$, where $1-\eta_{\rm in}$ is greater than $\Omega(1)$. Furthermore, suppose that all other eigenvalues of $\rho_{\rm in}$ are bounded above by $\alpha(1-\eta_{\rm in})$ for some constant $\alpha < 1$. Then the expected number of copies consumed and the expected number of swap tests is the same as stated in \cref{prop:distillation_swap_test_high_fidelity}, up to a multiplicative $O(1)$ constant. 
\end{proposition}
\begin{proof}
    This follows from a generalization of the example above to arbitrary $\alpha < 1$. 
\end{proof}

\subsubsection{An asymptotically optimal distillation protocol with simultaneous use of all copies}\label{subsec:optimalQPA}

\refcite{li2024optimalQuantumPurityAmplification} describes a state-agnostic quantum purity amplification protocol based on the Schur transform, which processes all copies in parallel. The authors prove {\cite[Theorem II.3]{li2024optimalQuantumPurityAmplification}} that for a generic quantum state $\rho_{\rm in}$ with principal eigenvector $\ket{\Xi}$, their protocol's sample complexity for outputing a quantum state $\rho_{\rm out}$ such that $\bra{\Xi}\rho_{\rm out}\ket{\Xi}\geq 1-\varepsilon_{\rm dist}$ is asymptotically optimal in the $\varepsilon_{\rm dist}\rightarrow 0$ limit\footnote{This means that for any fixed spectrum $S=(\lambda_1,\lambda_2,\ldots,\lambda_d)$ the optimal sample complexity is $\frac{1}{\varepsilon_{\rm dist}}\sum_{i=2}^{d}\frac{\lambda_i}{(\lambda_1 - \lambda_i)^2} + O(1)$ given that $\rho_{\rm in}$ has spectrum $S$. However, this does not say much about what happens for, say, constant $\varepsilon_{\rm dist}\approx 1$, see \refcite[Appendix D]{li2024optimalQuantumPurityAmplification}.} and their protocol has sample complexity
\begin{equation}\label{eq:SchurCplx}
\underset{\varepsilon_{\rm dist}\rightarrow 0}{\sim}\,\,\,
    	\frac{1}{\varepsilon_{\rm dist}}\sum_{i=2}^{d}\frac{\lambda_i(\rho_{\rm in})}{(\lambda_1(\rho_{\rm in}) - \lambda_i(\rho_{\rm in}))^2} + O(1).
\end{equation}

When $\lambda_1(\rho_{\rm in})=1-\eta_{\rm in}$, the above expression is maximized by $\lambda_2(\rho_{\rm in})=\eta_{\rm in}$, resulting in complexity
\[
\underset{\varepsilon_{\rm dist}\rightarrow 0}{\sim}\,\,\,
\frac{1}{\varepsilon_{\rm dist}}\frac{\eta_{\rm in}}{(1-2\eta_{\rm in})^2} + O(1),
\]
which is rather close to the complexity achieved by \cref{prop:distillation_swap_test_high_fidelity}.

If we only assume that $\lambda_2(\rho_{\rm in})\leq \alpha \lambda_1(\rho_{\rm in})$, then the expression in \cref{eq:SchurCplx} is maximized when all nonzero, non-principal eigenvalues equal $\alpha \lambda_1(\rho_{\rm in})$, giving rise to the complexity expression
\begin{equation}\label{eq:SchurCplxGapped}
\underset{\varepsilon_{\rm dist}\rightarrow 0}{\sim}\,\,\,
\frac{1}{\varepsilon_{\rm dist}}\frac{1-\gamma_{\rm in}}{(1-\alpha)^2\gamma_{\rm in}^2} + O(1).
\end{equation}

As noted in \refcite{li2024optimalQuantumPurityAmplification}, this is exponentially better in the $\gamma_{\rm in}\rightarrow 0$ regime (i.e., low input fidelity) than the iterated swap test protocol described in \cref{sec:distillation_iterated_swap_test}. However, a major drawback of the corresponding protocol of \refcite{li2024optimalQuantumPurityAmplification} is that it requires storing and processing all copies in parallel, resulting in a large space complexity.

The authors of \refcite{grier2025StreamingQStatePurification} note that there is no known protocol that can be applied in a streaming fashion but gets close to the complexity of \cref{eq:SchurCplxGapped} in the $\gamma_{\rm in}\ll 1$ regime. 
In the following subsection we derive such a protocol that uses only two qudits of memory while matching the above asymptotically optimal sample complexity.

\subsubsection{Improved distillation in the regime of small input fidelity via quantum PCA}\label{subsec:quantumPCA}

Now we show that a gate-efficient procedure inspired by quantum principal component analysis (PCA)~\cite{lloyd2013QPrincipalCompAnal,kimmel2016hamiltonian} requires only two qudits plus three qubits of storage to output the top eigenstate with fidelity at least $1-\varepsilon_{\rm dist}$ using
\[
O\left(\Big(\frac{1}{\varepsilon_{\rm dist}}+\frac{1}{\gamma_{\rm in}}\Big)\frac{1-\gamma_{\rm in}}{(1-\alpha)^2\gamma_{\rm in}^2}\right)
\]
copies of $\rho_{\rm in}$ in expectation, which matches the optimal asymptotic complexity of  \cref{eq:SchurCplxGapped} up to a constant factor. 

Intuitively speaking, the additional $1/\gamma_{\rm in}$ term next to $1/{\varepsilon_{\rm dist}}$ comes from the fact that we need to find the top eigenstate within the states stored in memory.
Since the protocol of \refcite{li2024optimalQuantumPurityAmplification} stores and processes all of the required copies in parallel, a tighter analysis might reveal that it performs better in the $\varepsilon_{\rm dist}\gg\gamma_{\rm in}$ regime.
However, once we pay the $1/\gamma_{\rm in}$ price of postselection, we can very efficiently distill further, so the overhead does not multiply with the high-precision-induced $1/\varepsilon_{\rm dist}$ cost. We expect that in the single-pass constant-storage setting, our protocol is essentially optimal, but we leave this problem of optimality as an open question. Finally, we speculate that one may be able to improve this protocol's sample complexity in  the following way: if an attempt of locating the top eigenstate failed, one may reuse the earlier copies that were used for density matrix exponentation in earlier rounds. 

\paragraph{Density matrix exponentiation.} 
Quantum PCA~\cite{lloyd2013QPrincipalCompAnal,kimmel2016hamiltonian} is based on the core primitive of density matrix exponentiation using a fractional swap operation $\exp(- \iUnit \Swap t)=\cos(t) \Id - \iUnit \sin(t) \Swap$,
where $\Id$ denotes the identity operation on a two-qudit system, and $\Swap$ denotes the swap operation of two qudits. Suppose that we have a density operator $\varsigma$ on systems $A$ and $S_1$, and we get a copy of $\varrho$ on system $S_2$ matching the dimension of $S_1$. The Lloyd--Mohseni--Rebentrost (LMR) density matrix exponentiation primitive applies a fractional swap of $S_1$ and $S_2$, then discards $S_2$, and the resulting state can be described as follows~\cite{kimmel2016hamiltonian}:
\begin{align}\label{eq:fracSwap}
	&\Tr_{S_2}\Big[ \big( \Id_A \otimes \exp(-\iUnit \Swap t) \big) \big( \varsigma\otimes\varrho \big) \big( \Id_A \otimes \exp( \iUnit \Swap t)\big) \Big]\\
    ={}&
	\Tr_{S_2}[\cos^2(t)\varsigma\!\otimes\!\varrho+ \iUnit \cos(t)\sin(t)(\varsigma\!\otimes\!\varrho)(\Id_A \!\otimes\! \Swap)-\iUnit \cos(t)\sin(t)(\Id_A \!\otimes\! \Swap )(\varsigma\!\otimes\!\varrho)+\sin^2(t)(\Id \!\otimes\! \Swap)(\varsigma\!\otimes\!\varrho)(\Id_A \!\otimes\! \Swap)]\nonumber\\
    ={}&\cos^2(t)\varsigma+\iUnit \cos(t)\sin(t)\varsigma(\Id_A \otimes \varrho)-\iUnit \cos(t)\sin(t)(\Id_A \otimes \varrho)\varsigma+\sin^2(t)\Tr_{S_1}[\varsigma]\otimes\varrho. \tag{see~\cref{fig:diagramProof}}
\end{align}
This can be viewed as density matrix exponentiation since it represents the map $\varsigma \mapsto (\Id_A \otimes \e^{-\iUnit \varrho t})\varsigma(\Id_A \otimes \e^{\iUnit \varrho t})$ up to corrections of order $O(t^2)$. The procedure consumes one copy of $\varrho$ in order to approximately implement the unitary evolution generated by the Hermitian operator $\varrho$ for a short time $t$. One can also approximately implement controlled density matrix exponentiation $\varsigma \mapsto \big(\Id_A \otimes (\ketbra{0} \otimes \Id + \ketbra{1} \otimes \e^{-\iUnit \varrho' t})\big)\varsigma\big(\Id_A \otimes (\ketbra{0} \otimes \Id +  \ketbra{1} \otimes \e^{\iUnit \varrho' t} )\big)$ by choosing $\varrho = \ketbra{1} \otimes \varrho'$; see \refcite[Appendix C]{kimmel2016hamiltonian}. We will exploit this trick in our protocol below.

\begin{figure}[ht]
	\centering
	\newcommand{\round}{.. controls +(0.3,0) and +(-0.3,0) ..}
	\newcommand{\mround}{.. controls +(-0.3,0) and +(0.3,0) ..}
	\begin{minipage}[t]{0.20\textwidth}
		\centering
		$\Tr_{S_2}[\varsigma\otimes\varrho]\colon$
		
		\vspace{0.2cm}
		\begin{tikzpicture}[thick,
			box/.style = {fill = white, draw = black, inner sep = 0pt, text width = 0.55cm, text height = 0.55cm},
			tallbox/.style = {fill = white, draw = black, inner sep = 0pt, text width = 0.55cm, text height = 1.25cm},
			arc/.style = {start angle = 90, end angle = 0, radius = 0.3cm},
			ell/.style = {x radius = 0.2cm, y radius = 0.1cm}]
			\def\w{0.4cm}
			\def\h{0.7cm}
			\def\d{0.4cm}
			\def\r{0.06cm}
			\path coordinate (0) at (-4.5*\w,2*\h);
			\path coordinate (1) at (-4.5*\w,1*\h);
			\path coordinate (2) at (-4.5*\w,0*\h);
			\path coordinate (3) at (-4.5*\w,-1*\h);	
			
			\path (0)+(0,0.3cm) node {$A$};
			\path (1)+(0,0.3cm) node {$S_1$};
			\path (2)+(0,-0.3cm) node {$S_2$};
			
			\draw (0) -- ++(3*\w,0);
			\draw (1) -- ++(3*\w,0);
			\draw (2) -- ++(3*\w,0);
			\draw (3) -- ++(3*\w,0);
			
			\path (0)+(3*\w,0.3cm) node {$A$};
			\path (2)+(3*\w,\h+0.3cm) node {$S_1$};
			\path (1)+(3*\w,-\h-0.3cm) node {$S_2$};
			
			\path (1)+(1.5*\w,0.5*\h) node[tallbox] {} node {$\varsigma$};
			
			\path (2)+(1.5*\w,0) node[box] {} node {$\varrho$};
			
			\draw (2) +(0*\w,0) arc (90:270:0.5*\h);
			\draw (3) +(3*\w,0) arc (-90:90:0.5*\h);
		\end{tikzpicture}
		\vspace{-0.4cm}
		$$=$$
		\begin{tikzpicture}[thick,
			box/.style = {fill = white, draw = black, inner sep = 0pt, text width = 0.55cm, text height = 0.55cm},
			tallbox/.style = {fill = white, draw = black, inner sep = 0pt, text width = 0.55cm, text height = 1.25cm},
			arc/.style = {start angle = 90, end angle = 0, radius = 0.3cm},
			ell/.style = {x radius = 0.2cm, y radius = 0.1cm}]
			\def\w{0.4cm}
			\def\h{0.7cm}
			\def\d{0.4cm}
			\def\r{0.06cm}
			\path coordinate (0) at (-4.5*\w,2*\h);
			\path coordinate (1) at (-4.5*\w,1*\h);
			
			\path (0)+(0,0.3cm) node {$A$};
			\path (0)+(3*\w,0.3cm) node {$A$};
			\path (1)+(0,0.3cm) node {$S_1$};
			\path (1)+(3*\w,0.3cm) node {$S_1$};
			
			\draw (0) -- ++(3*\w,0);
			\draw (1) -- ++(3*\w,0);
			
			\path (1)+(1.5*\w,0.5*\h) node[tallbox] {} node {$\varsigma$};
		\end{tikzpicture}
	\end{minipage}
	\hfill
	\begin{minipage}[t]{0.23\textwidth}
		\centering
		$\Tr_{S_2}[(\varsigma\otimes\varrho)(\Id_A\otimes \Swap)]\colon$
		
		\vspace{0.2cm}
		\begin{tikzpicture}[thick,
			box/.style = {fill = white, draw = black, inner sep = 0pt, text width = 0.55cm, text height = 0.55cm},
			tallbox/.style = {fill = white, draw = black, inner sep = 0pt, text width = 0.55cm, text height = 1.25cm},
			arc/.style = {start angle = 90, end angle = 0, radius = 0.3cm},
			ell/.style = {x radius = 0.2cm, y radius = 0.1cm}]
			\def\w{0.4cm}
			\def\h{0.7cm}
			\def\d{0.4cm}
			\def\r{0.06cm}
			\path coordinate (0) at (-4.5*\w,2*\h);
			\path coordinate (1) at (-4.5*\w,1*\h);
			\path coordinate (2) at (-4.5*\w,0*\h);
			\path coordinate (3) at (-4.5*\w,-1*\h);	
			
			\path (0)+(0,0.3cm) node {$A$};
			\path (1)+(0,0.3cm) node {$S_1$};
			\path (2)+(0,-0.3cm) node {$S_2$};
			
			\draw (0) -- ++(5*\w,0);
			\draw (1) -- ++(3*\w,0) \round ++(2*\w,-\h);
			\draw (2) -- ++(3*\w,0) \round ++(2*\w,+\h);
			\draw (3) -- ++(5*\w,0);
			
			\path (0)+(5*\w,0.3cm) node {$A$};
			\path (2)+(5*\w,\h+0.3cm) node {$S_1$};
			\path (1)+(5*\w,-\h-0.3cm) node {$S_2$};
			
			\path (1)+(1.5*\w,0.5*\h) node[tallbox] {} node {$\varsigma$};
			
			\path (2)+(1.5*\w,0) node[box] {} node {$\varrho$};
			
			\draw (2) +(0*\w,0) arc (90:270:0.5*\h);
			\draw (3) +(5*\w,0) arc (-90:90:0.5*\h);
		\end{tikzpicture}
		\vspace{-0.4cm}
		$$=$$
		\begin{tikzpicture}[thick,
			box/.style = {fill = white, draw = black, inner sep = 0pt, text width = 0.55cm, text height = 0.55cm},
			tallbox/.style = {fill = white, draw = black, inner sep = 0pt, text width = 0.55cm, text height = 1.25cm},
			arc/.style = {start angle = 90, end angle = 0, radius = 0.3cm},
			ell/.style = {x radius = 0.2cm, y radius = 0.1cm}]
			\def\w{0.4cm}
			\def\h{0.7cm}
			\def\d{0.4cm}
			\def\r{0.06cm}
			\path coordinate (0) at (-4.5*\w,2*\h);
			\path coordinate (1) at (-4.5*\w,1*\h);
			\path coordinate (2) at (-4.5*\w,0*\h);
			\path coordinate (3) at (-4.5*\w,-1*\h);	
			
			\path (0)+(0,0.3cm) node {$A$};
			\path (0)+(5*\w,0.3cm) node {$A$};
			\path (1)+(0,0.3cm) node {$S_1$};
			\path (1)+(5*\w,0.3cm) node {$S_1$};
			
			\draw (0) -- ++(5*\w,0);
			\draw (1) -- ++(5*\w,0);
			
			\path (1)+(1.5*\w,0.5*\h) node[tallbox] {} node {$\varsigma$};
			
			\path (1)+(3.5*\w,0) node[box] {} node {$\varrho$};
		\end{tikzpicture}
	\end{minipage}
	\hfill
	\begin{minipage}[t]{0.23\textwidth}
		\centering
		$\Tr_{S_2}[(\Id_A\otimes \Swap)(\varsigma\otimes\varrho)]\colon$
		
		\vspace{0.2cm}
		\begin{tikzpicture}[thick,
			box/.style = {fill = white, draw = black, inner sep = 0pt, text width = 0.55cm, text height = 0.55cm},
			tallbox/.style = {fill = white, draw = black, inner sep = 0pt, text width = 0.55cm, text height = 1.25cm},
			arc/.style = {start angle = 90, end angle = 0, radius = 0.3cm},
			ell/.style = {x radius = 0.2cm, y radius = 0.1cm}]
			\def\w{-0.4cm}
			\def\h{0.7cm}
			\def\d{0.4cm}
			\def\r{0.06cm}
			\path coordinate (0) at (-4.5*\w,2*\h);
			\path coordinate (1) at (-4.5*\w,1*\h);
			\path coordinate (2) at (-4.5*\w,0*\h);
			\path coordinate (3) at (-4.5*\w,-1*\h);	
			
			\path (0)+(0,0.3cm) node {$A$};
			\path (1)+(0,0.3cm) node {$S_1$};
			\path (2)+(0,-0.3cm) node {$S_2$};
			
			\draw (0) -- ++(5*\w,0);
			\draw (1) -- ++(3*\w,0) \mround ++(2*\w,-\h);
			\draw (2) -- ++(3*\w,0) \mround ++(2*\w,+\h);
			\draw (3) -- ++(5*\w,0);
			
			\path (0)+(5*\w,0.3cm) node {$A$};
			\path (2)+(5*\w,\h+0.3cm) node {$S_1$};
			\path (1)+(5*\w,-\h-0.3cm) node {$S_2$};
			
			\path (1)+(1.5*\w,0.5*\h) node[tallbox] {} node {$\varsigma$};
			
			\path (2)+(1.5*\w,0) node[box] {} node {$\varrho$};
			
			\draw (3) +(0*\w,0) arc (-90:90:0.5*\h);
			\draw (2) +(5*\w,0) arc (90:270:0.5*\h);
		\end{tikzpicture}
		\vspace{-0.4cm}
		$$=$$
		\begin{tikzpicture}[thick,
			box/.style = {fill = white, draw = black, inner sep = 0pt, text width = 0.55cm, text height = 0.55cm},
			tallbox/.style = {fill = white, draw = black, inner sep = 0pt, text width = 0.55cm, text height = 1.25cm},
			arc/.style = {start angle = 90, end angle = 0, radius = 0.3cm},
			ell/.style = {x radius = 0.2cm, y radius = 0.1cm}]
			\def\w{-0.4cm}
			\def\h{0.7cm}
			\def\d{0.4cm}
			\def\r{0.06cm}
			\path coordinate (0) at (-4.5*\w,2*\h);
			\path coordinate (1) at (-4.5*\w,1*\h);	
			
			\path (0)+(0,0.3cm) node {$A$};
			\path (0)+(5*\w,0.3cm) node {$A$};
			\path (1)+(0,0.3cm) node {$S_1$};
			\path (1)+(5*\w,0.3cm) node {$S_1$};
			
			\draw (0) -- ++(5*\w,0);
			\draw (1) -- ++(5*\w,0);
			
			\path (1)+(1.5*\w,0.5*\h) node[tallbox] {} node {$\varsigma$};
			
			\path (1)+(3.5*\w,0) node[box] {} node {$\varrho$};
		\end{tikzpicture}
	\end{minipage}
	\hfill
	\begin{minipage}[t]{0.26\textwidth}
		\centering
		$\Tr_{S_2}[(\Id_A\otimes \Swap)(\varsigma\otimes\varrho)(\Id_A\otimes \Swap)]\colon$
		
		\vspace{0.2cm}
		\begin{tikzpicture}[thick,
			box/.style = {fill = white, draw = black, inner sep = 0pt, text width = 0.55cm, text height = 0.55cm},
			tallbox/.style = {fill = white, draw = black, inner sep = 0pt, text width = 0.55cm, text height = 1.25cm},
			arc/.style = {start angle = 90, end angle = 0, radius = 0.3cm},
			ell/.style = {x radius = 0.2cm, y radius = 0.1cm}]
			\def\w{0.4cm}
			\def\h{0.7cm}
			\def\d{0.4cm}
			\def\r{0.06cm}
			\path coordinate (0) at (-6.5*\w,2*\h);
			\path coordinate (1) at (-6.5*\w,1*\h);
			\path coordinate (2) at (-6.5*\w,0*\h);
			\path coordinate (3) at (-6.5*\w,-1*\h);	
			
			\path (0)+(0,0.3cm) node {$A$};
			\path (1)+(0,0.3cm) node {$S_1$};
			\path (2)+(0,-0.3cm) node {$S_2$};
			
			\draw (0) -- ++(7*\w,0);
			\draw (1) \round ++(2*\w,-\h) -- ++(3*\w,0) \round ++(2*\w,+\h);
			\draw (2) \round ++(2*\w,+\h) -- ++(3*\w,0) \round ++(2*\w,-\h);
			\draw (3) -- ++(7*\w,0);
			
			\path (0)+(7*\w,0.3cm) node {$A$};
			\path (2)+(7*\w,\h+0.3cm) node {$S_1$};
			\path (1)+(7*\w,-\h-0.3cm) node {$S_2$};
			
			\path (1)+(3.5*\w,0.5*\h) node[tallbox] {} node {$\varsigma$};
			
			\path (2)+(3.5*\w,0) node[box] {} node {$\varrho$};
			
			\draw (2) +(0*\w,0) arc (90:270:0.5*\h);
			\draw (3) +(7*\w,0) arc (-90:90:0.5*\h);
		\end{tikzpicture}
		\vspace{-0.4cm}
		$$=$$
		\begin{tikzpicture}[thick,
			box/.style = {fill = white, draw = black, inner sep = 0pt, text width = 0.55cm, text height = 0.55cm},
			tallbox/.style = {fill = white, draw = black, inner sep = 0pt, text width = 0.55cm, text height = 1.25cm},
			arc/.style = {start angle = 90, end angle = 0, radius = 0.3cm},
			ell/.style = {x radius = 0.2cm, y radius = 0.1cm}]
			\def\w{0.4cm}
			\def\h{0.7cm}
			\def\d{0.4cm}
			\def\r{0.06cm}
			\path coordinate (0) at (-4.5*\w,2*\h);
			\path coordinate (1) at (-4.5*\w,1*\h);
			\path coordinate (2) at (-4.5*\w,0*\h);
			\path coordinate (3) at (-4.5*\w,-1*\h);	
			
			\path (0)+(0,0.3cm) node {$A$};
			\path (0)+(3*\w,0.3cm) node {$A$};
			\path (3)+(0,0.3cm) node {$S_1$};
			\path (3)+(3*\w,0.3cm) node {$S_1$};
			
			\draw (0) -- ++(3*\w,0);
			\draw (1) -- ++(3*\w,0);
			\draw (2) -- ++(3*\w,0);
			\draw (3) -- ++(3*\w,0);
			
			\path (1)+(1.5*\w,0.5*\h) node[tallbox] {} node {$\varsigma$};
			
			\path (3)+(1.5*\w,0) node[box] {} node {$\varrho$};
			
			\draw (1) +(0*\w,0) arc (90:270:0.5*\h);
			\draw (2) +(3*\w,0) arc (-90:90:0.5*\h);
		\end{tikzpicture}
	\end{minipage}
	\caption{Diagrammatic representation~\cite{wood2011TNGraphicalCalculus}, and simplification of the four terms in \cref{eq:fracSwap}. For a derivation by direct computation consider that
		$\Tr_{S_2}[(\varsigma\otimes\varrho)(\Id_A\otimes \Swap)]=	\sum_i (\Id_{AS_1}\otimes \bra{i})(\varsigma\otimes\varrho)(\Id_A\otimes \Swap) (\Id_{AS_1}\otimes \ket{i})$ which is 
		$\sum_i (\varsigma(\Id_A\otimes \ket{i}))\otimes(\bra{i}\varrho)=\sum_i \varsigma(\Id_A\otimes \ketbra{i}\varrho) = \varsigma(\Id_A\otimes \varrho)$.
	}
	\label{fig:diagramProof}
\end{figure}

For an efficient implementation of the fractional swap operation,  note that the unitary $(\e^{\iUnit \theta_{\!+} Y}\otimes \Id) \CSwap (\e^{-\iUnit \theta_{\!-} X}\otimes \Id)$ is a block-encoding of the operator $\exp(-\iUnit \Swap t)/2$, with $\theta_{\pm}$ defined below and $\CSwap$ the controlled swap gate with the first register acting as the control; thus, the fractional swap gate can be implemented using $3$ $\CSwap$ gates and $4$ single-qubit gates
\begin{align*}
	\ketbra{0}\otimes (\cos(t) \Id -\iUnit \sin(t) \Swap) & + \ketbra{1}\otimes (\iUnit \sin(t) \Id + \cos(t) \Swap )
    \\
	& = -(\e^{\iUnit\theta_{\!+} Y}\otimes \Id)\, \CSwap\, (\e^{-\iUnit\theta_{\!-} X}Z \e^{\iUnit\theta_{\!-} X}\otimes \Id)\, \CSwap \,(\e^{-\iUnit\theta_{\!+} Y}Z \e^{\iUnit\theta_{\!+} Y}\otimes \Id)\, \CSwap \,(\e^{-\iUnit\theta_{\!-} X}\otimes \Id),\\ \text{where}\quad
	\theta_{\pm} &=\frac{\arccos\left(\frac{\cos (t)-\sin (t)}{2}\right)\pm\arccos\left(\frac{\cos (t)+\sin (t)}{2}\right)}{2},
\end{align*}
which corresponds to one iteration of oblivious amplitude amplification. As elsewhere in the paper, $X$, $Y$, and $Z$ refer to the single-qubit Pauli operators. 

\paragraph{A simple (suboptimal) protocol for the case $\mathbf{\lambda_2(\rho_{\rm in})\ll \lambda_1(\rho_{\rm in})\sqrt{\lambda_1(\rho_{\rm in})\varepsilon_{\rm dist}}}$.}

We now show how to use a simple version of Kitaev's phase estimation~\cite{kitaev1995HiddenSubgroupProblem} to distill the top eigenstate of $\rho_{\rm in}$ when we are promised that $\lambda_1(\rho_{\rm in})\in[\gamma,3\gamma]$, $\varepsilon_{\rm dist}\leq (1-\gamma)$ and $\lambda_2(\rho_{\rm in})\ll \gamma \sqrt{\gamma \varepsilon_{\rm dist}/(1-\gamma)}$, for a known value of $\gamma$. Specifically, the protocol aims to implement the following circuit involving controlled density matrix exponentiation, which may be viewed as phase estimation to one bit of precision. 
\begin{circuit}\label{eq:Kitaev_phase_estimation}
\scalebox{1.0}{
\begin{quantikz}[row sep={2em,between origins}, column sep=0.75em, align equals at=2]
    \lstick{$\ketbra{{+}}$} & &  \ctrl{1}  &  \meter{\ketbra{\pm}} \\
     \lstick{${\rho}_{\rm in}$} & \qwbundle{} & \gate{\e^{-\iUnit \rho_{\rm in} \tau}} & 
     \end{quantikz}
}
\end{circuit}
With the right choice of $\tau$, if density matrix exponentiation is performed in small enough steps $t$, then there is a substantial chance of measuring the first register in $\ketbra{-}$, and when this occurs, the non-principal eigenstates of the input state $\rho_{\rm in}$ are appropriately suppressed in the output.

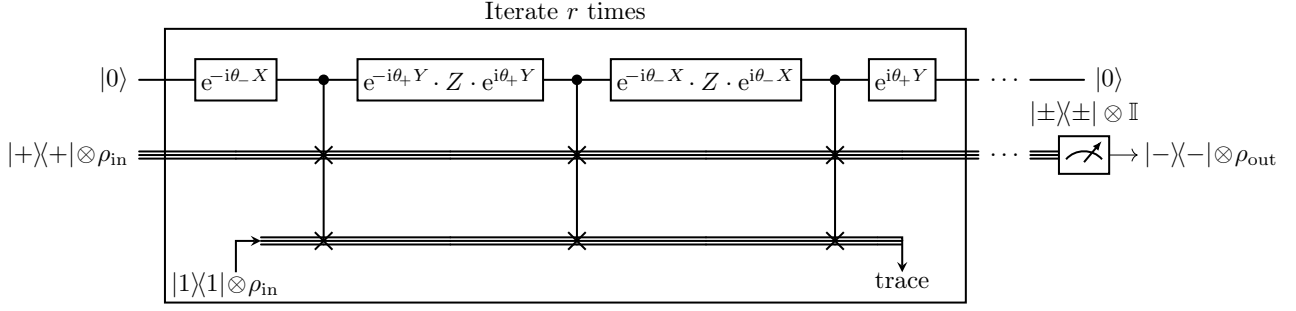
\begin{figure}[ht!]
	\tikzset{
		uptrash/.style={draw=none     , path picture={\draw[internal,inner sep=0pt,-stealth] (path picture bounding box.south) -- (path picture bounding box.center) -- (path picture bounding box.east);},minimum height=2.5em,minimum width=2em}
	}
	\DeclareExpandableDocumentCommand{\uptrash}{O{}{m}}{|[uptrash,label={below:#2},#1]| {} \wireoverride{n}}
    \scalebox{0.95}{	
    \begin{quantikz}[wire types={q,b,b},classical gap=0.5mm, column sep=4mm]
		\lstick{$\ket{0}$} &[2mm] \gate{\!\e^{\!-\iUnit\theta_{\!-} X}\!}\gategroup[3,steps=7,style={inner sep=3mm}]{Iterate $r$ times} & \ctrl{1} & \gate{\!\e^{\!-\iUnit\theta_{\!+} Y\!}\cdot Z \cdot \e^{\iUnit\theta_{\!+} Y}\!} & \ctrl{1} & \gate{\!\e^{\!-\iUnit\theta_{\!-} X\!}\cdot Z\cdot \e^{\iUnit\theta_{\!-} X}\!} & \ctrl{1} & \gate{\!\e^{\iUnit\theta_{\!+} Y}\!} &[2mm] \rule{1mm}{0mm}\ldots\rule{1mm}{0mm} & \rstick{$\ket{0}$} \\
		\lstick{$\!\!\ketbra{+}\!\otimes\! \rho_{\rm in}$}&& \swap{1} && \swap{1} && \swap{1} && \rule{1mm}{0mm}\ldots\rule{1mm}{0mm} & \meter{ \ketbra{\pm} \otimes \Id } \arrow[r] &
		\rstick{$\!\ketbra{-} \! \otimes \! \rho_{\rm{ out}}\!\!$}\setwiretype{n}\\
		&\uptrash{\kern-1.5mm\ketbra{1} \!\otimes \! \rho_{\rm in}\kern1.5mm}& \targX{} && \targX{} && \targX{} &\trash{\text{trace}} 	
	\end{quantikz}
    }
	\caption{A simple procedure based on density matrix exponentiation for extracting the top eigenstate of an unknown density operator $\rho_{\rm in}$. The procedure approximately implements one step of Kitaev's phase estimation of \cref{eq:Kitaev_phase_estimation}. The parameters $\theta_{\pm} =\frac12\arccos\left(\frac{\cos (t)-\sin (t)}{2}\right)\pm\frac12\arccos\left(\frac{\cos (t)+\sin (t)}{2}\right)$ determine the length $t$ of the approximated density matrix evolution-time segment per iteration. Each iteration consumes a fresh copy of $\rho_{\rm in}$ and the first ancilla qubit returns to state $\ket{0}$ after each iteration. After all iterations are completed, the second ancilla qubit is measured in the $\ket{\pm}$ basis and we only accept the $\ket{-}$ outcome.}
	\label{fig:LMRDistill}
\end{figure}

The controlled density matrix exponentiation in \cref{eq:Kitaev_phase_estimation} is approximated with $r=\frac{\tau}{t}$ applications of the LMR procedure, giving the circuit in \cref{fig:LMRDistill}. Namely, consider the effect  of the LMR protocol of \cref{eq:fracSwap} in the special case when the system $A$ is not present, and the protocol is repeatedly applied on the initial state $\varsigma = \smash{\rho^{(0)}} := \ketbra{+} \otimes \rho_{\rm in}$ on system $S_1$, using copies of the amended mixed state $\varrho = \ketbra{1} \otimes \rho_{\rm in}$ on system $S_2$. 
Denote the state (on system $S_1$) after $r$ iterations of the LMR protocol by $\rho^{(r)}$, which can be written as 
\begin{align}
    \rho^{(r)} = \sum_j \sigma^{(r)}_j \otimes  \lambda_j\ketbra{\psi_j},
\end{align}
where $\rho_{\rm in} = \sum_j \lambda_j \ketbra{\psi_j}$ is the eigendecomposition of $\rho_{\rm in}$ and $\sigma^{(r)}_j$ is a normalized single-qubit density operator, for example, $\sigma^{(0)}_j=\ketbra{+}$.
According to \cref{eq:fracSwap}, we find that 
\begin{align*}
	\rho^{(r+1)}&=
	\cos^2(t)\rho^{(r)}+\iUnit \cos(t)\sin(t)\rho^{(r)}(\ketbra{1}\otimes \rho_{\rm in})-\iUnit \cos(t)\sin(t)(\ketbra{1}\otimes \rho_{\rm in})\rho^{(r)}+\sin^2(t) (\ketbra{1} \otimes \rho_{\rm in}),
\end{align*}
in particular $\rho^{(r+1)}=\sum_j \sigma^{(r+1)}_j \otimes \lambda_j\ketbra{\psi_j}$ where
\begin{align}\label{eq:sigmaChannel}
	\sigma^{(r+1)}_j&=
	\cos^2(t)\sigma^{(r)}_j+\iUnit \cos(t)\sin(t)\sigma^{(r)}_j(\lambda_j\ketbra{1})-\iUnit \cos(t)\sin(t)(\lambda_j\ketbra{1})\sigma^{(r)}_j+\sin^2(t)\Tr[\sigma^{(r)}_j]\ketbra{1}.
\end{align}
This means that $\sigma^{(r+1)}_j=\Phi_j[\sigma^{(r)}_j]$ for a quantum channel $\Phi_j$ with Choi matrix 
\begin{align}
	\left(
	\begin{array}{cccc}
		\cos ^2(t) & 0 & 0 & \cos ^2(t)+\iUnit \lambda_j \sin (t) \cos (t) \\
		0 & \sin ^2(t) & 0 & 0 \\
		0 & 0 & 0 & 0 \\
		\cos ^2(t)-\iUnit \lambda_j \sin (t) \cos (t) & 0 & 0 & 1 \\
	\end{array}
	\right),
\end{align}
which is indeed positive semidefinite for all $\lambda_j\in [0,1]$.
The vectorization of the superoperator $\Phi_j[\cdot]$ is 
\begin{align*}
	\left(
	\begin{array}{cccc}
		\cos^2(t) & 0 & 0 & 0 \\
		0 & \cos ^2(t)+\iUnit \sin (t) \cos (t) \lambda_j & 0 & 0 \\
		0 & 0 & \cos^2(t)-\iUnit \sin (t) \cos (t) \lambda_j & 0 \\
		\sin^2(t) & 0 & 0 & 1 \\
	\end{array}
	\right).
\end{align*}
The difference from the desired time-evolution superoperator $\e^{-\iUnit \lambda_j t \ketbra{1}}[\cdot]\e^{\iUnit \lambda_j t \ketbra{1}}$ is 
\begin{align}\label{eq:diffVectorized}
	\left(
	\begin{array}{cccc}
		\sin ^2(t) & 0 & 0 & 0 \\
		0 & \e^{\iUnit \lambda_j t}-\cos^2(t)-\iUnit \cos (t)\sin(t)\lambda_j & 0 & 0 \\
		0 & 0 & \e^{-\iUnit \lambda_j t}-\cos^2(t)+\iUnit \cos (t)\sin(t)\lambda_j & 0 \\
		-\sin ^2(t) & 0 & 0 & 0 \\
	\end{array}
	\right)\,.
\end{align}
For all $t\in[0,\pi/2]$ the operator norm of the matrix in \cref{eq:diffVectorized} is $\sqrt{2}\sin^2(t)$, since 
\begin{align}
	|\e^{\iUnit \lambda_j t}-\cos^2(t)-\iUnit \cos (t)\sin(t)\lambda_j|
	&=\sqrt{(\cos(\lambda_jt)-\cos^2(t))^2+(\sin(\lambda_j t)-\cos(t)\sin(t)\lambda_j)^2}\label{eq:normFun}\\&
	\leq \sqrt{(1-\cos^2(t))^2}\tag{monotonically decreasing in $\lambda_j$}\\&
	=\sin^2(t)\,.\label{eq:normBound}
\end{align}
In the above inequality we used that \cref{eq:normFun} is monotonically decreasing for all $\lambda_j\in[0,1]$ when $t\in[0,\pi/2]$. Indeed, the derivative of the square of \cref{eq:normFun} in $\lambda_j$ can be computed as follows:
\begin{align*}
	2 \cos (t) \left(\underset{=\int_0^t -t\sin(t) \rd t \leq 0}{\underbrace{(t\cos(t)-\sin(t))}} \sin \left(t \lambda _j\right)+\lambda _j \sin (t) \underset{\leq 0}{\underbrace{\left(\sin(t)\cos(t)-t \underset{\geq \cos(t)}{\underbrace{\cos(t\lambda_j)}}\right)}}\right)\leq 0\,.
\end{align*}

Let us now consider superoperator norms induced by Schatten $p$-norms~\cite{watrous2004NotesOnSuperoperatorNorms} for some $p\in [1,\infty]$. By norm conversion and the observation that the operator norm of \cref{eq:diffVectorized} is $\sqrt{2}\sin^2(t)$ due to \cref{eq:normBound}, we directly get that 
\begin{align}
	\lVert \e^{-\iUnit \lambda_j t \ketbra{1}}[\cdot]\e^{\iUnit \lambda_j t \ketbra{1}} - \Phi_j[\cdot] \rVert_{1\text{--}1}\leq \sqrt{2}\lVert \e^{-\iUnit \lambda_j t \ketbra{1}}[\cdot]\e^{\iUnit \lambda_j t \ketbra{1}} - \Phi_j[\cdot] \rVert_{2\text{--}2}=2 \sin^2(t)\leq 2 t^2.
\end{align}
Due to the contractiveness of quantum channels in the Schatten $1$-norm we get that after $r$ iterations
\begin{align}\label{eq:erroGuarantee}
	\lVert\sigma^{(r)}_j-\e^{-\iUnit r\lambda_j t \ketbra{1}}[\sigma^{(0)}_j]\e^{\iUnit r\lambda_j t \ketbra{1}}\rVert_{1}\leq 2rt^2.
\end{align}
\begin{proposition}\label{prop:SimpleQPCADistillation}
	Suppose we are given real numbers $\gamma,\varepsilon_{\rm dist}\in(0,1)$, and copies of a qudit density operator $\rho_{\rm in}$ such that $\lambda_1(\rho_{\rm in})\in[\gamma,3\gamma]$, $\varepsilon_{\rm dist}\leq 1-\gamma$ and $\lambda_2(\rho_{\rm in})\leq \gamma \sqrt{\frac{8 \gamma \varepsilon_{\rm dist}}{3\pi^2(1-\gamma)}}$. Choosing $r=\left\lceil\frac{3\pi^2(1-\gamma)}{2\gamma^3\varepsilon_{\rm dist}}\right\rceil$ and $t=\frac{\pi}{2r\gamma}$ the protocol of \cref{fig:LMRDistill} succeeds with probability at least $\frac{\gamma}{3}$, and upon success produces a state $\rho^{(r)}_-$ such that $[\rho_{\rm in}, \rho^{(r)}_-]=0$, and $\bra{\Xi}\rho^{(r)}_-\ket{\Xi}\geq 1-\varepsilon_{\rm dist}$, where $\ket{\Xi}$ is the principal eigenvector of $\rho_{\rm in}$. The protocol uses an expected number of copies of $\rho_{\rm in}$ which is at most $\frac{3}{\gamma}(r+1)$, and the same order of controlled qudit swap and single-qubit gates. The space complexity is  two qudits and three qubits of storage.
\end{proposition}
\begin{proof}
    The expected complexity bound follows from \cref{fig:LMRDistill} and the success probability bound $\geq \frac{\gamma}{3}$, so it suffices to prove this latter bound. 
	By \cref{eq:erroGuarantee}, we get that projecting down $\rho^{(r)}$ to the $\ketbra{-}$ ancilla state we get the subnormalized state 
	\begin{align}
		(\bra{-} \otimes \Id)\rho^{(r)}(\ket{-} \otimes \Id)&=
		\sum_j \bra{-}\sigma^{(r)}_j\ket{-}  \otimes \lambda_j\ketbra{\psi_j}  \\&
		=\sum_j \lambda_j\ketbra{\psi_j}\Big(\underset{=\frac{|1-\e^{-\iUnit r\lambda_j t}|^2}{4}\leq\frac{|r\lambda_j t|^2}{4}}{\underbrace{|\bra{-}\e^{-\iUnit r\lambda_j t \ketbra{1}}\ket{+}|^2}}+\underset{|\cdot|	\leq rt^2}{\underbrace{\chi_j}}\Big).
	\end{align}
	Let $\rho^{(r)}_-:=( \bra{-}\otimes \Id)\rho^{(r)}(\ket{-}\otimes \Id)/\Tr[(\bra{-} \otimes \Id)\rho^{(r)}(\ket{-} \otimes \Id)]$ be the state we get by postselecting on the $\ket{-}$ outcome of the ancilla measurement in the $\ket{\pm}$ basis. From \cref{eq:sigmaChannel} it is evident that $\rho^{(r)}_-$ commutes with $\rho_{\rm in}$.
	Since $\lambda_1(\rho_{\rm in})\in[\gamma,3\gamma]$ and $rt=\frac{\pi}{2\gamma}$ we get that the probability that we measure the ancilla qubit in the $\ket{-}$ state is
	\begin{align}\label{eq:measLowerSimple}
		\Tr[(\bra{-} \otimes \Id)\rho^{(r)}(\ket{-} \otimes \Id)]&
		\geq \lambda_1\bra{-}\sigma^{(r)}_1\ket{-}
		\geq \gamma\bra{-}\sigma^{(r)}_1\ket{-}\\\text{where }
		\bra{-}\sigma^{(r)}_1\ket{-}&\geq \Big(\frac{|1-\e^{-\iUnit r\lambda_1 t}|^2}{4}-rt^2\Big)
		\geq \Big(\frac{1}{2}-\frac{\pi^2}{4\gamma^2 r}\Big). \tag{since $\lambda_1\in[\gamma,3\gamma]$ and $rt=\frac{\pi}{2\gamma}$}
	\end{align}
	By our choice of $r$ we get that $\frac{\pi^2}{4\gamma^2 r}\leq\frac{\gamma\varepsilon_{\rm dist}}{6(1-\gamma)}\leq\frac{1}{6}$ and (recalling the principal eigenvector $\ket{\Xi}$ is denoted $\ket{\psi_1}$) 
	\begin{align*}
		\bra{\psi_1}\rho^{(r)}_-\ket{\psi_1}
		&=\frac{\lambda_1\bra{-}\sigma^{(r)}_1\ket{-}}{\lambda_1\bra{-}\sigma^{(r)}_1\ket{-}+\sum_{j>1} \lambda_j\underset{\leq \frac{|r\lambda_j t|^2}{4}+rt^2}{\underbrace{\bra{-}\sigma^{(r)}_j\ket{-}}}}\\&
		\geq \frac{\lambda_1\bra{-}\sigma^{(r)}_1\ket{-}}{\lambda_1\underset{\geq\frac13\text{ by }\eqref{eq:measLowerSimple}}{\underbrace{\bra{-}\sigma^{(r)}_1\ket{-}}}+(\underset{\leq 1 - \gamma}{\underbrace{1-\lambda_1}})\Big(\underset{\leq\frac{\gamma\varepsilon_{\rm dist}}{3(1-\gamma)}}{\underbrace{\Big|\frac{\pi\lambda_2}{4\gamma}\Big|^2+\frac{\pi^2}{4\gamma^2 r}}}\Big)}\tag{since $rt=\frac{\pi}{2\gamma}$, $\lambda_2^2\leq \frac{8 \gamma^3 \varepsilon_{\rm dist}}{3\pi^2(1-\gamma)}$}\\&	
		\geq \frac{\lambda_1\bra{-}\sigma^{(r)}_1\ket{-}}{(1+\varepsilon_{\rm dist})\lambda_1\bra{-}\sigma^{(r)}_1\ket{-}}
		=\frac{1}{1+\varepsilon_{\rm dist}}
		\geq 1-\varepsilon_{\rm dist}. \tag*{\qedhere}
	\end{align*}
\end{proof}
	Note that if we only know that $\lambda_1 \geq \gamma$ and $\lambda_2\leq \gamma \sqrt{\frac{8 \gamma \varepsilon_{\rm dist}}{3\pi^2(1-\gamma)}}$, but don't know whether $\lambda_1 \leq 3\gamma$, we can still perform the distillation using the above protocol through combining it with standard techniques, such as exponential search to guess the right order of $\frac{\lambda_1}{\gamma}$. The resulting protocol shall even have the same asymptotic complexity up to constant factors.

\paragraph{An improved protocol for the general case.} We can improve the overhead in the previous protocol by recursively filtering the smaller eigenvalues in a similar fashion to~\refcite{chen2024QSpeedupApxTopEigenvectors}. As we show, it suffices to filter a constant fraction of the unwanted eigenstates in each iteration. This relieves the burden of error magnification due to the postselection on a small probability $\approx \lambda_1$ event hindering the previous simple variant described in the proof of~\cref{prop:SimpleQPCADistillation}. For this purpose, we can use Kitaev's phase estimation~\cite{kitaev1995HiddenSubgroupProblem} combined by standard error reduction techniques, which requires a total simulation time of $\Theta(\frac{\log(1/\epsilon)}{\delta})$ in controlled density matrix exponentiation to output, with probability at least $1-\epsilon$, a phase estimate that is less than $\delta$ off~\cite{mande2023TightQPEBounds}. For practical considerations one might also consider using improved iterative phase estimation variants~\cite{wiebe2015BayesianQPE} or eigenstate filtering techniques via quantum singular value transformation or quantum signal processing~\cite{gilyen2018QSingValTransf,dong2022GroundStatePrepQSP}.

The controlled Hamiltonian simulation can be performed using the same circuit as before (\cref{fig:LMRDistill}), however the subsequent applications of the protocol require a slightly adapted analysis because we need to track the state of multiple qubits and/or the entire measurement history.

\begin{proposition}\label{prop:QPCADistillation}
	Suppose we are given real numbers $\gamma,\varepsilon_{\rm dist},\alpha\in(0,1)$, and can request copies of a qudit density operator $\rho_{\rm in}$ such that $\lambda_1(\rho_{\rm in})\geq\gamma$, $\varepsilon_{\rm dist}< (1-\gamma)$ and $\lambda_2(\rho_{\rm in})\leq \alpha\gamma$. By iteratively applying the circuit of \cref{fig:LMRDistill} with appropriate choices of $r$ and $t$ we can prepare a state $\rho_{\rm out}$ such that $[\rho_{\rm in}, \rho_{\rm out}]=0$, and $\bra{\Xi}\rho_{\rm out}\ket{\Xi}\geq 1-\varepsilon_{\rm dist}$, where $\ket{\Xi}$ is the principal eigenvector of $\rho_{\rm in}$. The protocol uses an expected number of copies of $\rho_{\rm in}$ which is at most $O\left(\frac{1-\gamma}{(1-\alpha)^2\gamma^2}\Big(\frac{1}{\varepsilon_{\rm dist}}+\frac{1}{\gamma}\Big)\right)$, and the same order of controlled qudit swap gates and single-qubit gates. The space complexity is  two qudits and three qubits of storage. 
\end{proposition}
The core of our analysis is to understand what happens when we apply the LMR controlled density matrix exponentiation protocol to a mixed state whose reduced density matrix commutes with $\rho_{\rm in}=\sum_j \lambda_j\ketbra{\psi_j}$, where $\ket{\Xi} =\ket{\psi_1}$. 

\begin{lemma}\label{lem:diagonalFormPreservation}
	Let $A$ label a system of arbitrary finite dimension, and let $C$ label a two-dimensional (qubit) system. Suppose that we have a quantum algorithm that receives as input a normalized state $\rho^{(\rm start)}:=\sum_j \sigma^{(\rm start)}_j\otimes\ketbra{\psi_j}$, where $\sigma^{(\rm start)}_j$ are subnormalized quantum states on the $AC$ register and $\ketbra{\psi_j}$ are the eigenstates of $\rho_{\rm in}$. If the algorithm only interacts with the final register through controlled Hamiltonian simulation $\Id_A\otimes\ketbra{0}_C\otimes\Id +  \Id_A\otimes\ketbra{1}_C\otimes \e^{-\iUnit \rho_{\rm in} \tau_k}$, then the output state can be written as $\rho^{(\rm end)}:=\sum_j\sigma^{(\rm end)}_j\otimes \ketbra{\psi_j}$. Moreover, if the total simulation time is $T=\sum_k\tau_k$ and we approximate each such controlled Hamiltonian simulation step by the LMR protocol with step size at most $t\leq\frac{\pi}{2}$, then the output state can be written as $\tilde{\rho}^{(\rm end)}:=\sum_j \tilde{\sigma}^{(\rm end)}_j\otimes \ketbra{\psi_j}$, where it holds that $\lVert \tilde{\sigma}^{(\rm end)}_j - \sigma^{(\rm end)}_j \rVert_1\leq 3Tt\max(\Tr[\sigma^{(\rm start)}_j],\lambda_j)$.
\end{lemma}

\begin{proof}
	Operations that do not touch the final register clearly preserve the diagonal form $\sum_j  \sigma_j\otimes\ketbra{\psi_j}$ of quantum states. Controlled Hamiltonian simulation can be equivalently described as an operation controlled by eigenstates on the final register, that is,  $\ketbra{0}_C \otimes\Id + \ketbra{1}_C \otimes\e^{-\iUnit \rho_{\rm in} \tau}=\sum_j \e^{\lambda_j\ketbra{1}\tau}\otimes\ketbra{\psi_j}$, which therefore also preserves the diagonal form, proving that we can write $\rho^{(\rm end)}=\sum_j \sigma^{(\rm end)}_j\otimes \ketbra{\psi_j}$.
	
	Now consider what happens when we apply a density matrix exponentiation step of \cref{eq:fracSwap} on $\varsigma = \rho^{(r)}:=\sum_j \tilde{\sigma}^{(r)}_j\otimes\ketbra{\psi_j}$ using mixed state $\varrho = \ketbra{1}_C \otimes \rho_{\rm in}$. According to \cref{eq:fracSwap} we get that
	\begin{align*}
		\rho^{(r+1)}\! &=
		\cos^2(t)\rho^{(r)}\!+\iUnit \cos(t)\sin(t)\rho^{(r)}(\Id_A\kern-0.7mm\otimes\!\ketbra{1}_C\!\otimes\!\rho_{\rm in}) \\
        &\quad\quad -\iUnit \cos(t)\sin(t)(\Id_A\kern-0.7mm\otimes\!\ketbra{1}_C\!\otimes\!\rho_{\rm in})\rho^{(r)}\kern-0.7mm+\sin^2(t)(\tilde{\sigma}^{(r)}\!\!\otimes\!\ketbra{1}_C\!\otimes\!\rho_{\rm in})\,,
	\end{align*}
	where $\tilde{\sigma}^{(r)}$ is a normalized state on the $A$ register defined by
    \begin{align*} \tilde{\sigma}^{(r)}=\Tr_C\left[\sum_j\tilde{\sigma}^{(r)}_j\right]=\sum_j(\Id_A\otimes\bra{0}_C)\tilde{\sigma}^{(r)}_j(\Id_A\otimes\ket{0}_C)+(\Id_A\otimes\bra{1}_C)\tilde{\sigma}^{(r)}_j(\Id_A\otimes\ket{1}_C)\,.
    \end{align*}
    From this we can see that $\rho^{(r+1)}:=\sum_j\tilde{\sigma}^{(r+1)}_j\otimes\ketbra{\psi_j}$ where
	\begin{align}\label{eq:componentMap}
		\tilde{\sigma}^{(r+1)}_j&=
		\underset{\widetilde{\Phi}_j[\tilde{\sigma}^{(r)}_j]:=}{\underbrace{\cos^2(t)\tilde{\sigma}^{(r)}_j+\iUnit \cos(t)\sin(t)\lambda_j\tilde{\sigma}^{(r)}_j(\Id_A\!\otimes\!\ketbra{1}_C)-\iUnit \cos(t)\sin(t)\lambda_j(\Id_A\!\otimes\!\ketbra{1}_C)\tilde{\sigma}^{(r)}_j}}+\sin^2(t)\lambda_j(\tilde{\sigma}^{(r)}\!\otimes\!\ketbra{1}_C).
	\end{align}
	Since $\rho^{(r+1)}\succeq 0$ we immediately get that $\tilde{\sigma}^{(r+1)}_j\succeq 0$. We also get that 
	\begin{align}\label{eq:reqTrBound}
		\lVert\tilde{\sigma}^{(r+1)}_j\rVert_{1}=\Tr[\tilde{\sigma}^{(r+1)}_j]&=\cos^2(t)\Tr[\tilde{\sigma}^{(r)}_j]+\sin^2(t)\lambda_j\leq \max\left(\Tr[\tilde{\sigma}^{(r)}_j],\lambda_j\right)\,.
	\end{align}
	If we start the procedure with $\rho^{(0)}$, consequently by induction we get that after $r$ iterations we have
	\begin{align*}
		\Tr[\tilde{\sigma}^{(r)}_j]\leq \max\left(\Tr[\tilde{\sigma}^{(0)}_j],\lambda_j\right)\,.
	\end{align*}
	 Note that the channel $\widetilde{\Phi}_j$ defined in \cref{eq:componentMap} may be written as $\widetilde{\Phi}_j = \CI_A \otimes \widetilde{\Phi}^{(C)}_j$, where $\CI_A$ is the identity channel on system $A$ and
    \begin{align}
        \widetilde{\Phi}^{(C)}_j[\cdot] = \cos^2(t)[\cdot] + \iUnit \cos(t)\sin(t) \lambda_j [\cdot]\ketbra{1}-\iUnit \cos(t)\sin(t) \lambda_j \ketbra{1}[\cdot]\,.
    \end{align}
    Let us now consider the difference of $\widetilde{\Phi}_j[\cdot]$ and $\CI_A[\cdot]\otimes \e^{-\iUnit \lambda_j t \ketbra{1}}[\cdot]\e^{\iUnit \lambda_j t \ketbra{1}}$. 
    \begin{align}
		\lVert \CI_A[\cdot]\otimes\e^{-\iUnit \lambda_j t \ketbra{1}}[\cdot]\e^{\iUnit \lambda_j t \ketbra{1}} - \widetilde{\Phi}_j[\cdot] \rVert_{1\text{--}1}
		&\leq \lVert \e^{-\iUnit \lambda_j t \ketbra{1}}[\cdot]\e^{\iUnit \lambda_j t \ketbra{1}} - \widetilde{\Phi}^{(C)}_j[\cdot] \rVert_{\diamond}\nonumber\\
        %
        &\leq 2\lVert \e^{-\iUnit \lambda_j t \ketbra{1}}[\cdot]\e^{\iUnit \lambda_j t \ketbra{1}} - \widetilde{\Phi}^{(C)}_j[\cdot]\rVert_{2\text{--}2}\nonumber\\
        &=2 \sin^2(t),\label{eq:diamonBound}
	\end{align}
    where the first inequality follows from the definition of the diamond norm, the second inequality follows from norm conversion from the diamond norm to the $2\text{--}2$ norm for superoperators on 2-dimensional systems \cite[Appendix C]{vanDam2002phdThesis}, and the last equality follows by applying
    \cref{eq:normBound} to evaluate the operator norm of the matrix representation of the superoperator $\e^{-\iUnit \lambda_j t \ketbra{1}}[\cdot]\e^{\iUnit \lambda_j t \ketbra{1}} - \widetilde{\Phi}^{(C)}_j[\cdot]$, given by 
	\begin{align}\label{eq:diffVectorizedGen}
		\left(
		\begin{array}{cccc}
			\sin^2(t) & 0 & 0 & 0 \\
			0 & \e^{\iUnit \lambda_j t}-\cos^2(t)-\iUnit \cos (t)\sin(t)\lambda_j & 0 & 0 \\
			0 & 0 & \e^{-\iUnit \lambda_j t}-\cos^2(t)+\iUnit \cos (t)\sin(t)\lambda_j & 0 \\
			0 & 0 & 0 & \sin^2(t) \\
		\end{array}
		\right)\,.
	\end{align} 
	It follows (dropping the subscripts $A$ and $C$ for brevity) that 
	\begin{align*}
		\lVert\tilde{\sigma}^{(r)}_j - \e^{-\iUnit r\lambda_j t (\Id\otimes\ketbra{1})}\tilde{\sigma}^{(0)}_j\e^{\iUnit r\lambda_j t (\Id\otimes\ketbra{1})} \rVert_1 \leq 3r\sin^2(t)\max\left(\Tr[\tilde{\sigma}^{(0)}_j],\lambda_j\right).
	\end{align*}
	Indeed, this trivially holds for $r=0$, and we can prove it by induction for $r\geq 1$ as follows
	\begin{align*}
		&\lVert\tilde{\sigma}^{(r+1)}_j\! - \e^{\!-\iUnit (r+1)\lambda_j t (\Id\otimes\ketbra{1})}\tilde{\sigma}^{(0)}_j\e^{\iUnit (r+1)\lambda_j t (\Id\otimes\ketbra{1})} \rVert_1 \\
		\leq{}& \lVert\tilde{\sigma}^{(r+1)}_j - \widetilde{\Phi}_j[\tilde{\sigma}^{(r)}_j]\rVert_1
        %
        +
        \lVert\widetilde{\Phi}_j[\tilde{\sigma}^{(r)}_j] - \e^{\!-\iUnit (r+1)\lambda_j t (\Id\otimes\ketbra{1})}\tilde{\sigma}^{(0)}_j\e^{\iUnit (r+1)\lambda_j t (\Id\otimes\ketbra{1})} \rVert_1 \\
		\leq{}& \sin^2(t)\lambda_j\tag{by \cref{eq:componentMap}}
        %
        +
        \lVert\widetilde{\Phi}_j[\tilde{\sigma}^{(r)}_j] - \e^{\!-\iUnit \lambda_j t (\Id\otimes\ketbra{1})}\tilde{\sigma}^{(r)}_j\e^{\iUnit \lambda_j t (\Id\otimes\ketbra{1})} \rVert_1 
        \\
        {}& 
		\phantom{=}
        +
        \! 
        \lVert \e^{\!-\iUnit \lambda_j t (\Id\otimes\ketbra{1})}[\tilde{\sigma}^{(r)}_j\!-\e^{\!-\iUnit r\lambda_j t (\Id\otimes\ketbra{1})}\tilde{\sigma}^{(0)}_j\e^{\iUnit r\lambda_j t (\Id\otimes\ketbra{1})}]\e^{\iUnit \lambda_j t (\Id\otimes\ketbra{1})}\rVert_1 \\
		\leq{}& \sin^2(t)\lambda_j + 2\sin^2(t)\lVert\tilde{\sigma}^{(r)}_j\rVert_1\tag{by \cref{eq:diamonBound}}\\
		&\phantom{=}+\! \lVert \tilde{\sigma}^{(r)}_j\!-\e^{\!-\iUnit r\lambda_j t (\Id\otimes\ketbra{1})}\tilde{\sigma}^{(0)}_j\e^{\iUnit r\lambda_j t (\Id\otimes\ketbra{1})}\rVert_1 \\
		\leq{}& 3(r+1) \sin^2(t)\max\left(\Tr[\tilde{\sigma}^{(0)}_j],\lambda_j\right). \tag{by \cref{eq:reqTrBound} and induction}
	\end{align*}
	Thus, when approximating controlled Hamiltonian simulation $\Id_A\otimes\e^{-\iUnit \tau_k(\ketbra{1}\otimes\rho_{\rm in})}$ with $t'\leq t$ step size we need $r=\frac{\tau_k}{t'}$ steps, inducing an overall error in Schatten 1-norm upper bounded by $3 \frac{\tau_k}{t'}t'^2\max\left(\Tr[\tilde{\sigma}^{(0)}_j],\lambda_j\right)\leq 3\tau_kt\max\left(\Tr[\tilde{\sigma}^{(0)}_j],\lambda_j\right)$. Other steps that do not depend on the final register can be described by quantum channels, which are contractive with respect to the $1$-norm. We can conclude that $\lVert \tilde{\sigma}^{(\rm end)}_j - \sigma^{(\rm end)}_j \rVert_1\leq 3\sum_k \tau_k t\max(\Tr[\sigma^{(\rm start)}_j],\lambda_j)$, as claimed.
\end{proof}

\begin{proof}[Proof of \cref{prop:QPCADistillation}]
	 The protocol will proceed through a sequence of $\ell$ iterations, each of which aims to make progress on amplifying the principal eigenstate, but has some probability of failure. If a round fails, the procedure is restarted from the beginning.  
     First, we establish some properties of the starting and ending state for an individual iteration.
     
     Suppose that at the beginning of an iteration, we have a quantum state $\rho^{(\rm start)}:=\sum_j p^{(\rm start)}_j \ketbra{\psi_j}$. Performing Kitaev's phase estimation with precision $\delta$ and success probability at least $ 1-\epsilon$ requires $T=\Theta(\frac{\log(1/\epsilon)}{\delta})$ controlled Hamiltonian simulation time~\cite{mande2023TightQPEBounds}. Kitaev's phase estimation is accomplished using only 1 ancilla qubit (denoted system $C$), through a sequence of circuits like \cref{eq:Kitaev_phase_estimation} for different values of $\tau=r \cdot t$, which result in a single-bit measurement outcome stored in classical memory. We may equivalently view the system $A$ as initially holding a set of fresh ancilla qubits in the state $\ket{+}$, and for each application of \cref{eq:Kitaev_phase_estimation}, one of these ancilla qubits is swapped into register $C$ and then swapped back after being measured. The swapping in of these fresh qubits and the postprocessing of measurement outcomes is entirely classical and does not incur any additional quantum gates. Viewed this way, it is clear that the full (boosted) Kitaev phase estimation procedure only interacts with the register holding $\rho^{(\rm start)}$ through controlled Hamiltonian simulation and thus, when the controlled Hamiltonian simulation is implemented via the LMR density matrix exponentiaton protocol,   \cref{lem:diagonalFormPreservation} applies. 
    If the stepsize is  $t=\frac{\zeta}{3T} < \pi/2$, then a total of $\Theta\Big(\frac{\log^2(1/\epsilon)}{\delta^2\zeta}\Big)$ copies of $\rho_{\rm in}$ are consumed via the protocol of~\cref{fig:LMRDistill}. The gate complexity is 4 single-qubit gates and 3 controlled qudit swap gates per copy consumed (with no additional gate complexity associated to the input state); this tells us that it suffices to bound the number of consumed copies. We will set the precision of the phase estimation to $\delta = \frac{(1-\alpha)\gamma}{2}$, and say that the round succeeds if the energy estimate register (a classical register) is greater than $\frac{1+\alpha}{2}\gamma$ (in order to distinguish $\lambda_1 \geq \gamma$ from $\lambda_2 \leq \alpha\gamma$ with failure probability at most $\epsilon$). 
	Let $\rho^{(\rm proj)}=\sum_j p^{(\rm proj)}_j \ketbra{\psi_j}$ be the subnormalized output state of the LMR-based algorithm when projected down to the energy estimate register  being greater than $\frac{1+\alpha}{2}\gamma$ and tracing out all ancilla registers. 
	If $p_1^{(\rm start)}\geq\lambda_1/2$, then according to \cref{lem:diagonalFormPreservation}, we get that 
    \begin{align}\label{eq:pProjLower}
		p^{(\rm proj)}_1 &\geq p^{(\rm start)}_1(1-\epsilon- 3Tt)\geq p^{(\rm start)}_1(1-\epsilon-\zeta)\,,\\
		\text{and for }j>1, \qquad \quad p^{(\rm proj)}_j &\in [0,\epsilon p^{(\rm start)}_j+3Tt\max(p^{(\rm start)}_j,\lambda_j)/2)]\subseteq [0,(\epsilon+\frac\zeta2)p^{(\rm start)}_j+\frac\zeta2\lambda_j].\label{eq:pProjUpper}
	\end{align}
	The probability that the round succeeds is denoted by
    \begin{align}
        p^{(\rm succ)} = \sum_j p_j^{(\rm proj)}
    \end{align}
    and the output state conditioned on success is denoted by
    \begin{align}
        \rho^{(\rm end)} = \frac{\rho^{(\rm proj)}}{p^{(\rm succ)}} := \sum_j p^{(\rm end)}_j \ketbra{\psi_j}
    \end{align}
    From \cref{eq:pProjLower}, this state satisfies 
\begin{align}\label{eq:succProBound}
		p^{(\rm end)}_1=\frac{p^{(\rm proj)}_1}{p^{(\rm succ)}}\geq \frac{1-\epsilon-\zeta}{p^{(\rm succ)}}p^{(\rm start)}_1, 
	\end{align}
	This equation shows how performing (approximate) phase estimation and postselecting on energy estimates above $\frac{1+\alpha}{2}\gamma$ can lead to magnification of the overlap $p_1$ with the principal eigenvector, depending on parameters $\epsilon$ and $\zeta$. 

	The full protocol iterates this process over $\ell$ rounds, indexed by $i=1,2,\ldots, \ell$. Let $\rho^{({\rm start},i)}$ denote the starting state for iteration $i$  and $\rho^{({\rm end},i)}$ the ending state, which is taken to be the starting state $\rho^{({\rm start},i+1)}$ at the next iteration. For brevity of notation, let $\rho^{(i)} = \rho^{({\rm end},i)}$, with the starting state at iteration $i=1$ taken to be $\rho^{(0)}= \rho^{({\rm start},1)}=\rho_{\rm in}$. In the $i$-th iteration, we set the failure probability in Kitaev's phase estimation to $\epsilon_i$ and the stepsize to $t_i=\frac{\zeta_i}{3T}$, which determines the success probability $p^{({\rm succ},i)}$ of the $i$th iteration. The overlap with the principal eigenvector---the key figure of merit---after the $i$-th round is denoted $p^{(i)}_1 \equiv p^{({\rm end}, i)}_1 \equiv p^{({\rm start}, i+1)}_1 :=\bra{\psi_1}\rho^{(i)}\ket{\psi_1}$. As long as we maintain $p_1^{( i)}\geq\lambda_1/2$, the bound in \cref{eq:succProBound} holds, and for any value of $i$, we derive the following bound on the inverse overall success probability of all iterations $i+1, i+2,\ldots,\ell$:
	\begin{align}\label{eq:succProd}
		1\geq p^{( \ell)}_1\geq \prod_{k=i+1}^{\ell}\frac{1-\epsilon_{k}-\zeta_{k}}{p^{({\rm succ},k)}}p^{(i)}_1
		\Longrightarrow  
		\prod_{k=i+1}^{\ell}\frac{1}{p^{({\rm succ},k)}}\leq \frac{1}{p^{(i)}_1} \prod_{k=i+1}^{\ell}\frac1{1-\epsilon_{k}-\zeta_{k}}
		\leq \frac{1}{p^{(i)}_1} \frac1{1-\sum_{k=i+1}^{\ell}(\epsilon_{k}+\zeta_{k})}.
	\end{align}
	On the other hand, note that by \cref{eq:pProjLower} we have that for all $i=1,\ldots,\ell$,
	\begin{align}
		p^{(i)}_1
		= \frac{p^{({\rm proj},i)}_1}{p^{({\rm proj},i)}_1+\sum_{j>1} p^{({\rm proj},i)}_j}
		&\geq \frac{p^{({\rm start},i)}_1(1-\epsilon_i-\zeta_i)}{p^{({\rm start},i)}_1(1-\epsilon_i-\zeta_i)+\sum_{j>1} p^{({\rm proj},i)}_j}\tag{by monotonicity of $\frac{x}{x+c}$ and \cref{eq:pProjLower}}\\&		
		\geq \frac{p^{({\rm start},i)}_1}{p^{({\rm start},i)}_1+\frac{\epsilon_i+\zeta_i}{1-\epsilon_i-\zeta_i}} = \frac{p^{(i-1)}_1}{p^{(i-1)}_1+\frac{\epsilon_i+\zeta_i}{1-\epsilon_i-\zeta_i}}.\label{eq:smallPBound}	
	\end{align}
    where in the last line we have used that $\sum_{j>1} p_j^{({\rm proj},i)} \leq \epsilon_i+\zeta_i$ as a consequence of \cref{eq:pProjUpper}.
	Let $c:=\frac{\epsilon_i+\zeta_i}{1-\epsilon_i-\zeta_i}$ and note that the function $\frac{p}{p+c}$ in \cref{eq:smallPBound} is monotone increasing on the interval $p\in[0,1]$ for all fixed $c\geq 0$, that is, when $\epsilon_i+\zeta_i<1$; and also monotone decreasing in $c>0$ for all fixed $p\in[0,1]$. For $p=\frac{1}{3}$ and $\epsilon_i+\zeta_i=\frac{1}{8}$ it evaluates to $\frac{p}{p+c}=\frac7{10}>\frac23$, so we can conclude that 
    if $\epsilon_i+\zeta_i\leq \frac{1}{8}$ and $p^{(i-1)}_1\geq \frac{1}{3}$, then $p^{(i)} >\frac{2}{3}$.
	If $p^{(i-1)}_1\leq\frac{1}{3}$ and $\epsilon_i+\zeta_i\leq \frac{1}{8}$, then we also have that $p^{(i)} \geq 2 p^{(i-1)}_1$, so that $p_1^{(0)}\geq\min(\lambda_1,\frac12)$ implies that $p_1^{(i)}\geq\lambda_1/2$ holds for all $i\geq 0$.

	\paragraph{The $\varepsilon_{\rm dist}>\frac{1}{3}$ case.} Note that this is a subset of the case $\gamma<\frac{2}{3}$.
	Here we will conveniently set $\zeta_i:=\epsilon_i\leq\frac{1}{16}$ so that starting with $\rho^{(0)}:=\rho_{\rm in}$, for which $p_1^{(0)} \geq \gamma$ by assumption, we get that $\ell:=\lceil\log_2(\frac{1}{3\gamma})\rceil+1$ successful iterations of the procedure suffice to achieve $p^{(\ell)}_1>\frac{2}{3}$. This follows by noting that $p^{(i)}_1\geq 2^i\gamma$ holds for every $i\in\{0,1,\ldots,\ell-1\}$, a fact that can be shown by induction and \cref{eq:smallPBound}, which implies $p^{(\ell-1)}\geq 1/3$, and thus via the observation above also that $p^{(\ell)} \geq 2/3$.	
	
	We choose the failure probability and step size parameters to decrease with $i$ as $\epsilon_{i}=\zeta_i:=\frac{2^{(1-i)/2}}{16}$, so that $\sum_{i=1}^\ell\epsilon_i = \sum_{i=1}^\ell\zeta_i \leq \frac{1}{4}$, and \cref{eq:succProd} gives $\prod_{k={i+1}}^{\ell}\frac{1}{p^{({\rm succ},k)}}\leq \frac{2}{p^{(i)}_1}$.
	
	Let $C^{(i)}$ denote the expected copy  complexity to successfully complete the procedure up to and including the $i$-th iteration. We have $C^{(0)}=1$ and $C^{(i+1)}=\frac{C^{(i)}+\Theta\Big(\frac{\log^2(1/\epsilon_i)}{\delta^2\epsilon_i}\Big)}{p^{({\rm succ},i)}}$, from which we can see by induction that 
	\begin{align*}
		C^{(\ell)}
		\!=\Theta\!\left(\prod_{i=1}^{\ell}\frac{1}{p^{({\rm succ},i)}}\!+\!\sum_{i=1}^\ell\frac{\log^2(1/\epsilon_i)}{\delta^2\epsilon_i}\prod_{k=i}^{\ell}\frac{1}{p^{({\rm succ},k)}}\!\right)
		\!\leq O\!\left(\frac{1}{p^{(0)}_1}\!+\!\sum_{i=1}^\ell\frac{\log^2(1/\epsilon_i)}{\delta^2\epsilon_i p_1^{(i-1)}}\!\right)
		\!\leq O\!\left(\!\frac{1}{\gamma}\!+\!\sum_{i=1}^\ell\frac{i^2 2^{-i/2}}{\delta^2 \gamma}\!\right)
		\!=O\!\left(\frac{1}{\delta^2\gamma}\right).
	\end{align*}
	We can conclude that for $\gamma<\frac{2}{3}$ we can prepare a state $\rho_{\rm out} = \rho^{(\ell)}$ such that $\bra{\psi_1}\rho_{\rm out}\ket{\psi_1}>\frac23$ with an expected $O(\frac{1}{\delta^2\gamma})=O(\frac{1}{(1-\alpha)^2\gamma^3})$ copy complexity.
	
	\paragraph{The $\varepsilon_{\rm dist} \leq \frac{1}{3}$ case.}
	If $\gamma<\frac{2}{3}$ we first run the above protocol which produces a state with $p_1\geq\frac{2}{3}$, that we will call $\rho^{(0)}$ for the purposes of this subcase analysis. If $\gamma\geq\frac{2}{3}$ we can simply set $\rho^{(0)}$ to be $\rho_{\rm in}$. In both cases, we can say that the preparation of $\rho^{(0)}$ has complexity $C=O(\frac{1}{(1-\alpha)^2\gamma^3})$, since in the $\gamma\geq\frac{2}{3}$ case, it holds that $(1-\alpha)\geq \frac12$ without loss of generality.

	Let $\eta^{(i)}\equiv \eta^{({\rm end},i)}:=1-p^{({\rm end},i)}_1 \equiv 1-p^{(i)}_1$ and $\eta^{({\rm start},i)}:=1-p^{({\rm start},i)}_1 \equiv 1-p^{(i-1)}$.
	Similarly to \cref{eq:smallPBound}, by \cref{eq:pProjLower,eq:pProjUpper} we get for $i=1,\ldots,\ell$ that
	\begin{align*}
		\eta^{(i)}
		= 1-\frac{p^{({\rm proj},i)}_1}{p^{({\rm proj},i)}_1+\sum_{j>1} p^{({\rm proj},i)}_j}
		&\leq 1- \frac{p^{({\rm start},i)}_1(1-\epsilon_i-\zeta_i)}{p^{({\rm start},i)}_1(1-\epsilon_i-\zeta_i)+\sum_{j>1} p^{({\rm proj},i)}_j}\tag{by monotonicity of $\frac{x}{x+c}$ and \cref{eq:pProjLower}}\\&		
		\leq1- \frac{p^{({\rm start},i)}_1(1-\epsilon_i-\zeta_i)}{p^{({\rm start},i)}_1(1-\epsilon_i-\zeta_i)+(\epsilon_i+\zeta_i)(1-p^{({\rm start},i)}_1) + \zeta_i(1-\lambda_1)}\tag{by \cref{eq:pProjUpper}}\\&
		= \frac{(\epsilon_i+\zeta_i)\eta^{({\rm start},i)} + \zeta_i(1-\lambda_1)}{(1-\eta^{({\rm start},i)})(1-\epsilon_i-\zeta_i)+(\epsilon_i+\zeta_i)\eta^{({\rm start},i)} + \zeta_i(1-\lambda_1)}\nonumber\\&
		\leq \frac{(\epsilon_i+\zeta_i)(\eta^{({\rm start},i)}) + \zeta_i(1-\lambda_1)}{(1-\eta^{({\rm start},i)})(1-\epsilon_i-\zeta_i)} 
        =
         \frac{(\epsilon_i+\zeta_i)(\eta^{(i-1)}) + \zeta_i(1-\lambda_1)}{(1-\eta^{(i-1)})(1-\epsilon_i-\zeta_i)}\,.
	\end{align*}
	Therefore,
	\begin{align}\label{eq:smallIter}
		\eta^{(i)}\leq \frac{\eta^{(i-1)}+\zeta_i(1-\lambda_1)}{4} \quad \text{ as long as }\quad \eta^{(i-1)}\leq\frac13 \text{ and } \epsilon_i+\zeta_i\leq \frac18.
	\end{align}
	
	We choose $\ell:=\lceil\log_2(\frac{1-\gamma}{\varepsilon_{\rm dist}})\rceil$, $\zeta_i:=2^{-3-i}$, and $\epsilon_{i}:=2^{-4-\ell+i}$. 
	Since $\eta^{(0)}\leq 1-\gamma$, using \cref{eq:smallIter} and induction we see that $\eta^{(i)}\leq 2^{-i}(1-\gamma)$ and therefore $\eta^{(\ell)}\leq \varepsilon_{\rm dist}$.
	
	As $\sum_{i=1}^\ell\epsilon_i+\zeta_i\leq \frac{1}{4}$, we get by \cref{eq:succProd} that the sequence of iterations succeeds with probability $\prod_{i=1}^{\ell} p^{({\rm succ}, i)} \geq \frac{3}{4}p^{(0)}_1 \geq \frac12$. Therefore, the expected complexity can be bounded up to a constant factor by
	\begin{align*}
		C+\sum_{i=1}^\ell\frac{\log_2^2(1/\epsilon_i)}{\delta^2\zeta_i}
		=C+\sum_{i=1}^\ell\frac{2^{i+3}(4+\ell-i)^2}{\delta^2}
		=C+\frac{2^{\ell+3}}{\delta^2}\sum_{i'=0}^{\ell-1}2^{-i'}(i'+4)^2
		=O\left(\frac{1}{(1-\alpha)^2\gamma^2}\Big(\frac{1}{\gamma}+\frac{1-\gamma}{\varepsilon_{\rm dist}}\Big)\right).\tag*{\qedhere}
	\end{align*}
\end{proof}

\subsection{Resource state teleportation}\label{sec:teleportation_detailed}

Once we have prepared the resource state $\ketbra{\ol{\Psi(g)}}$ (up to low error), the next step is to consume this resource state to enact a transformation on the encoded address qubits. Suppose the $n$-qubit encoded register on which we want to apply the operation $\ol{V(g)}$ is in an arbitrary pure state
\begin{align}
    \ket{\ol{\alpha}} = \sum_{w \in \{0,1\}^{n}}\alpha_{w} \ket{\ol{w}}\,, \qquad \qquad \sum_{w} |\alpha_w|^2 = 1\,.
\end{align}
The teleportation procedure calls for applying a set of $n$ disjoint logical CNOT gates, each controlled by one of the logical qubits in the register holding $\ket{\ol{\alpha}}$ and targeting one of the logical qubits in the resource state, and then measuring the resource state. This was depicted in \cref{eq:teleportation_circuit_large}, which is reproduced more compactly as 
\begin{circuit}\label{eq:teleportation_figure}
\tikzset{
    operator/.style={draw, filling, inner sep=2pt, thickness, align=center, baseline=(current bounding box.center)},
    internal/.style={thickness, line width=\thickWire}
}
    \scalebox{1.0}{
    \begin{quantikz}[
        row sep={0em,between origins}, 
        column sep={0em, between origins}, 
        align equals at=1.5, 
        wire types = {n,n}
    ]
    \lstick{$\ketbra{\ol{\alpha}}$} &  \wire[r][4][line width = \thickWire]{q} &[2em]  \qwbundle{n} &[1em] \ctrl[style = {line width = \thickWire, draw=\tColor, fill=\tColor}]{1} &[2.5em] &[3em] \rstick{$\ol{V(g^{\oplus m})}\ketbra{\ol{\alpha}}\ol{V(g^{\oplus m})}$}  \\[3em]
    \lstick{$\ketbra{\ol{\Psi(g)}}$} & \wire[r][3][line width = \thickWire]{q} &  \qwbundle{n} & \targ[style={line width = \thickWire, draw=\tColor}, internal/.append style={line width=\thickWire}]{} &  \meter[style={line width = \thickWire,draw=\tColor}]{} &\wire[l][1]{c} \rstick{$m$} \\[1em]
    \end{quantikz}
    }
\end{circuit}
On the output label of the circuit, we have used the notation  $g^{\oplus m}$ to represent the Boolean function defined for any $x$ by the equation
\begin{align}\label{eq:g_oplus_m}
    g^{\oplus m}(x) = g(x \oplus m)\,.
\end{align}
To verify the circuit's output and the assertions in \cref{sec:teleporting_the_QRAM_gate}, note that the CNOT gates enact the transformation 
\begin{align}
    \ket{\ol{\alpha}}\ket{\ol{\Psi(g)}} = \frac{1}{2^{n/2}}\sum_{w,z \in \{0,1\}^{n}}(-1)^{g(z)}\alpha_{w} \ket{\ol{w}}\ket{\ol{z}}
    \quad\longmapsto&\quad 
    \frac{1}{2^{n/2}}\sum_{w,z \in \{0,1\}^{n}}(-1)^{g(z)}\alpha_{w} \ket{\ol{w}}\ket{\ol{w \oplus z}} \\
    &\quad = \frac{1}{2^{n/2}}\sum_{w,m \in \{0,1\}^{n}} (-1)^{g(w\oplus m)}\alpha_{w} \ket{\ol{w}}\ket{\ol{m}}\,.
\end{align}
A logical measurement is performed on the resource state register, obtaining outcome $m$. Regardless of the state $\ket{\ol{\alpha}}$ and the function $g$, the distribution over measurement outcomes $m$ is the uniform distribution, and the post-measurement state on the first register is
\begin{align}
    \ol{V(g^{\oplus m})} \ket{\ol{\alpha}} =  \frac{1}{2^{n/2}}\sum_{w \in \{0,1\}^n} (-1)^{g^{\oplus m}(w)}\alpha_{w} \ket{\ol{w}}=\frac{1}{2^{n/2}}\sum_{w \in \{0,1\}^n} (-1)^{g(w\oplus m)}\alpha_{w} \ket{\ol{w}} \,.
\end{align}
Examining the definition of $g^{\oplus m}$, we can see that it is defined by the same underlying table of data that defines $g$; however, the addresses to which the data items have been assigned are scrambled by adding (modulo 2) the random measurement outcome $m$ to each address. 
Nevertheless, we may conclude that the teleportation procedure successfully inserts information about $g$ into the phases of the state $\ket{\ol{\alpha}}$, albeit in an obfuscatedured way, due to the fact that the phase applied is $(-1)^{g(w\oplus m)}$ rather than the desired $(-1)^{g(w)}$, where $m$ is uniformly random. 

We now wish to understand what happens in the teleportation step if we use the state $\ol{\phi(g)}_{\rm dist}$ (discussed in \cref{sec:distillation}) as our resource state, which is an imperfect approximation to $\ketbra{\ol{\Psi(g)}}$. First, we define the ideal teleportation channel $\ol{\CT(g)}$ as the procedure that perfectly prepares $\ketbra{\ol{\Psi(g)}}$, applies the logical CNOTs, and measures the resource register (equivalently, applies a completely dephasing channel). This has the action
\begin{align}
    \ol{\CT(g)}[\rho] = \frac{1}{2^{n}} \sum_{m \in \{0,1\}^{n}} \ol{\CV(g^{\oplus m})}[\rho] \otimes \ketbra{m} \,,
\end{align}
where we have used the channel definition $\CV(h)[\rho] = V(h) \, \rho \, V(h)^\dag$ with the overlined version $\ol{\CV(h)}$ denoting the logical version of the channel. 

We may now argue that as long as the error of $\ol{\phi(g)}_{\rm dist}$ is low, relative to $\ketbra{\ol{\Psi(g)}}$, the channel enacted is close to $\ol{\CT(g)}$. 
\begin{proposition}[Teleportation with approximate resource state]\label{prop:teleportation}
    Let $\ol{\CT(g)}_{\rm appr}$ be the channel realized by using the state $\ol{\phi(g)}_{\rm dist}$ in place of $\ketbra{\ol{\Psi(g)}}$ in the teleportation protocol depicted in \cref{eq:teleportation_figure}. Then, we have
    \begin{align}
        \frac{1}{2}\nrm{\ol{\CT(g)} - \ol{\CT(g)}_{\rm appr}}_{\diamond} \leq \frac{1}{2} \nrm{\ketbra{\ol{\Psi(g)}}-\ol{\phi(g)}_{\rm dist}}_1\,,
    \end{align}
    where $\nrm{\cdot}_\diamond$ indicates the diamond norm distance between channels. 
\end{proposition}
\begin{proof}
    Define the quantum operation $\ol{\CC}$ as the operation on $2n$ logical qubits enacted by the CNOT gates and measurements (i.e., completely dephasing channel) in \cref{eq:teleportation_figure}
    Thus, for arbitrary $n$-qubit state $\ol{\sigma}$ in the codespace of the first register, we have
    \begin{align}
        \ol{\CT(g)} [\ol{\sigma}] &= \ol{\CC}\left[\ol{\sigma} \otimes \ketbra{\ol{\Psi(g)}}\right]\\
        \ol{\CT(g)}_{\rm appr}[\ol{\sigma}] &= \ol{\CC}\left[\ol{\sigma} \otimes \ol{\phi(g)}_{\rm dist}\right]\,.
    \end{align}
    Now, introduce an environment register  of $a$ qubits, and consider an arbitrary state $\ol{\rho}$ on $a+n$ logical qubits (whose reduced density matrix after tracing out the environment is in the codespace of the $n$-logical-qubit code). Let $\CI_{\ol{E}}$ be the identity channel on the environment register. The diamond norm distance is defined as the supremum (over $\ol{\rho}$ for all possible sizes $a$ of the environment) in the trace distance between the action of the two channels
    \begin{align}
        \frac{1}{2} \lVert \ol{\CT(g)} - \ol{\CT(g)}_{\rm appr}\rVert_{\diamond} &= \max_{a}\sup_{\ol{\rho}}\frac{1}{2}\nrm{(\CI_{\ol{E}} \otimes \ol{\CT(g)})[\ol{\rho}]-(\CI_{\ol{E}} \otimes \ol{\CT(g)}_{\rm appr})[\ol{\rho}]}_1 \label{eq:teleportation_difference_with_environment}\\
        &= \max_a \sup_{\ol{\rho}}\frac{1}{2} \left\lVert(\CI_{\ol{E}} \otimes \ol{\CC})\left[\ol{\rho}\otimes \left(\ketbra{\ol{\Psi(g)}}-\ol{\phi(g)}_{\rm dist}\right)\right]\right\rVert_1\\
        &\leq \max_a \sup_{\ol{\rho}}\frac{1}{2} \nrm{\ol{\rho} \otimes \left(\ketbra{\ol{\Psi(g)}}-\ol{\phi(g)}_{\rm dist}\right)}_1 \\
        &= \max_a \sup_{\ol{\rho}}\frac{1}{2} \nrm{\ketbra{\ol{\Psi(g)}}-\ol{\phi(g)}_{\rm dist}}_1 =
        \frac{1}{2} \nrm{\ketbra{\ol{\Psi(g)}}-\ol{\phi(g)}_{\rm dist}}_1\,,
    \end{align}
    where the inequality follows from monotonicity of the trace distance under quantum channels. This completes the proof. 
\end{proof}

\subsection{Adaptive correction and classical update rule}\label{sec:adaptive_correction}

Recall that the goal is to implement the diagonal logical QRAM unitary $\ol{V(f)}$, given a data table (Boolean function) $f$. The previous section established that for any function $g$, the teleportation channel $\ol{\CT(g)}$ receives a uniformly random measurement outcome $m$, and then, conditioned on $m$, enacts the unitary transformation $\ket{\ol{\alpha}} \mapsto \ol{V(g^{\oplus m})}\ket{\ol{\alpha}}$.  

Our protocol repeats the teleportation process over a sequence of $n$ rounds. In each round, the function $g$ will be different, and the measurement outcome $m$ will be sampled independently and uniformly at random. The value of $g$ used in round $j$ will be denoted $g^{(j)}$, and the value of $m$ denoted by $m^{(j)}$. 

We begin in round 1 setting $g = g^{(1)} = f$, receiving measurement outcome $m = m^{(1)} \in \{0,1\}^{n}$, and enacting logical unitary $\ol{V(g^{\oplus m})}$. The key observation is that $\ol{V(g^{\oplus m})}^2 = \ol{\Id}$, the logical identity operation, and hence
\begin{align}
    \ol{V(g)} = \ol{V(g)}\;\ol{V(g^{\oplus m})}\;\ol{V(g^{\oplus m})} = \ol{V(g \oplus g^{\oplus m})} \;\ol{ V(g^{\oplus m})}\,,
\end{align}
where we have invoked the composition rule $V(h)V(h') = V(h \oplus h')$. Thus, given that we have already implemented $\ol{V(g^{\oplus m})}$, we must now implement the \textit{correction unitary} $\ol{V(g \oplus g^{\oplus m})}$, which is also a member of the family of diagonal (logical) QRAM operators. We thus compute a new Boolean function by the \textit{classical update rule} $\UR$, which takes as input a data table $g$ and a bit string $m \in \{0,1\}^n$ and outputs a new data table
\begin{align}\label{eq:update_rule}
    \UR(g,m) = g \oplus g^{\oplus m}\,,
\end{align}
which is depicted as a circuit as
\begin{circuit}\label{eq:circuit_update_rule}
        \scalebox{1.0}{
    \begin{quantikz}[
        row sep={0em,between origins}, 
        column sep={0em, between origins}, 
        align equals at=1.5, 
        wire types = {n,c}]
    &[1.5em] &[1em] \lstick{$m$}&[0em] \wire[r][2]{c} &[1.5em] \qwbundle{n}& [1em] \ctrl[vertical wire=c]{1}  &[2.5em] &[2.5em] \\[3em]
    \lstick{$g$} & \qwbundle{2^n} & & & & \urGate & \qwbundle{2^n} & \rstick{$g \oplus g^{\oplus m}$}
    \end{quantikz}
        }
\end{circuit}
and we update $g \gets g^{(2)} = \UR(g^{(1)},m^{(1)})$ to use for round 2.  Consequently, the correction unitary is now equal to $\ol{V(g)}$ and the goal of round $2$ is again simply to implement the unitary $\ol{V(g)}$, just as it was in round 1.  As before, we receive a new measurement outcome $m = m^{(2)}$, we compute $g^{(3)} = \UR(g^{(2)},m^{(2)})$, and we update $g \gets g^{(3)}$ for round 3. 

We then iterate this process a number of times. We note that, if after applying the update rule we ever obtain $g = \mathbf{0}$, the zero function, then we may terminate the procedure, because the correction unitary will be $\ol{V(\mathbf{0})} = \ol{\Id}$ at the next round. Moreover, $g = \mathbf{1}$ (the constant function that outputs 1 on all inputs), then the correction unitary is $-\ol{\Id}$, which is equivalent to $\ol{\Id}$ up to an unphysical global sign. We claim that after at most $n$ rounds, we will be certain to obtain $g \in \{\mathbf{0}, \mathbf{1}\}$, based on the following proposition.  

\begin{proposition}
    Let $g$ be an $n$-bit Boolean function. Define $\deg(g)$ to be the degree of $g$ when it is expanded as a polynomial of its input bits over the field $\mathbb{F}_2$. Suppose that $\deg(g) = d$. Let $m\in\{0,1\}^{n}$, and let $h = \UR(g,m)$, as defined in \cref{eq:update_rule}. Then, we have 
    \begin{equation}
        \deg(h) \leq d-1\,.
    \end{equation}
\end{proposition}
\begin{proof}
    This is a consequence of the reasoning in \cref{app:QRAM_in_Clifford_hierarchy}, specifically the cancellation in \cref{eq:cancellation_of_monomials} and the reasoning underneath. 
\end{proof}

The proposition establishes that each application of the update rule $g \gets \UR(g,m)$ decreases the degree of $g$ by at least 1. Recall that in round $1$, we have $\deg(g)= \deg(f) \leq n$, simply by virtue of the fact that $f$ is an $n$-bit function. Thus, we may assert that in round $j$ we have $g= g^{(j)}$ and
\begin{align}
    \deg(g^{(j)}) \leq n+1-j\,.
\end{align}
In particular, after applying the update rule in round $n$ with $g = g^{(n)}$ and $m = m^{(n)}$, we are guaranteed to obtain $h = \UR(g,m)$, which leads to a degree $\deg(h) = 0$. If the degree of $h$ is 0, this implies that either $h = \mathbf{0}$ or $h = \mathbf{1}$.

\subsection{Total complexity of protocol}

We have now explained the action of each component of the protocol, and may we state its overall complexity. 

\begin{theorem}[Main result]\label{thm:total_complexity}
    Let $f$ be an arbitrary dataset ($n$-bit Boolean function) for which we wish to implement the fault-tolerant diagonal QRAM unitary $\ol{V(f)}$ of \cref{eq:QRAM_intro}, and let $\varepsilon$ be an error parameter. Suppose that we have access to a noisy physical QRAM device subject to dataset-independent noise (\cref{def:dataset-independent_noise}) which on input $g$ produces state $\tpsig{g}$ on $n$ physical qubits achieving fidelity $\bra{\Psi(g)} \tpsig{g} \ket{\Psi(g)} \geq F$ for all $g$. Suppose further that we have the capability to reload the QRAM device with a new dataset, and that we have the capability to move the 
    $n$-qubit output state to a fault-tolerant quantum processor subject to circuit-level stochastic noise (\cref{def:stochastic-noise}) with error rate $p$. If $p$ is below a constant threshold $p_0$ determined by the QEC code family, and separately if $pn^2$ is below a different constant threshold related to the fault-tolerant encoding procedure, then there exists an adaptive distillation--teleportation procedure that implements a quantum channel $\CP_{\rm DT}$ for which 
    \begin{align}
        \frac{1}{2}\nrm{\CP_{\rm DT} - \ol{V(f)}[\cdot]\ol{V(f)}^\dag}_\diamond \leq \varepsilon\,.
    \end{align}
    The protocol uses (in expectation over internal randomness) $Q$ queries to the noisy physical QRAM device, $Q$ applications of the encoding process $\CE_{\rm FT}$ (\cref{cor:encoding_full}), and $Q'$ additional fault-tolerant operations (controlled-SWAP, Hadamard, CNOT, single-qubit logical state preparations, and single-qubit logical measurements), where
    \begin{align}
        Q &= O\left(\frac{n(1-F)}{F^2}\left(\frac{n}{\varepsilon}+\frac{1}{F}\right)\right)\\
        Q' &= O\left(n^2  Q\right)\,.
    \end{align}
    Additionally, the protocol applies the classical update rule  (\cref{eq:update_rule}) at most $n$ times, and the classical partial Clifford twirling operation $g \mapsto g_C$ (\cref{eq:g_C}) $Q$ times, each time on a data table of size $2^n$ classical bits. 
\end{theorem}
\begin{proof}
    The protocol is depicted as a quantum circuit in \cref{fig:protocol_overview}. We begin with correctness. In each of the $n$ rounds indexed by $j=1,\ldots,n$, it implements a channel $\ol{\CT(g^{(j)})}$. As discussed in \cref{sec:adaptive_correction}, if all resource states are prepared perfectly, the procedure is guaranteed to implement the unitary $\ol{V(f)}$, up to a global sign which does not impact the channel \mbox{$\ol{V(f)}[\cdot] \ol{V(f)}^\dag$}. 

    However, the teleportation channel $\ol{\CT(g^{(j)})}$ is not implemented perfectly by the protocol. To ensure the overall diamond norm error is $\varepsilon$, it suffices to choose parameters such that $\ol{\CT(g)}$ is implemented up to $\varepsilon/n$ diamond distance for all $g$, since the errors from each of the $n$ channel applications add linearly in the worst case, when performed in succession. By \cref{prop:teleportation}, it suffices to distill resource states $\ol{\phi(g)}_{\rm dist}$ that satisfy
    \begin{align}
        \frac{1}{2}\nrm{\ketbra{\ol{\Psi(g)}}-\ol{\phi(g)}_{\rm dist}}_1 \leq \varepsilon_{\rm dist}
    \end{align}
    with distillation error $\varepsilon_{\rm dist} = \varepsilon/n$. Meanwhile, by \cref{prop:distillation_general},  this is accomplished by the distillation protocol using $O(\frac{1-F_{\rm min}}{F_{\rm min}^2}(\frac{1}{\varepsilon_{\rm dist}}+\frac{1}{F_{\rm min}}))$ copies of the input states $\ol{\phi(g)}_{\rm twirl}$, defined in \cref{eq:phi_twirl}, provided that for all $g$ $\bra{\ol{\Psi(g)}} \ol{\phi(g)} \ket{\ol{\Psi(g)}} \geq F_{\rm min}$ for some $F_{\rm min}$.  
    We are guaranteed from \cref{cor:encoding_full} that $F_{\rm min} \geq (1-O(np)-O(n\sqrt{p})) F$, which can be replaced by $\Omega(F)$ as long as $pn^2$ is below a certain constant. The number of gates required by distillation is a factor of $O(n)$ larger than the number of copies. 
    
     Since there are $n$ rounds, we require $n$ calls to the distillation procedure. Thus, the total number of queries to the noisy QRAM device and the $\poly(n)$-cost encoding procedure $\CE_{\rm FT}$ is 
     \begin{align}
         Q= O\left(\frac{n(1-F)}{F^2}\left(\frac{n}{\varepsilon}+\frac{1}{F}\right)\right)
     \end{align}
     The total fault-tolerant gate complexity from distillation is $O(nQ)$.  The teleportation procedure also requires $n$ fault-tolerant CNOT gates in each of the $n$ rounds. Finally, the Clifford twirling step requires the application of $O(n^2)$ fault-tolerant Clifford gates for each of the $Q$ copies. In total, these Clifford gates dominate the fault-tolerant gate count, which is $Q' = O(n^2Q)$.  
     Each round requires only one call to the update rule, and each of the $Q$ copies requires classically applying a partial Clifford update $g \mapsto g_C$.  This completes the proof. 
\end{proof}

\section{Complexity of the classical update rule}\label{sec:complexity_of_classical_update}

The classical update rule calls for updating a data table $g$ to the data table $h = \UR(g,m) = g \oplus g^{\oplus m}$, for a certain fixed measurement outcome $m \in \{0,1\}^n$, as in \cref{eq:update_rule}. That is, the entry at address $x$ in the data table should be updated from $g(x)$ to $g(x) \oplus g(x\oplus m)$. In this section, we analyze the complexity of this transformation under several different frameworks.  The guiding question is to understand the ways in which the classical update rule is a more complex operation than a RAM query.  

We note that the partial Clifford twirling step also requires substantial classical computation to randomly transform the dataset.  We do not specifically analyze this step here, because we view it as less fundamental to our protocol. For example, if the physical QRAM device and encoding step were error free (or if they have errors but the principal eigenvalue of the resulting state is correct), then partial Clifford twirling is not necessary, but the need for the update rule remains. 

\subsection{Classical circuit complexity}

The update rule takes $2^n + n$ bits as input and produces $2^n$ bits as output, as in \cref{eq:circuit_update_rule}. It is straightforward to see that a classical circuit built from elementary gates (NOT, AND, NAND, etc.) would require $\Omega(2^n)$ gates to implement the update rule. Observe that for any fixed nonzero value $m \neq 0^n$, we have for every $x$
\begin{align}
    h(x) = h(x \oplus m) = g(x) \oplus g(x \oplus m)
\end{align}
That is, the $2^n$ addresses are partitioned into pairs $\{x, x \oplus m\}$, where $h$ takes the same value on both elements of the pair, and its value is equal to the parity of the input bits at locations $x$ and $x \oplus m$. Each of these $2^{n-1}$ parity calculations is independent and requires at least one elementary gate. 

\subsection{Complexity in a classical RAM model}

The need for $\Omega(2^n)$ circuit complexity does not alone entail that the update rule is an expensive classical calculation. After all, the standard RAM operation also has $\Omega(2^n)$ circuit complexity, but it is commonplace to consider a model of classical computation where RAM has unit cost. 

However, the single-bit RAM operation takes as input $n$ bits (an address) and returns just 1 bit, so it appears to be a much simpler operation than the classical update rule. Since (i) the output of the update rule has $2^n$ bits, (ii) the output depends on all $2^n$ input bits $g(x)$, and (iii) each RAM query can access only 1 of the bits, we conclude that, in a model where the input data $g$ can only be accessed via RAM queries, implementing the update rule requires $\Omega(2^n)$ RAM queries.

\subsection{Classical circuit depth in an all-to-all model}\label{sec:shallow_circuit_update_rule}

The structured nature of the update rule suggests it may still be amenable to some degree of parallelization. Ideally, one could design a specialized shallow circuit that directly implements the update, rather than relying on RAM. Indeed, if one imposes no restrictions on spatial layout of the $2^n+n$ input bits and $2^n$ output bits (i.e., one allows all-to-all gates), then one can perform the update rule with a classical circuit of depth $O(n)$ comprised of elementary gates each acting on only $O(1)$ bits. 

We provide one possible way to accomplish this. The construction requires $O(n2^n)$ ancilla bits and has the following steps; we provide an $n=3$ example to assist with understanding each step in \cref{fig:shallow_circuit_example} (and in the description of each step, we reference the colors in that figure for clarity of explanation).  Let $e_i$ refer to the length-$n$ bit string with a 1 in the $i$-th position and 0 in the other $n-1$ positions.
\begin{enumerate}
    \item For each $x \in \{0,1\}^n$ in parallel, the (blue) bit holding $g(x)$ is copied into a (yellow) ancilla bit, accomplished in depth 1.
    \item For each $i=1,\ldots, n$ in parallel, the (green) input bit holding $m_i$ is copied into $2^{n-1}-1$ (gray) ancilla bits, accomplished in depth $n-1$ (using a tree-like approach, the number of copies of $m_i$ can be doubled with each additional circuit layer). 
    \item For each $i = 1, \ldots, n$ in series, and for each of the $2^{n-1}$ address pairs $\{x, x \oplus e_i\}$ in parallel, if $m_i=1$ then the (yellow) ancilla bit which held the copy of $g(x)$ at the beginning of this step is swapped with the (yellow) ancilla bit which held the copy of $g(x \oplus e_i)$ at the beginning of this step. Each value of $i$ requires depth 1: the availability of $2^{n-1}$ (green and gray) copies of the $m_i$ bit created in step 2 allows parallelization of the $m_i$-controlled $\{x,x\oplus e_i\}$ swaps. Thus, the overall depth of this step is $n$.   Over the course of the $n$ steps, the (yellow) ancilla  bit that originally held $g(x)$ before this step undergoes the transformations
    \begin{equation}
        g(x) \overset{i=1}{\mapsto} g(x \oplus m_1e_1) \overset{i=2}{\mapsto} g(x \oplus m_1e_1 \oplus m_2e_2) \overset{i=3}{\mapsto} \cdots \overset{i=n}{\mapsto} g(x \oplus \bigoplus_{i=1}^n m_i e_i) = g(x \oplus m)\,,
    \end{equation}
    that is, for each $x$, the (yellow) ancilla bit originally holding $g(x)$ now holds $g(x \oplus m)$. 
    \item For each $x \in \{0,1\}^n$ in parallel, the (yellow) ancilla  bit that originally held the copy of $g(x)$ after step 1 (which now holds $g(x \oplus m)$) is added modulo 2 into the (blue) bit  holding $g(x)$, incurring depth $1$.
    \item The original $2^n$ (blue) input bits are taken to be the $2^n$ output bits; the ancilla bits are discarded.
\end{enumerate} 
The registers (blue) holding the original bits $g(x)$ have now been updated to $h(x) = g(x) \oplus g(x\oplus m)$, which is the desired output of the update rule. 

\begin{figure}[ht!]
\centering
\includegraphics[ width=0.95\textwidth]{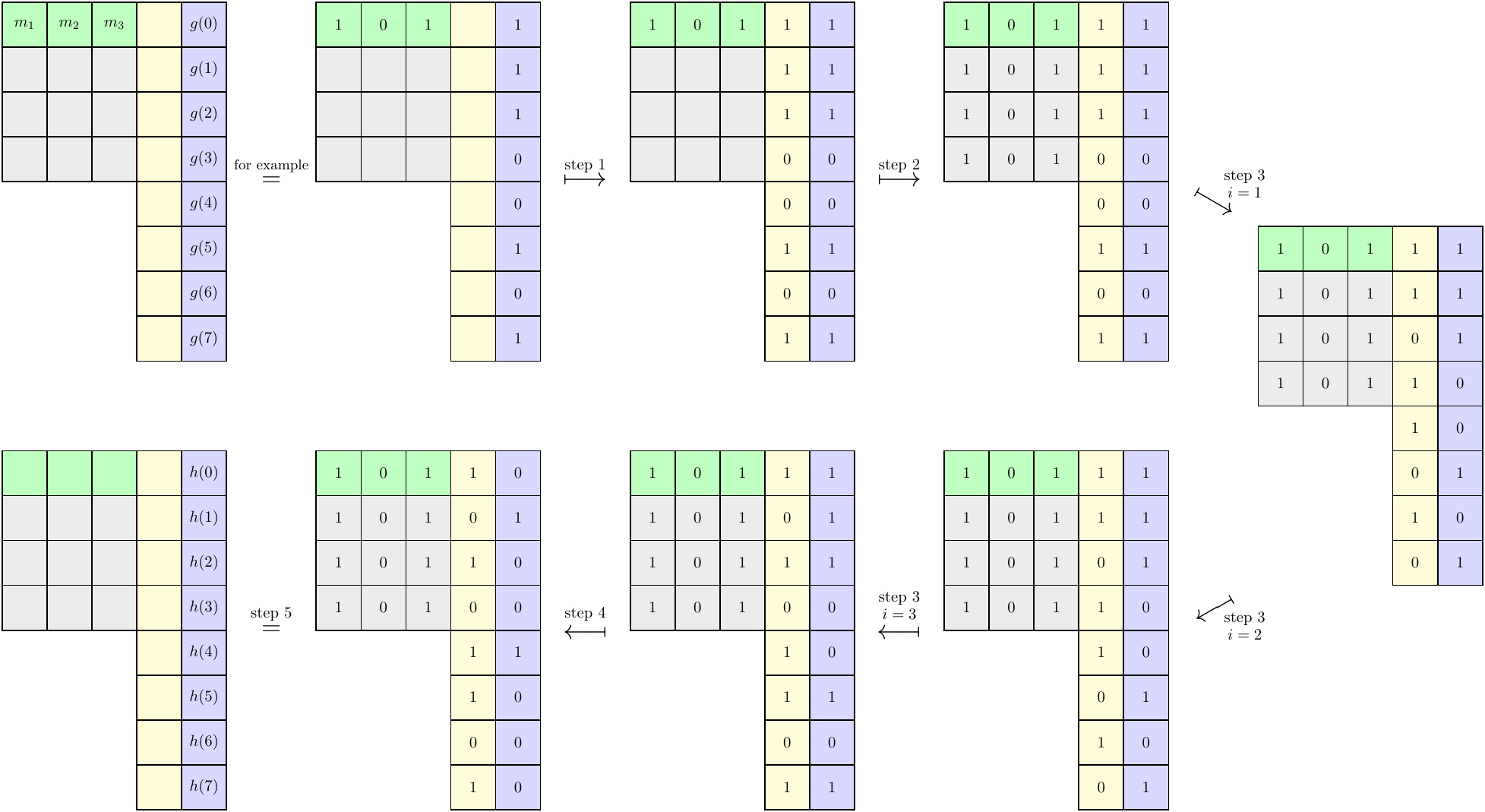}
\caption{\label{fig:shallow_circuit_example} Step-by-step action of the $O(n)$-depth classical circuit that implements the update rule $g \mapsto h = g \oplus g^{\oplus m}$, for a particular $n=3$ example input with $m = (1,0,1)$. Each box stores one bit of information, and each step modifies some subset of these bits with parallelized layers of local gates (note that step 2 requires $n-1$ layers to create all $2^{n-1}-1$ copies of each $m_i$). The inputs to the update rule are the $n$ bits of $m$ (green) and the $2^n$ bits in the classical data table storing $g$ (blue). The circuit utilizes $2^n + n(2^{n-1}-1)$ ancilla bits (gray and yellow).}
\end{figure}

\subsection{Classical circuit depth in a spatially local model}

If the input and output bits of the shallow circuit are embedded into $d$ spatial dimensions, then the $O(n)$-depth circuit above requires spatially nonlocal gates. For example, we may recognize the action of step 3 as a re-arrangement of the $2^n$ (yellow) ancilla bits by traversing edges of the $n$-dimensional Boolean hypercube. In $d=n$ spatial dimensions, this could be done using spatially local gates, and each (yellow) classical bit need only interact with $O(n)$ of the $2^n$ other (yellow) bits (its neighbors on the hypercube). Unfortunately, real classical circuits must be embedded into $d \leq 3$ spatial dimensions. In this case, step 3 cannot be performed exclusively with spatially local gates for more than $d$ of the values of $i \in \{1,\ldots,n\}$. 

In fact, we can show that if gates are local in $d$ spatial dimensions, then the depth required is at least $\Omega(2^{n/d})$. This follows from the fact that without knowing the value of $m$, the circuit must be prepared to connect the bit at address $x$ to all $2^n-1$ of the other bits; for any pair $x,y \in \{0,1\}^n$, if $m = x \oplus y$, then the circuit must be able to compute the parity of $g(x)$ and $g(y)$. If the bits storing $g(x)$ and $g(y)$ live on opposite sides of the $d$-dimensional array, computing this parity will require a classical circuit with depth at least $\Omega(2^{n/d})$. Ultimately, this essentially amounts to a speed of light--type restriction, where depth is restricted due to the fact that information can only move so quickly through space. This kind of argument could also be used to show that classical RAM requires a circuit of depth $\Omega(2^{n/d})$ in a local model, in $d$ spatial dimensions, so it does not represent a fundamental limitation that is unique to QRAM.  

\subsection{Wire density}

While speed of light--type restrictions can be relevant for RAM at large scales, they are not a factor in practice at small or  intermediate scales. If one ignores the speed of light, one can simply build long wires into the circuit and enable all-to-all connectivity. These long wires should not be thought of as completely free, however. For example, electronic ciruits typically dissipate energy and lead to heating in proportion to their total wire length. The need to cool electronic chips is a key limitation in practical computer systems. Thus, a cost metric worth considering \cite{jaques2023qram} is the wire density of the circuit, which we define as the total wire length divided by the total spacetime of the circuit, where the spacetime is defined as the number of bits the circuit uses (henceforth referred to as the circuit width) multiplied by the circuit depth.  

Classical circuits for RAM can have depth $O(n)$, width $O(2^n)$, and total wire length $O(n2^n)$ \cite{jaques2023qram}, even when embedded in one spatial dimension. Thus, the wire density is a constant with respect to $n$, suggesting the circuit can be scaled without causing heating issues. 

We now examine the circuit for the update rule described in \cref{sec:shallow_circuit_update_rule}.  It has depth $O(n)$ and width $O(n2^n)$. Steps 1 and 2 perform copying of the bits and contribute wire length $O(n2^n)$. If the circuit is embedded in $n$ spatial dimensions, then step 3 can also be accomplished with total wire length $O(n2^n)$, as each gate is local and has $O(1)$ wire length. However, in $d \leq 3$ spatial dimensions, the wire length is asymptotically larger. To implement step 3, the (yellow) ancilla bits storing the copy of $g(x)$ must be connected to the (yellow) ancilla bits initially storing the copies of $g(x \oplus e_1)$, $g(x \oplus e_2)$, \ldots, $g(x \oplus e_n)$---essentially, an embedding of the $n$-dimensional Boolean hypercube into $d$ dimensions. Consequently, the total wire length of gates acting on each (yellow) ancilla bit will be at least $\Omega(2^{n/d})$. Since there are $2^n$ (yellow) ancilla bits, the overall wire length of the circuit is at least $\Omega(2^{n(1+1/d)})$ and thus the wire density grows with $n$ as $\Omega(2^{n/d}/\poly(n))$, a fundamentally different outcome than the case of classical RAM. 

\subsection{Relation to matrix-vector multiplication and the Walsh--Hadamard transform}

matrix-vector multiplication for $2^n \times 2^n$ matrices is an operation with $2^n$ inputs (the entries of the input vector) and $2^n$ outputs (the entries of the output vector). This feature is similar to the update rule, although the inputs and outputs for matrix multipliation would typically each be multiple bits (e.g., an integer or floating point number), rather than just a single bit. Moreover, the analysis of \refcite{jaques2023qram} demonstrated how classical sparse matrix-vector multiplication requires a growing wire density, consistent with the observation above for the update rule. 

Here, we will argue that the update rule is in a certain sense equivalent to a sparse matrix-vector multiplication, and specifically it is equivalent to the  Walsh--Hadamard (WH) transform up to factors of $\poly(n)$. This equivalence holds under a parallel model of computation. Namely, we assume that we have $2^n$  classical co-processors, each with $\poly(n)$-size local memory, which may perform local arithmetic on their memory in parallel, but cannot communicate with one another, except through joint application of the update rule or through joint application of  a sparse matrix-vector multiplication. When jointly applying the update rule, each of the $2^n$ processors supplies one entry $g(x)$ and receives the output $h(x) = g(x) \oplus g(x \oplus m)$ for a fixed global $m$. When jointly applying sparse matrix-vector multiplication, each processor provides one of the $2^n$ entries of the input vector and receives one of the $2^n$ entries of the output vector. The remainder of this subsection aims to justify this claim of equivalence. 

\subsubsection{The (fast) Walsh--Hadamard transform}

The WH transform is the multiplication of a length-$2^n$ vector by a $2^n \times 2^n$ matrix denoted by $H$, where the matrix element associated with the transition from $n$-bit input address $y$ to $n$-bit output address $x$ is given by 
\begin{align}
    H_{xy} = \frac{1}{2^{n/2}}(-1)^{x \cdot y}
\end{align}
We can see that $H_{xy}$ is a dense matrix---all of its entries are nonzero---however, it can be shown that it is the product of $n$ sparse matrices. Specifically, let $H^{(i)}$ be defined as\footnote{We note that the matrix $H^{(i)}$ is proportional (by a factor $\sqrt{2}$) to the transformation applied to the amplitudes of a quantum state when a single-qubit Hadamard gate is applied to the $i$-th qubit, and identity is applied to the other $n-1$ qubits. The full WH transform is the product of the $H^{(i)}$ matrices, that is, the Hadamard gate on all $n$ qubits. }
\begin{align}
    H^{(i)}_{xy} = \begin{cases}
        1 & \text{if } x = y \oplus e_i \\
        (-1)^{x_i y_i} & \text{if } x = y  \\
        0 & \text{otherwise}
    \end{cases}
\end{align}
Then, it holds that
\begin{align}
    H = \frac{1}{2^{n/2}}H^{(n)}H^{(n-1)}\cdots H^{(2)}H^{(1)}\,,
\end{align}
This fact can be verified by noting that the nonzero offdiagonal entries of $H^{(i)}$ are matrix elements $H^{(i)}_{uv}$ where $u$ and $v$ differ only on the $i$-th bit. Thus, any offdiagonal transition element $H_{xy}$ can only be obtained in the product $H^{(n)}\cdots H^{(1)}$ by choosing the corresponding offdiagonal entry of $H^{(i)}$, (equal to 1) whenever $x_i \neq y_i$ differ, and the diagonal entry $(-1)^{x_iy_i}$ of $H^{(i)}$ whenever $x_i = y_i $. This gives precisely the quantity $(-1)^{\sum_i x_i y_i} = (-1)^{x \cdot y}$. 

Matrix multiplication by $H^{(i)}$ requires only $O(2^n)$ arithmetic operations, since each row of $H^{(i)}$ has only 2 nonzero entries. The fast WH transform utilizes this decomposition to implement multiplication by $H$ in classical time $\poly(n)2^n$, much smaller than the $\Omega(2^{2n})$ time required to multiply general dense matrices that do not have this kind of decomposition.  

\subsubsection{Reduction from update rule to Walsh--Hadamard transform}

Now, we show how the update rule can be accomplished with two applications of the WH transform and parallel local arithmetic. Let $\mathbf{g}$ denote the length-$2^n$ vector with entry $g(x)$ at index $x$, and let $\mathbf{h}$ denote the length-$2^n$ vector with entry $g(x) + g(x \oplus m)$ at entry $x$. Note here that we are using normal addition $+$, rather than modular addition $\oplus$, and we allow $\mathbf{h}$ to take integer values in $\{0,1,2\}$. We have that the $x$-th entry of the matrix-vector product $H \mathbf{h}$ is
\begin{align}
    (H\mathbf{h})_x &= \frac{1}{2^{n/2}}\sum_{y \in \{0,1\}^n}  (-1)^{x \cdot y} (g(y) + g(y \oplus m)) \\
    &= \left[\frac{1}{2^{n/2}}\sum_{y \in \{0,1\}^n}  (-1)^{x \cdot y} g(y)\right] + (-1)^{m \cdot x} \left[ \frac{1}{2^{n/2}}\sum_{y \in \{0,1\}^n}  (-1)^{x \cdot (y\oplus m)} g(y \oplus m) \right] \\
    &= (H\mathbf{g})_x + (-1)^{m \cdot x}(H\mathbf{g})_x \\
    &= \begin{cases}
        2 (H \mathbf{g})_x & \text{if } m \cdot x = 0 \\
        0 & \text{if } m \cdot x = 1
    \end{cases}
\end{align}
That is, the WH transform $H\mathbf{h}$ of the update rule output $\mathbf{h}$ can be easily computed from the WH transform $H\mathbf{g}$ of the update rule input $\mathbf{g}$. 

Let $M^{(m)}$ be the diagonal matrix for which $M^{(m)}_{xx} = 2$ if $m \cdot x = 0$ and $M^{(m)}_{xx} = 0$ if $m \cdot x = 1$. Then, we may use the equation above to say that $H\mathbf{h} = MH\mathbf{g}$. Since $H^2 = \Id$, we then have
\begin{align}
    \mathbf{h} = HH\mathbf{h} = HM^{(m)}H\mathbf{g}
\end{align}
This gives a straightforward way to transform $\mathbf{g} \mapsto \mathbf{h}$ in the model where $2^n$ classical co-processors can perform local arithmetic in parallel or jointly perform the WH transform. First, we assume that $2^n$ processors each locally store the bit $g(x)$ for one address $x$, and an $n$-bit copy of $m$ (note that a copy of $m$ could be pre-distributed to each of the $2^n$ processors via a tree-like circuit of constant wire density, as in step 2 of \cref{sec:shallow_circuit_update_rule}). Then, $\mathbf{g} \mapsto \mathbf{h}$ is accomplished by applying the WH transform, applying the diagonal transformation $M^{(m)}$, and finally applying the WH transform again. The $x$-th entry of the diagonal matrix $M^{(m)}$ can be computed in $O(1)$ depth locally by the $x$-th processor acting on its internal memory.  We recall that $\mathbf{h}$ may have entries in  $\{0,1,2\}$; to recover binary entries, we simply take every entry modulo 2 via parallel local arithmetic, which does not discard any important information, owing to the fact that $(-1)^{g(x) + g(x \oplus m)} = (-1)^{g(x) \oplus g(x \oplus m)}$. 

The conclusion is that $O(1)$ applications of the WH transform, along with parallel local arithmetic, is sufficient to implement the classical update rule, or in other words, the classical update rule is no harder than the WH transform.

\subsubsection{Reduction from WH transform to update rule}\label{eq:reduction_WH_to_UR}
Now, we show the converse: that the WH transform can be implemented using $\poly(n)$ applications of the update rule along with parallel local arithmetic. Specifically, we aim to implement the matrix transformation $H^{(i)}$. Suppose we are given a vector $\mathbf{g}$ of integers each represented by at most $n$ bits.\footnote{The assumption of integer entries is without loss of generality. If the entries are not integers but rather non-integer numbers expressed in binary with a finite number of bits of precision, then we can always multiply by a power of 2 so that all the entries are integers, and divide by that power of 2 after performing the calculation.}  Let $g(x)$ denote the integer corresponding to the entry of $\mathbf{g}$ at address $x$, which is stored in the local memory of one of the $2^n$ processors. Let $g_j(x)$ denote the $j$-th bit of the integer $g(x)$.  We perform the following steps
\begin{enumerate}
    \item For each address $x$, the co-processor storing $g(x)$ makes a copy of $g(x)$ in its local memory.
    \item We fix global $n$-bit string $m = e_i$ (known by all processors) and for each $j$, the $2^n$ co-procesoors jointly apply the update rule $\UR(g_j,e_i)$ onto the set of $2^n$ bits $g_j(x)$ (i.e., bit-wise application of update rule).  For each $x \in \{0,1\}^n$, and each $j$, the co-processor at address $x$ now has in its memory the value  $g_j(x)$ and the value $g_j(x) \oplus g_j(x \oplus e_i)$.
    \item For each $x$ and each $j$, the co-processor at address $x$ performs local arithmetic to add (modulo 2) the local register holding $g_j(x)$ into the local register holding $g_j(x) \oplus g_j(x\oplus e_i)$  so that the two registers now hold $g_j(x)$ and $g_j(x \oplus e_i)$. Effectively, this step and the previous step have together performed a swap $g(x) \leftrightarrow g(x \oplus e_i)$ between the various co-processors. 
    \item For each address $x$, the co-processor at address $x$  takes the local register holding integer $g(x)$ and flips it to $-g(x)$ only if $x_i = 1$. 
    \item Finally, for each address $x$, the co-processor at address $x$ adds the local register holding $(-1)^{x_i}g(x)$ into the local register holding $g(x\oplus e_i)$, such that the latter register now holds $g(x \oplus e_i) + (-1)^{x_i} g(x)$.  
\end{enumerate}
The latter register in the local memory of co-processor $x$ now holds precisely the value $(H^{(i)} \mathbf{g})_x$, indicating that the co-processors have successfully managed to apply one step of the fast WH transform. Accomplishing this required one application of the update rule for each bit of the entries of the vector. Assuming the integer entries have at most $\poly(n)$ bits, this means $H^{(i)}$ can be accomplished with $\poly(n)$ applications of the update rule, and local arithmetic, specifically, copying (step 1), bit-wise addition mod 2 (step 3), negation (step 4), and integer addition (step 5). 

Since $H$ is the product of $n$ matrices $H^{(i)}$ up to a proportionality constant, we conclude that $\poly(n)$ applications of the update rule are sufficient to implement the WH transform in this model of computation; in other words, the WH transform is no harder than the update rule, up to a factor of $\poly(n)$.  

\subsubsection{Performing general sparse matrix-vector multiplication using the update rule}

Multiplication by $H^{(i)}$ is in some sense easier than multiplication by an arbitrary sparse matrix, since (i) the location of the nonzero entries is highly structured and (ii) all entries are $\pm 1$, avoiding the need for any integer or floating-point multiplications. The only arithmetic required is addition and subtraction. Here we discuss how the update rule can also be used for arbitrary sparse matrix-vector multiplications. 

Lemma A.3 of \refcite{jaques2023qram} examines performing general sparse matrix-vector multiplication using parallel classical co-processors, each capable of local addition and multiplication. If the parallel processors can communicate with each other via links that form a sorting network, then the sparse matrix-vector multiplication can be accomplished in roughly the time required to perform a sort using the sorting network. As shown in steps 2 and 3 of the procedure in \cref{eq:reduction_WH_to_UR}, the update rule and parallel bitwise xor together allow a swap of data at location $x$ and location $x \oplus e_i$ for all $x$ in parallel. By copying the data before performing the swap, and then choosing whether to discard the original copy or the swapped copy, one can use the update rule to swap in parallel any \textit{subset} of the location pairs $(x, x\oplus e_i)$.  Thus, the ability to perform the update rule enables access to parallel swaps along a Boolean hypercube connectivity in $n$ dimensions, a model for which it is known that sorting can be completed in $\poly(n)$ time \cite{ajtai1983sortingNetwork,beals2013efficientDistributedQuantumComputing}.  
Together with \refcite[Lemma A.3]{jaques2023qram} this shows how $\poly(n)$ applications of the update rule and $\poly(n)$ rounds of parallel local arithmetic on the database entries enables arbitrary sparse matrix-vector multiplication.

Unlike for the WH transform, this procedure for arbitrary sparse matrix-vector multiplication also requires the local arithmetic to include multiplication, rather than just addition. However, $\poly(n)$-bit integer multiplication can be accomplished with only $\poly(n)$ integer additions. In any case, this suggests that the update rule is essentially equivalent to a general sparse matrix-vector multiplication, at least in this model where parallel local arithmetic is possible. We note that it could still be possible that in practice that the constant and $\poly(n)$ prefactors for the update rule are substantially better than those for general sparse matrix-vector multiplication.

\subsection{Concluding comments on parallelization of the update rule }

The above arguments establish that the update rule can be parallelized to $O(n)$ depth, but only in a model where all-to-all gates are possible. Embedding this all-to-all circuit into a finite number of spatial dimensions causes the circuit to have wire density that grows exponentially with $n$, a feature that suggests the parallelized update rule is less scalable than classical RAM. It is unclear whether this growing wire density is problematic for practical sizes of $n$. 

Relatedly, we have shown that in a parallel model of computation, the ability to apply the update rule is roughly equivalent to the ability to apply a sparse matrix-vector multiplication, and there is a particularly close connection to the WH transform. In many instances, sparse matrix-vector multiplication can be parallelized very successfully in practice, for example, by leveraging graphics processing units (GPUs), where the parallelization is hardwired into the chip. This comes despite the fact that, asymptotically speaking, the wire density of a sparse matrix-vector multiplication would increase exponentially with $n$ \cite{jaques2023qram}. Since the classical update rule seems to be an operation that is no harder than a sparse matrix-vector multiplication, we are hopeful that it would be possible to effectively parallelize the classical update rule in practice. 

However, this argument also clarifies the opportunity cost of the classical resources dedicated to performing the update rule. For example, if the quantum algorithm requiring fault-tolerant QRAM aims to solve a certain linear algebra problem, one must consider whether the classical device that performs the classical update rule is capable of solving that problem on its own, without the need for a quantum computer at all. After all, many linear algebra problems, such as solving sparse linear systems, can be solved using a small number of sparse matrix-vector multiplications; see \cref{sec:applications_ML} and \refcite{jaques2023qram}.

\section{Applications}\label{sec:applications}
In this section, we make a coarse attempt at estimating the resources required by our protocol in several applications. The goal is to shed light on which aspects of our protocol are most in need of improvement for it to be useful. 

Generally, these applications come in two flavors---first, there are big-data appliations that heavily rely on QRAM and require the cheap QRAM assumption from \cref{sec:introduction} to have a significant quantum speedup. For these applications, conceptually speaking, the limiting aspect of using our protocol is the exponential \textit{classical} computation required to do the update rule (and the Clifford twirling), as this prevents a true exponential speedup and full justification of the cheap QRAM assumption. One must always consider that the classical resources required to perform the update rule could be re-purposed directly toward solving the computational problem; in \cref{sec:complexity_of_classical_update}, we discussed how the classical update rule is in a sense equivalent in complexity to a sparse matrix-vector multipliation.

The second flavor of application are those where the QRAM operation of \cref{eq:QRAM_intro} (or its $b$-bit generalization, discussed in \cref{app:extension_to_b_bits}) is required, but the cheap QRAM assumption is not essential---cryptanalysis and chemistry, below. In fact, in these applications, the fault-tolerant QRAM operation is typically compiled as a space-efficient quantum circuit, requiring only $O(n)$ logical qubits and $\Omega(2^n)$ depth---in this context, the operation is often referred to as QROM (quantum read-only memory) \cite{babbush2018EncodingElectronicSpectraLinearT}, rather than QRAM. Assigning QRO(A)M cost $\Omega(2^n)$ does not necessarily jeopardize the possibility of quantum advantage. This is an appealing place to apply our protocol, because it may be viewed as offloading exponential resources from the quantum processor to the classical processor, which is typically a favorable trade. The issue we face here is that the $O(1/\varepsilon)$ overhead for distillation in our method compares unfavorably to the $\polylog(1/\varepsilon)$ overhead achieved for distillation of $T$ and $\mathrm{CCZ}$ magic states, and $\varepsilon$ may need to be taken quite small if QRAM is applied many times. Furthermore, the values of $n$ encountered in practice may not be sufficiently large for the asymptotic advantage of our method to kick in, although future improvements could cut down on the overhead of our protocol. 

\subsection{Arbitrary quantum state preparation}\label{sec:application_state_prep}

Our protocol for fault-tolerant QRAM could be used to prepare an arbitrary $n$-qubit state, given a list of its $2^n$ amplitudes stored in classical memory. This subroutine could be useful for preparing initial states for the simulation of the dynamics of many-body systems or ansatz states for quantum phase estimation or variational algorithms. It would also be useful in certain algorithms for quantum machine learning or solving differential equations. 

The process for creating an arbitrary state with QRAM goes roughly as follows \cite{kerenidis2016QRecSys}. First, one classically pre-processes the $2^n$ complex amplitudes that define the state (comprising $2^{n+1}-2$ independent real degrees of freedom, accounting for normalization and an unphysical global phase) to compute a list of $2^{n+1}-2$ rotation angles. The arbitrary state can be prepared through a sequence of steps labeled by $i=1,\ldots, n$. On step $i$, a single-qubit rotation is performed on qubit $i$ by one of $2^{i-1}$ angles; which angle to use depends on the setting of qubits $1,\ldots,i-1$. If the angles are accessible via QRAM, the quantum algorithm can compute the angle with a single query by using qubits $1,\ldots,i-1$ as the address. Once the angle has been read in, the single-qubit rotation on qubit $i$ can be efficiently performed, and then a second query to the QRAM can be made to uncompute the angle.  For more details on this strategy, see for example \refcite{kerenidis2016QRecSys, low2018tradingTgatesforDirtyQubits,clader2022resourcesForBlockEncoding,dalzell2023quantumAlgorithmsSurvey}. At most $2n$ fault-tolerant QRAM queries would be required, two queries each to datasets of sizes $1,2,4, \ldots,2^{n}$ angles. These $O(n)$ queries could be implemented fault-tolerantly by our protocol---by \cref{thm:total_complexity}, each of the $n$ fault-tolerant queries can be implemented up to error $\varepsilon/n$ (so that the total error is $\varepsilon$) using $O(n^3/\varepsilon)$ queries to the physical QRAM device that can coherently access datasets of size up to $2^n$, giving a total of $O(n^4/\varepsilon)$ physical queries. However, we note that it is possible the $n$ dependence could be improved by leveraging the fact that only a small fraction of these queries truly require the full $2^n$-size QRAM. 

Thus, our protocol could represent a great improvement over the $\Omega(2^n)$ fault-tolerant quantum gates (of which at least $\Omega(\sqrt{2^n})$ must be non-Clifford gates \cite{low2018tradingTgatesforDirtyQubits}) required by a standard circuit approach to state preparation. Since state preparation is typically only one component of a larger algorithm, whether this would be a worthwhile approach within end-to-end applications may depend on the details of the application. 

\subsection{Quantum machine learning}\label{sec:applications_ML}

The quantum linear system algorithm \cite{harrow2009QLinSysSolver} prepares a quantum state $\ket{\mathbf{x}}$ encoding the solution to a well-conditioned $2^n \times 2^n$ linear system $A \mathbf{x} = \mathbf{b}$ using only $\poly(n)$ queries to data access oracles for the entries of the matrix $A$ and the vector $\mathbf{b}$. Thus, remarkably, the number of queries needed can be exponentially smaller than the size of the matrices and vectors themselves, due to the ability of the quantum algorithm to explore the exponentially large Hilbert space in superposition.  This insight, and the broader framework of quantum linear algebra \cite{gilyen2018QSingValTransf,lin2022LectureNotes, dalzell2022socp}, has led to a number of quantum algorithms in the realm of machine learning \cite{ciliberto2018QMLReview,biamonte2016QuantumMachineLearning}, where linear algebra problems are ubiquitous. 

When the entries of $A$ and $\mathbf{b}$ have succinct formulas, the data access oracles can be implemented with $\poly(n)$ gate complexity directly on a fault-tolerant quantum processor. However, in big-data applications, it is more relevant to consider $A$ and $\mathbf{b}$ as having arbitrary entries determined by the data, which is stored in classical memory. In these cases, the data access oracles, for example, a unitary block-encoding of $A$ \cite{gilyen2018QSingValTransf,lin2022LectureNotes, clader2022resourcesForBlockEncoding} or a state-preparation unitary for $\ket{\mathbf{b}}$ (see \cref{sec:application_state_prep})  can be implemented with $\poly(n)$ QRAM queries. Since these algorithms will require many quantum gates and many calls to the data access oracles to solve an end-to-end machine learning problem, they will only be possible once fault-tolerant quantum computers are available, and furthermore, to potentially provide exponential speedups, they will require the ability to perform QRAM fault-tolerantly at cost $\poly(n)$ (cheap QRAM assumption from \cref{sec:introduction}).  

It is worth briefly mentioning that many of these quantum machine learning algorithms have been dequantized \cite{tang2018QuantumInspiredRecommSys,tang2018QInspiredClassAlgPCA,chia2022sampling, Shao2022FQILSS, tang2022dequantizingAlgorithmsMLreview, tang2023phdThesisQMLwithoutAnyQuantum}, in the sense that quantum-inspired classical algorithms can also achieve $\poly(n)$ complexity for problems involving datasets of size $2^n$, using the ability to query individual entries of the dataset (e.g., via RAM) and also to sample an entry with probability proportional to its magnitude (a classical analogue of arbitrary state preparation). In these cases,  the quantum algorithm cannot provide an exponential speedup.  Polynomial speedups may still be possible if the cheap QRAM assumption holds. Moreover, these classical methods do not apply in every case; notably, when the matrices involved are sparse and high rank, the possibility of exponential quantum speedup persists \cite{yamakasi2020learningOptimizedRandomFeatures}, provided that the cheap QRAM assumption is true. 

To see how our protocol impacts the outlook of these applications, we suppose generically that a quantum machine learning algorithm requires $\poly(n)$ fault-tolerant quantum gates and $\poly(n)$ fault-tolerant queries to QRAM to complete its task. By implementing QRAM using our fault-tolerant QRAM protocol with error parameter $\varepsilon = 1/\poly(n)$, the problem can be solved with $\poly(n)$ fault-tolerant gates and $\poly(n)$ calls to the faulty physical QRAM device, which keeps open the possibility of superpolynomial speedup in quantum resources, here assuming that the computational cost of the physical query is also $\poly(n)$. 

However, our protocol also requires classical computations of complexity $O(2^n)$ to perform the update rule. Although this classical complexity may be parallelized, as we discussed in \cref{sec:complexity_of_classical_update}, it is essentially equivalent to the ability to perform a sparse matrix-vector multiplication for a vector of size $2^n$. This is a crucial caveat for machine learning, since sparse matrix-vector multiplication often suffices to efficiently solve the problem in the first place (see discussion in \refcite{jaques2023qram}). For example, the classical conjugate gradient method \cite{Hestenes1952MethodsOC,hackbusch1994IterativeSolLargeSparseLE} can invert well-conditioned, sparse linear systems $A\mathbf{x} = \mathbf{b}$ with $\poly(n)$ sparse matrix-vector multiplications. In fact, the number of sparse matrix-vector multiplications required by conjugate gradient has better asymptotic complexity (scaling as the square root of the condition number for positive semidefinite $A$) than the number of QRAM calls required by the quantum linear system algorithm (scaling at least linearly in the condition number \cite{Orsucci2021PositiveDefiniteLinearSystem}). This suggests that in a generic instance of the sparse linear system problem, it would be more efficient to apply $\Omega(2^n)$ classical computational resources directly toward solving the linear algebra problem, rather than to perform the update rule required by our protocol. 

That said, a few opportunities remain. For instance, one can consider the case that the sparse $2^n \times 2^n$ matrix $A$ has repeated entries or some other kind of high-level structure, such that the number of classical degrees of freedom is asymptotically less than $2^n$. For concreteness, suppose that the size of the dataset determining $A$ is only $\sqrt{2^n}$. Then, the QRAM operation needs only to be applied for size-$\sqrt{2^n}$ datasets, and the classical complexity of the update rule is only $O(\sqrt{2^n})$, quadratically cheaper than the number of arithmetic operations required for a full matrix-vector multiplication by the matrix $A$. In this situation, our protocol may enable an end-to-end solution with $\poly(n)$ quantum cost and $O(\sqrt{2^n})$ classical cost, which could represent a quadratic speedup in terms of classical complexity and an exponential speedup in terms of quantum complexity. While quadratic \textit{quantum} speedups are typically thought of as insufficient to overcome the large disadvantage in constant prefactor for quantum computation \cite{babbush2021FocusBeyondQuadratic}, a quadratic \textit{classical} speedup faces no such obstacles, and could be considered quite large. Generalizing this line of thinking, we may expect to find end-to-end speedups in situations where the best achievable  classical complexity is asymptotically greater than the number of classical degrees of freedom of the problem. Toward this end, another example worth exploring may be the inversion of dense matrices with $O(2^{2n})$ degrees of freedom, which generally requires $2^{\omega n}$ classical arithmetic operations with $\omega > 2$. It remains to find concrete end-to-end examples where quantum advantages of this kind may come to fruition. 

\subsection{Cryptanalysis}

Shor's algorithm for factoring or computing the discrete logarithm relies on performing coherent modular arithmetic. Gidney's windowed quantum arithmetic~\cite{gidney2019windowed} has been used to reduce the cost of Shor's algorithm by replacing some of this costly arithmetic with coherent reads from a quantum lookup table of classically precomputed values~\cite{gidney2021HowToFactor,haner2020ImprovedEllipticCurve}. The algorithm can be expressed as repeated blocks of a coherent lookup table read, each followed by a coherent addition.\footnote{This could be regular addition, modular addition, or elliptic curve point addition.} We investigate the performance of our distillation-based QRAM scheme for implementing the coherent lookup table reads in Shor's algorithm. Elliptic curve cryptography (ECC) is a more attractive target than factoring because of the smaller number of calls to the lookup table, which allows a less stringent target $\varepsilon$, as well as the high cost of elliptic curve addition compared to modular addition. We follow the presentation of \refcite{litinski2023EllipticCurvesBaseline}, where Shor's algorithm for ECC is presented as two applications of quantum phase estimation on unitaries of the form:
\begin{equation}
    U_{X} \ket{R} = \ket{R+X},
\end{equation}
where $R=(x,y)$ is an elliptic curve point, and $X$ denotes either the base point $P$ or the public key $Q$. The goal is to compute integer $j$ such that $Q = [j]P$, where the notation $[j]$ indicates that we are performing elliptic curve scalar multiplication of the point $P$ $j$ times (see \refcite{litinski2023EllipticCurvesBaseline} for full definitions of the addion and multiplication operations). Quantum phase estimation uses controlled unitaries of the form $U_X^{2^j}$ for $0 \leq j \leq k-1$, where $k$ denotes the number of bits in the ECC scheme (e.g., ECC-256). In \refcite{litinski2023EllipticCurvesBaseline}, windowed quantum arithmetic is used to replace blocks of 16 controlled unitaries with a single QROM load of $2^{16}$ classically pre-computed values. The load is followed by an elliptic curve addition operation. The cost of loading $N$ pieces of classical data from a QROM is $N$ Toffoli gates. For $N=2^{16}$, this is approximately $6.5 \times 10^4$ Toffolis. We note that this is much smaller than the cost of the elliptic curve addition operation, which is approximately $8.34 \times 10^6$. The minimum Toffoli cost of the algorithm is obtained by loading groups of $2^{19}$ values, while the minimum active volume cost is obtained at $2^{16}$ values. 

For the sake of calculation, we suppose that the physical QRAM device produces states with minimum fidelity of $F= 50\%$, and we seek to achieve error $\varepsilon_{\rm tot} = 0.1$ (where $\varepsilon_{\rm tot}$ is the sum of the errors of all QRAM loads in the algorithm). We use the generalized $b$-bit version of our protocol from \cref{app:extension_to_b_bits}. If we choose the iterated swap test distillation protocol for its simplicity (each swap test requiring exactly $n+b$ non-Clifford Toffoli gates), for non-vanishing $F$, the number of non-Clifford gates for distillation scales roughly as $O((1-F)n^2(n+b)/\varepsilon))$ with $\varepsilon$ the error per fault-tolerant QRAM query (note that the twirling and teleportation steps are entirely Clifford). For $F \sim 50\%$, and assuming for simplicity a constant prefactor of 1, we estimate the non-Clifford cost as $n^2(n+b)/(2\varepsilon)$.  We use this expression to calculate that the QRAM-based approach achieves its minimum Toffoli and active volume cost at a group size of $2^{64}$ values, which is impractically large for the classical update step of our scheme. A more realistic size of $2^{32}$ values only increases the costs by approximately $10$\%. Nevertheless, both the minimum Toffoli and active volume costs of the existing QROM-based approach are approximately a factor of $1.8\times$ and $1.6\times$ lower than the minimum costs of our QRAM-scheme, respectively, and we have not even considered the computational cost of applying the QRAM device itself.  A major contribution to the cost of our method stems from the large amount of data to be loaded ($b=512$), which features multiplicatively in our Toffoli costs. In contrast, the Toffoli cost of the QROM scheme is independent of the value of $b$. 
Improved methods for distillation, especially when $b$ is large, could make our scheme more competitive in this application.

\subsection{Chemistry}
Coherent data access using a quantum lookup table has become a key subroutine in modern algorithms for quantum chemistry~\cite{babbush2018EncodingElectronicSpectraLinearT}. These algorithms assume access to a block-encoding of the Hamiltonian. This block-encoding is typically implemented by writing the Hamiltonian as a linear combination of unitaries $H = \sum_{j=0}^{L-1} c_j U_j$ and using oracles of the form:
\begin{align}
       \mathrm{PREPARE}  \ket{0^{\lceil \log_2(L)\rceil }} &= \frac{1}{\sqrt{\lambda}} \sum_{j=0}^{L-1} \sqrt{|c_j|} \ket{j} \\
       \mathrm{SELECT} &= \sum_{j=0}^{L-1} \ketbra{j} \otimes \mathrm{sign}(c_j) U_j  + \sum_{j=L}^{2^{\lceil \log_2(L)\rceil}-1} \ketbra{j} \otimes \Id ,
\end{align}
where $\lambda = \sum_j |c_j|$ is the normalization factor of the block-encoding. The LCU approach to block-encodings still works if $\mathrm{PREPARE}$ results in each computational basis state $\ket{j}$ being entangled with a garbage register. The technique of coherent alias sampling, introduced in \refcite{babbush2018EncodingElectronicSpectraLinearT}, provides an efficient approach for implementing $\mathrm{PREPARE}$ in this way, with a cost of $O\left(L + \log\left(1/\delta \right) \right)$ Toffoli gates using QROM, where $\delta$ is the largest error in $\sqrt{|c_j|}$. The dependence on $L$ can be improved quadratically using the techniques of \refcite{low2018tradingTgatesforDirtyQubits}. In the most straightforward application of these techniques (referred to as ``sparse qubitization''~\cite{Berry2019QubitizationOfArbitraryBasisChemistry}) the complexity of $\mathrm{PREPARE}$ dominates the algorithm, and its cost depends on the cost to load the Hamiltonian coefficients from a quantum lookup table. 

Sparse qubitization requires $O\left(\lambda/\Delta \right)$ calls to the block-encoding of the Hamiltonian, where $\Delta$ is the precision on the energy estimate of the Hamiltonian. For typical $\lambda$ values of small molecules, in the range $10^2$--$10^4$ Hartree, and $\Delta = 10^{-3}$ Hartree, this implies at least $10^5$ calls to the quantum lookup table. If we were to replace the calls to the quantum lookup table with our distillation-based QRAM scheme for the example of FeMo-co ($\lambda=\num{7614}$, $N=L=\num{179498}$, $b=84$~\cite{Berry2019QubitizationOfArbitraryBasisChemistry}), using the same methodology as in the previous section, we conclude that we would require at least $2.5 \times 10^{11}$ Toffoli gates per call to the block-encoding, substantially larger than the value of $10^4$ Toffolis in \refcite{Berry2019QubitizationOfArbitraryBasisChemistry}. We note that even if the dependence of our approach on the error $\varepsilon$ per query could be reduced to $O(\log(1/\varepsilon))$, it is unclear whether our approach would improve over existing methods, as the amount of data (i.e., $Lb$) loaded for the chemistry Hamiltonian is sufficiently modest that the QROM scaling of $O(\sqrt{Lb})$ is comparable to the scaling of our protocol $O(\log^2(L)(\log(L)+b))$ (\cref{thm:total_complexity_b_bits}). 

The sparse qubitization approach has been superseded in some applications by more efficient methods~\cite{burg2021QuantumComputingEnhancedComputationalCataylysis,lee2021EvenMoreEfficientChemistryTensorHyp}, which also use a lookup table to coherently load angles for basis rotations that are used in the $\mathrm{SELECT}$ oracle. We refer to Refs.~\cite{lee2021EvenMoreEfficientChemistryTensorHyp,kim2022FaultTolerantQuantumChemicalSimulationsLiIon} for a detailed accounting of the contribution of the lookup table reads to the total gate and space complexity. It is unlikely that our aproach to implementing QRAM would reduce the costs of these algorithms, at least in its current form.

\section{Outlook on the cheap QRAM assumption}\label{sec:outlook}

A central goal of this research program is to determine whether QRAM can be considered equally cheap as RAM---at least in an abstract, asymptotic sense---or whether its quantum nature makes QRAM 
fundamentally more difficult, preventing a justification of the cheap QRAM assumption from \cref{sec:introduction}. 
The central aspect of QRAM that differentiates it from RAM is the need to protect the quantum information about which address (or superposition of addresses) is being queried, even when the hardware has errors. To emphasize this distinction, it is worth noting that a fault-tolerant RAM could be easily constructed from a faulty RAM device using a simple repetition code: by repeatedly querying the faulty RAM device on the same input address, one can efficiently boost the probability of a successful RAM query, provided that as the device has nonzero bias in favor of the correct answer. In contrast, this same strategy fails for QRAM because even a single query to a faulty QRAM device may leak the address information and decohere the address register. Formally speaking, we might say that logical (Q)RAM is a transversal gate for the repetition code, but that this is not sufficient for fault-tolerant QRAM, since the repetition code does not protect against phase-flip errors. For fault-tolerant QRAM, a more sophisticated strategy is required.

Our protocol provides such a method for fault-tolerant QRAM, successfully protecting which address state $\sum_{x} \alpha_x \ket{\ol{x}}$ is being queried without performing QEC on the $\Omega(2^n)$ components of the QRAM device. It does this by relying on resource states (see \cref{eq:QRAM_resource_state}) that are equal superpositions of all $2^n$ addresses, independent of which superposition of addresses one wants to query---the noisy device is prevented from direct interaction with the address information. However, these resource states have exponentially small amplitude on any individual address. Consequently, if the data at one address is modified, the ideal resource state barely changes. This begs the question: if these resource states are insensitive to the underlying bits in the dataset, how, then, can they be used to insert information about the dataset into the quantum state? Our protocol achieves this by iteratively and adaptively inserting \textit{global} information about the dataset $f$---that is, properties of $f$ that depend on all $2^n$ data entries---into the quantum state. One way to see this is by decomposing $f(x)$ into a sum over global, oscillatory contributions of the form $(-1)^{k \cdot x}$, that is, its Fourier expansion
\begin{align}\label{eq:frequency_decomp}
    f(x) = \frac{1}{2^{n/2}}\sum_{k \in \{0,1\}^n} (-1)^{k \cdot x}\tilde{f}(k),
\end{align}
where $\tilde{f}$ is the Walsh--Hadamard transform of $f$. 
When we teleport the resource state $\ket{\ol{\Psi(f)}}$, we enact the QRAM unitary $\ol{V(f^{\oplus m})}$ instead of $\ol{V(f)}$, where $m$ is uniformly random and $f^{\oplus m}$ is defined by $f^{\oplus m}(x) = f(x \oplus m)$. While the datasets $f^{\oplus m}$ and $f$ are not guaranteed to agree on any of the addresses, they have a close relationship in frequency space. One can see from \cref{eq:frequency_decomp} that
\begin{align}
    \tilde{f}^{\oplus m}(k) = (-1)^{k \cdot m} \tilde{f}(k)\,.
\end{align}
Thus, $\tilde{f}^{\oplus m}(k)= \tilde{f}(k)$ for half the values of $k$ and $\tilde{f}^{\oplus m}(k)= -\tilde{f}(k)$ for the other half  (except in the unlikely case that $m = 0^n$, in which case they would agree on all values). In a sense, we may say that by implementing $\ol{V(f^{\oplus m})}$, we have successfully inserted half of the information about the function $f$ into the quantum state, regardless of which random $m$ is obtained. 
To insert the other half of the information, the protocol updates the dataset to $f' = f \oplus f^{\oplus m}$. We may observe that $f'$ is periodic in translation by $m$ (i.e., $f'(x) = f'(x\oplus m)$ for all $x$), and thus $\tilde{f'}(k) = 0$ for any $k$ for which $k \cdot m = 1$ (half the values of $k$). Each successive correction function will have half as many nonzero Fourier components as the previous one (assuming the $n$-bit measurement outcomes in previous rounds form a linearly independent set), until finally after $n$ rounds there is no remaining frequency information left to insert. This global approach to information insertion appears crucial for correctly applying QRAM on any input state, even while not knowing or learning what that input state is. 

The cost of this global approach is adaptivity and classical computation. At each iteration, we do not have control over which half of the frequency information we insert into the quantum state; this is determined by a uniformly random measurement outcome $m$. After receiving $m$, the protocol must adaptively update the \textit{entire} dataset to essentially remove the half of the global information that has already been applied to the quantum state. This removal requires touching all $2^n$ entries of the dataset, and then giving the physical QRAM device access to the new dataset. As discussed in \cref{sec:complexity_of_classical_update},  this updating of frequency information is a non-negligible classical computation, essentially equivalent in complexity to performing the Walsh--Hadamard transformation on the dataset.  While the Walsh--Hadamard transform may be parallelizable, it appears to be a harder calculation than a normal RAM query; for example, it requires fundamentally greater wire density than RAM.

The main theoretical open question, then, is to determine if the classical complexity of the protocol can be reduced to a quantity more similar to a RAM query, or otherwise find a way to justify that the overall (i.e., both classical and quantum) cost of QRAM is $\poly(n)$. Our protocol makes some progress in this direction, and it offers clear advantages over actively error corrected circuit QRAM at the practical level---enabling the usage of a specialized low-fidelity QRAM device, and eliminating the need for massively parallel QEC. However, in a theoretical sense, both our protocol and circuit QRAM ultimately require the same $\Omega(2^n)$ scaling of classical resources to protect the address information from decoherence.
One is left to wonder whether this scaling of classical complexity could be a fundamental requirement for achieving fault-tolerant QRAM.

\section*{Acknowledgments}
The authors thank Mario Berta,  Joe Iverson, Sam Jaques, Michael Kastoryano, Fernando Pastawski, Samson Wang, and John Wright for helpful conversations. We
also thank Simone Severini,  James Hamilton, Nafea Bshara, Peter DeSantis, and Andy Jassy for their involvement and support of the
research activities at the AWS Center for Quantum
Computing.

\appendix 

\section{Generalization of the protocol to multiple output bits}\label{app:extension_to_b_bits}

The main text provided a protocol to implement the single-bit diagonal QRAM operation of \cref{eq:QRAM_intro}, where the classical data table $f$ is an $n$-bit Boolean function  consisting of a single bit stored at each of $2^n$ addresses. The operation applies a sign to each corresponding basis state $\ket{x}\mapsto (-1)^{f(x)} \ket{x}$. 

In many applications, there are $b > 1$ bits stored at each of the $2^n$ addresses in memory, and often we wish to coherently read all $b$ bits into a separate bus register, that is, perform the operation (cf.~\cref{eq:QRAM_intro})
\begin{align}\label{eq:QRAM_b_bits}
    \sum_{x \in \{0,1\}^n} \sum_{u \in \{0,1\}^b} \alpha_{x,u} \ket{x} \ket{u} \quad \overset{\raisebox{3pt}{$\mathlarger{U(f)}$}}{\longmapsto} \quad \sum_{x \in \{0,1\}^n} \sum_{u \in \{0,1\}^b} \alpha_{x,u} \ket{x}  \ket{u \oplus f(x)}
\end{align}
We refer to the first register storing $\ket{x}$ as the address register, and the second register storing the data as the bus register. 
Our protocol can also straightoforwardly handle this case, although it requires introducing a bit more notation and slightly more overhead. 

\subsection{Definitions}
First, let $\CF_n$ denote the set of $n$-bit Boolean functions $f \colon \{0,1\}^n \rightarrow \{0,1\} $ (considered as data tables in the main text). Next, let $\CF_n^{(b)}$ denote the set of $n$-bit functions with $b$ bits of output $f\colon \{0,1\}^n \rightarrow \{0,1\}^b$. We wish to implement the QRAM operation $U(f)$ of \cref{eq:QRAM_b_bits} for arbitrary data tables $f \in \CF_n^{(b)}$. 

 To do so, we will need to generalize the set $\CF_n^{(b)}$ to include an extra ``sign'' bit.  These are essentially equivalent to $n$-bit functions with $b+1$ bits of output, that is, the set $\CF_n^{(b+1)}$, but where the sign bit plays a different role than the other $b$ bits. Formally, we let $\CF_{n}^{(b\pm)} = \CF_{n} \times \CF_{n}^{(b)}$, and  we denote elements of $f \in \CF_{n}^{(b\pm)}$ by a pair $f = (f_{\pm},f_{\Hsquare})$, where $f_{\pm}$ is an $n$-bit Boolean function that dictates the ``sign'' and $f_{\Hsquare}$ is an $n$-bit function with $b$ bits of output.

For any $f \in \CF_n^{(b\pm)}$, we define a unique Boolean function $\hat{f} \in \CF_{n+b}$ to act on inputs of $n+b$ bits denoted by pairs $(x,u)$ with $x \in \{0,1\}^n$ and $u \in \{0,1\}^b$:
\begin{align}\label{eq:hat_f}
    \hat{f}(x,u) &= f_{\pm}(x) \oplus \left[\bigoplus_{i=1}^b u_i f_i(x)\right] = f_{\pm}(x) \oplus (u \cdot f_{\Hsquare}(x))\,,
\end{align}
where $f_i$ is the $i$-th output bit of $f_{\Hsquare}$, and in the final expression, $\cdot$ denotes the standard dot product (modulo 2) between a pair of $b$-bit strings, viewed as vectors.
We can see that any $\hat{f}$ obtained from $f$ via \cref{eq:hat_f} contains all the same information as $f$, just represented in a different form. 
We can also see that
\begin{align}\label{eq:deg_hat_f}
    \deg(\hat{f}) = \max\left(\deg(f_{\pm}), 1+ \max_{i \in [b]} \deg(f_i)\right) \leq n+1\,.
\end{align}
This is an important observation---since $\hat{f}$ is a Boolean function on $n+b$ bits, one would naively expect its degree could be as large as $n+b$, but due to the structure of \cref{eq:hat_f}, the actual degree is independent of $b$.

Having defined signed Boolean functions, we can generalize the standard QRAM operation $U(f)$ of \cref{eq:QRAM_b_bits} slightly by allowing the QRAM to also conditionally apply signs to the computational basis states, controlled by another Boolean function $f_{\pm}$. Specifically, given a signed $n$-bit function $f = (f_{\pm},f_\Hsquare) \in \CF_n^{(b\pm)}$, we define
\begin{align}\label{eq:QRAM_operation_with_signs}
    U(f)\ket{x}\ket{u} = (-1)^{f_{\pm}(x)} \ket{x} \ket{u \oplus f_{\Hsquare}(x)}\,.
\end{align}
We may think of $f \in \CF_n^{(b\pm)}$ as a data table storing $b+1$ bits of data at each of $2^n$ addresses. In fact, we may think of the main text as the special case of \cref{eq:QRAM_operation_with_signs} where $b=0$ and there is only a sign bit. 

In this appendix, we will be working with the more general signed QRAM function from \cref{eq:QRAM_operation_with_signs} and assuming $f \in \CF_n^{(b\pm)}$; we may always take $f_{\pm} = \mathbf{0}$ (the zero function) to recover the standard QRAM from \cref{eq:QRAM_b_bits}. 

\subsection{Generalized diagonal QRAM}
The key to implenting $U(f)$ will be to convert it to a diagonal operation generalizing $V(f)$ from \cref{eq:QRAM_intro}, so that the same approach as was used in the main text can be applied. We define $V(f)$ for $f \in \CF_n^{(b\pm)}$ to be equal to the QRAM operation $U(f)$ from \cref{eq:QRAM_operation_with_signs} conjugated by  the Hadamard transform on the bus register
\begin{align}\label{eq:diagonal_QRAM}
    V(f) = (\Id \otimes H^{\otimes b}) \, U(f) \, (\Id \otimes H^{\otimes b})\,.
\end{align}
A straightforward calculation shows that $V(f)$ is diagonal, with the diagonal entries given by $(-1)^{\hat{f}(z)}$ for $z$ a $2^{n+b}$-bit string: 
\begin{align}\label{eq:diagonal_QRAM_hat_f}
   V(f)\ket{x} \ket{u} = (-1)^{\hat{f}(x,u)} \ket{x} \ket{u}\,.
\end{align}
This fact is what motivated the definition of $\hat{f}$ in \cref{eq:hat_f}. Thus, this definition of $V(f)$ is seen to be an immediate generalization of \cref{eq:QRAM_intro}, with the function $f$ replaced by $\hat{f}$. 

As one example, if $b=1$ and $f_{\pm} = \mathbf{0}$, then the function $\hat{f}(x,u) = u_1f_{\Hsquare}(x)$ is an $(n+1)$-bit Boolean function. The operation $V(f)$ using the definition of $V(f)$ from \cref{eq:diagonal_QRAM_hat_f} is equivalent to $V(\hat{f})$ for the definition of $V(f)$ from \cref{eq:QRAM_intro}. Furthermore, the operation can be thought of as a controlled $V(f_{\Hsquare})$ operation, with the single bus qubit acting as the control, as alluded to in \cref{footnote:controlled-V(f)}. 

We will also need to generalize the update rule of \cref{eq:update_rule}. Suppose we modify the function $\hat{f}$ by adding (modulo 2) a fixed bit string $m \in \{0,1\}^{n+b}$ to the address and bus registers before evaluating the function. We may decompose $m = (m_A,m_B)$ where $m_A \in \{0,1\}^n$ and $m_B \in \{0,1\}^b$.  We observe that
\begin{align}
    \hat{f}(x \oplus m_A, u \oplus m_B)
    =    f_{\pm}(x\oplus m_A)  \oplus \left[\bigoplus_{i=1}^{b} (u_i\oplus m_{B,i}) f_i(x\oplus m_A)\right] \,.
\end{align}
Next, for any signed $n$-bit function $f = (f_{\pm}, f_1,\ldots, f_b) \in \CF_{n}^{(b\pm)}$ and any $m \in \{0,1\}^{n+b}$, we may define a unique function $f^{\oplus m} = (f^{\oplus m}_{\pm}, f^{\oplus m}_1,\ldots, f^{\oplus m}_b) \in \CF_{n}^{(b\pm)} $ by the relations
\begin{align}
\label{eq:f_oplus_m_generalized}
\begin{split}
    f_{\pm}^{\oplus m}(x) &= f_{\pm}(x\oplus m_A)  \oplus \left[\bigoplus_{i=1}^{b} m_{B,i} f_i(x\oplus m_A)\right]   \\
    f_{i}^{\oplus m}(x) &= f_i(x\oplus m_A) \qquad \qquad \text{for $i = 1,\ldots, b$}\,.
    \end{split}
\end{align}
This definition generalizes the one from \cref{eq:g_oplus_m}, and it was chosen to ensure that
\begin{align}\label{eq:hat_f_correct_behavior}
    \hat{f}(z \oplus m) = \hat{f}^{\oplus m}(z)\,.
\end{align}

\subsection{Protocol for generalized QRAM}

It is now simple to apply the ideas from the main text to fault-tolerantly implement the generalized $\ol{V(f)}$ from \cref{eq:diagonal_QRAM_hat_f}, which amounts to applying a sign $(-1)^{\hat{f}(x,u)}$ to each basis state $\ket{\ol{x}}\ket{\ol{u}}$.  The unitary $\ol{U(f)}$ can thus be implemented by conjugating $\ol{V(f)}$ with fault-tolerant Hadamard operations on the bus register (\cref{eq:diagonal_QRAM}). 

To implement the sign $(-1)^{\hat{f}(x,u)}$, we could simply view $\hat{f}$ as a Boolean function on $n+b$ bits of degree at most $n+1$ (\cref{eq:deg_hat_f}). We could store all $2^{n+b}$ outputs of this function $\hat{f}$ in classical memory, and apply the protocol from the main text. The drawback of this approach is that it does not take advantage of the structure of the function $\hat{f}$, which has only $(b+1)2^n$ binary degrees of freedom rather than $2^{n+b}$. If $b$ is large it would be classically intractable to store all $2^{n+b}$ outputs of $\hat{f}$ in classical memory, and to perform the update rule directly as $g \gets \UR(g,m)$ when the data table $g$ has $2^{n+b}$ classical bits. 

Fortunately, we may apply our protocol without storing and manipulating $2^{n+b}$ classical bits. Rather than keeping track of $\hat{g} \in \CF_{n+b}$ explicitly and applying an update rule $\hat{g} \gets \UR(\hat{g},m)$,  we simply maintain a more compact and natural representation of $g = (g_{\pm}, g_{\Hsquare})$ without any redundancy. We preserve the form of the update rule from \cref{eq:update_rule}
\begin{align}
    \UR(g,m) = g \oplus g^{\oplus m}
\end{align}
except that we interpret $g^{\oplus m}$ using \cref{eq:f_oplus_m_generalized}. Due to \cref{eq:hat_f_correct_behavior}, this update rule ensures that if $h=\UR(g,m)$ then $\hat{h}(z) = \hat{g}(z) \oplus \hat{g}(z \oplus m)$, even though we are not storing all $2^{n+b}$ outputs of $\hat{h}$ in classical memory. Indeed, by examining \cref{eq:f_oplus_m_generalized}, we see that we may implement the update rule by applying the original update rule of \cref{eq:update_rule} to each of the $b+1$ bits of $g$ using the bit string $m_A \in \{0,1\}^n$, and then additinally updating $g_\pm$ by adding (modulo 2)  the $i$-th bit of $g_{\Hsquare}$ to $g_{\pm}$ for each $i$ where the $i$-th bit of $m_B$ is 1.

Due to the fact that the degree of $\hat{g}$ is at most $n+1$, the number of rounds required is at most $n+1$, regardless of the size of $b$. 

Besides the update rule, the protocol proceeds as if the function being applied were an $(n+b)$-bit function, rather than an $n$-bit function. The resource states created by the oracle are $(n+b)$-qubit physical states (cf.~\cref{eq:QRAM_resource_state})
\begin{align}
    \ket{\Psi(g)} &= \frac{1}{2^{(n+b)/2}} \sum_{x \in \{0,1\}^n}\sum_{u \in \{0,1\}^b} (-1)^{\hat{g}(x,u) }\ket{x} \ket{u} \label{eq:physical_resource_state} \\
    &= \frac{1}{2^{(n+b)/2}} \sum_{z \in \{0,1\}^{n+b}} (-1)^{\hat{g}(z) }\ket{z} 
\end{align}
which can be created by applying the physical unitary $V(g)$ onto the state $\ket{+}^{\otimes(n+b)}$. The encoding protocol must now encode $n+b$ physical qubits into $n+b$ logical qubits, and the encoding error grows linearly in $n+b$ as $\varepsilon_{\rm enc} = O(\sqrt{p}(n+b))$ or better (straightforward generalization of the results of \cref{sec:encoding}). The distillation protocol based on the iterated swap test or quantum PCA with fractional swap gates proceeds identically, except that each controlled swap operation now requires $n+b$ fault-tolerant controlled swap gates on qubit systems, instead of only $n$. The teleportation procedure also now requires $n+b$  CNOT gates, rather than only $n$. As a result, we can state the following generalized theorem. 

\begin{theorem}\label{thm:total_complexity_b_bits}
    Let $f \in \CF_n^{(b\pm)}$ be an arbitrary signed $n$-bit function with $b$ output bits, for which we wish to implement the fault-tolerant diagonal QRAM unitary $\ol{U(f)}$ of \cref{eq:QRAM_operation_with_signs}, and let $\varepsilon$ be an error parameter. The same results hold as in \cref{thm:total_complexity}, provided that the physical QRAM device can perform the physical version of the generalized QRAM operation of \cref{eq:diagonal_QRAM_hat_f} on $n+b$ qubits with the same guarantees, and that the physical error rate of the main processor satisfies $p(n+b)^2 = O(1)$. The complexity of the protocol is the same with an identical asymptotic expression for $Q$, and
    \begin{align}
        Q' &= O\left((n+b)^2Q\right)\,.
    \end{align}
    Additionally, the protocol applies the (generalized) classical update rule $n+1$ times (rather than only $n$), and the (generalized) classical partial Clifford twirling operation $g \mapsto g_C$ (\cref{eq:g_C}) $Q$ times, each time on a data table of size $(b+1)2^n$ classical bits. 
\end{theorem}
\begin{proof}
    This follows straightforwardly from the observation that the same number of resource states are needed, but the Clifford twirling, distillation, and teleportation steps now act on states of $n+b$ qubits rather than only $n$. 
\end{proof}

\section{Delayed proofs for encoding error}\label{sec:delayed_proofs_encoding}
Here we restate and provide the full proofs of the two propositions in \cref{sec:encoding}. 
\pauliTwirlEncoding*
\begin{proof}
    There are broadly two components to this proof. First, we will define $\CE_{\rm FT}$ based on a Pauli twirl of $\CE'_{\rm FT}$ and show that it 
    is a stochastic channel with a large identity component, up to a small correction. Second, we will show how the channel being stochastic implies the claimed fidelity statement. 

    We begin with the Pauli twirl. Formally, we let $\PauliSetNoSign$ 
    denote the set of $n$-qubit Pauli operators, that is, the subset of the signed Pauli set $\PauliSet$ defined in \cref{def:signed_Pauli_strings}, where the sign bit $s$ is fixed to $s=0$. We define $\CE_{\rm FT}$ to be the procedure that (i) chooses a Pauli $G \in \PauliSetNoSign$ uniformly at random, (ii) applies physical $G$, (iii) applies $\CE_{\rm FT}'$ to encode the state, (iv) applies the fault-tolerant QEC gadget $\CQ_{\rm FT}$, (v) applies the fault-tolerant gadget for logical $\ol{G}$ on the encoded state. In the absence of noise, the physical $G$ and logical $\ol{G}$ will cancel out, and encoding will be perfect. The purpose of including them is to be able to guarantee that in the presence of noise the overall channel is a stochastic channel, as we will explain shortly.  Let the physical channel  implemented in step (ii) be denoted by $\tCG$, and let the channel enacted in step (v) on the encoded system be denoted by $\tCG_{\rm FT}$. The channel implemented by the noisy encoding procedure is thus given by an average over all choices of $G$ (denoted $\EV_{G \sim \PauliSetNoSign}$), as follows
    \begin{align}
    \tCE_{\rm FT} = \EV_{G \sim \PauliSetNoSign} \tCG_{\rm FT} \circ  \tCQ_{\rm FT} \circ \tCE'_{\rm FT} \circ \tCG \,.
    \end{align}
Since we have assumed circuit-level stochastic noise (\cref{def:stochastic-noise}) with strength $p$, and the physical $G$ has $n$ circuit locations, we may write
\begin{align}\label{eq:Pauli_twirl_gate_decomp}
    \tCG = (1-p)^n\CG + (1-(1-p)^n) \CN_G
\end{align} 
for some CPTP map $\CN_G$, where $\CG[\cdot] = G[\cdot]G^\dag$ is the ideal channel (note that for $G \in \PauliSetNoSign$, we have $G = G^\dag$). 
Furthermore, since $\CG_{\rm FT}$ is fault-tolerant, we have that 
\begin{align}
    \frac{1}{2}\norm{\CQ \circ \tCQ_{\rm FT} \circ \tCG_{\rm FT} \circ \tCQ_{\rm FT} - \ol{\CG} \circ \CQ \circ \tCQ_{\rm FT}}_{\diamond} \leq \Gamma(\CE)\,,
\end{align}
where $\Gamma(\CE)$ vanishes with increasing code size, as in \cref{eq:fault-tolerant_gate_error}.  
Let the superoperator inside the diamond norm in the equation above be denoted by $\CR$, which satisfies $\nrm{\CR}_{\diamond} \leq 2 \Gamma(\CE)$. Combining the above equations, we may write
\begin{align}
    \CQ \circ \tCQ_{\rm FT} \circ \tCE_{\rm FT} &= \EV_{G \sim \PauliSetNoSign} \Big[(\CQ \circ \tCQ_{\rm FT} \circ  \tCG_{\rm FT} \circ \tCQ_{\rm FT}) \circ \tCE'_{\rm FT} \circ \tCG  \Big]
    = \EV_{G \sim \PauliSetNoSign} \Big[(\CR + \ol{\CG} \circ \CQ \circ \tCQ_{\rm FT}) \circ \tCE'_{\rm FT} \circ \tCG  \Big]\\
    &= (1-p)^n\left(\CR_1 +  \EV_{G \sim \PauliSetNoSign} \ol{\CG} \circ (\CQ \circ \tCQ_{\rm FT}  \circ \tCE'_{\rm FT}) \circ \CG \right)  + (1-(1-p)^n) \CR_2 
\end{align}
where $\CR_1 = \EV_{G \in \PauliSetNoSign} \CR \circ \tCE'_{\rm FT} \circ \CG$, which also satisfies $\norm{\CR_1}_{\diamond} \leq 2\Gamma(\CE)$, and the term $\CR_2$ is a CPTP map that collects the contribution $\CN_G$ from \cref{eq:Pauli_twirl_gate_decomp}. This equation is a stochastic mixture of two CPTP channels (viewing the term in parentheses as one channel and $\CR_2$ as the other), and each contributes non-negatively to the overall fidelity. Thus, recalling that $\ol{\phi(g)} = \CQ \circ \tCQ_{\rm FT} \circ \tCE_\mathrm{FT}[\tpsig{g}]$,
we can ignore the contribution of the second term and say
\begin{align}
    \bra{\ol{\Psi(g)}} \ol{\phi(g)} \ket{\ol{\Psi(g)}} &= \bra{\ol{\Psi(g)}} \CQ \circ \tCQ_{\rm FT} \circ \tCE_{\rm FT}[\tpsig{g}] \ket{\ol{\Psi(g)}} \\
    &\geq (1-p)^n \bra{\ol{\Psi(g)}}\EV_{G \sim \PauliSetNoSign} \ol{G}\left(\CQ \circ \tCQ_{\rm FT} \circ \tCE'_{\rm FT}[G \tpsig{g} G]\right)\ol{G} \ket{\ol{\Psi(g)}} - (1-p)^n 2\Gamma(\CE)\,, \label{eq:bound_on_Fidelity_with_correction}
\end{align}
since $\bra{\ol{\Psi(g)}} \CR_1[\tpsig{g}]\ket{\ol{\Psi(g)}}  \leq \norm{\CR_1[\tpsig{g}]}_1 \leq \norm{\CR_1}_{\diamond} $. 

We may generically decompose the channel $\CQ \circ \tCQ_{\rm FT} \circ \tCE'_{\rm FT}[\cdot] = \sum_{P,P' \in \PauliSetNoSign} \chi_{P,P'} \ol{P}\CE[\cdot] \ol{P}'$ as a sum over logical Pauli operators acting on the left and right of the encoded state, weighted by numbers $\chi_{P,P'}$ arranged into a matrix. Now, defining the channel $\CN[\rho] = \EV_{G \sim \PauliSetNoSign} \ol{G} (\CQ \circ \tCQ_{\rm FT} \circ \tCE'_{\rm FT}[G \rho G])\ol{G}$, we have
\begin{align}
    \CN[\rho] = \sum_{P,P' \in \PauliSetNoSign} \chi_{P,P'}\EV_{G \sim \PauliSetNoSign} \ol{G} \, \ol{P} \CE[G \rho G] \ol{P'} \, \ol{G} =  \sum_{P,P' \in \PauliSetNoSign} \chi_{P,P'}\EV_{G \sim \PauliSetNoSign} \ol{G} \, \ol{P} \, \ol{G}\CE[\rho] \ol{G}\, \ol{P'} \, \ol{G}\,,
\end{align}
where we have noted that the physical $G$ acting prior to perfect encoding is equivalent to logical $\ol{G}$ acting after perfect encoding. 
Examining one term from this expansion, we have 
\begin{align}
   \EV_{G \sim \PauliSetNoSign} \ol{G} \, \ol{P} \, \ol{G} \,\CE[\rho ] \,\ol{G} \, \ol{P}' \, \ol{G} = \delta_{P,P' } \ol{P} \,\CE[\rho]\, \ol{P}'
\end{align}
where $\delta_{P,P'}$ is the Kronecker delta that equals 1 if and only if $P = P'$. This is true because $\ol{G}\, \ol{P}\, \ol{G} = \pm \ol{P}$ for every $G$, and if $P \neq P'$, then exactly 1/2 of the choices for $G$ will lead to a net positive sign and 1/2 will lead to a  net negative sign.  

Since the ``offdiagonal'' terms with $P \neq P'$ vanish, this establishes that the channel $\CN$ is a stochastic Pauli channel and we may write it as $\CN[\rho] = (1-\delta) \CE[\rho]  + \sum_{P \in \PauliSetNoSign\setminus\{ \Id\}} \chi_{P,P} \ol{P} \CE[\rho] \ol{P} $, where $1-\delta = \chi_{\Id,\Id}$. Since $\CN$ is CPTP by construction, we have $\chi_{P,P} \geq 0$ for all $P$ \cite[Lemma 5.2.4]{dankert2005efficientSimulationRandomQuantum} and $\sum_{P \in \PauliSetNoSign} \chi_{P,P} = 1$, implying $0 \leq \delta \leq 1$. 
We would now like to bound its contribution to the fidelity. We have
    \begin{align}
    \bra{\ol{\Psi(g)}} \CN[\tpsig{g}] \ket{\ol{\Psi(g)}} &= (1-\delta) \bra{\ol{\Psi(g)}} \CE[\tpsig{g}] \ket{\ol{\Psi(g)}} + \sum_{P \in \PauliSetNoSign\setminus\{\Id\}} \chi_{P,P} \bra{\ol{\Psi(g)}} \ol{P} \, \CE[\tpsig{g}] \, \ol{P} \ket{\ol{\Psi(g)}} \\
     &= (1-\delta) \bra{\Psi(g)} \tpsig{g} \ket{\Psi(g)} + \sum_{P \in \PauliSetNoSign\setminus\{\Id\}} \chi_P \bra{\Psi(g)} P \,\tpsig{g} \,P \ket{\Psi(g)} \\
    &\geq (1-\delta) \bra{\Psi(g)} \tpsig{g} \ket{\Psi(g)} = (1-\delta) F(g)_{\rm phys}\,. \label{eq:fidelity_in_terms_of_delta}
\end{align}
   
In the second part of the proof, we aim to bound $\delta$ in terms of $\varepsilon_{\rm enc}$ by using the fact that $\CN$ is close to the identity channel. By assumption, $\CE_{\rm FT}'$ has encoding error $\varepsilon_{\rm enc}$ (recall \cref{def:encoding-error}). We begin by observing that the twirled channel $\CN$ can only have smaller encoding error than $\CE_{\rm FT}'$. 
    \begin{align}
        \varepsilon_{\rm enc} &= \sup_{\rho} \frac{1}{2}\norm{\CQ \circ \tCQ_{\rm FT} \circ \tCE'_{\rm FT}[\rho] - \ol{\rho}}_1 = \sup_{\rho} \frac{1}{2}\norm{\ol{\CG} \circ \CQ \circ \tCQ_{\rm FT} \circ \tCE'_{\rm FT}[\rho] - \ol{G} \, \ol{\rho} \, \ol{G} }_1 \\
        &=  \sup_{\rho} \frac{1}{2}\norm{\ol{\CG} \circ \CQ \circ \tCQ_{\rm FT} \circ \tCE'_{\rm FT}[G \rho G] - \ol{\rho} }_1 = \EV_{G \in \PauliSetNoSign} \sup_{\rho} \frac{1}{2}\norm{\ol{\CG} \circ \CQ \circ \tCQ_{\rm FT} \circ \tCE'_{\rm FT}[G \rho G] - \ol{\rho} }_1 \\
        &\geq  \sup_{\rho} \frac{1}{2}\norm{\EV_{G \in \PauliSetNoSign} \ol{\CG} \circ \CQ \circ \tCQ_{\rm FT} \circ \tCE'_{\rm FT}[G \rho G] - \ol{\rho} }_1  = \sup_{\rho} \frac{1}{2}\norm{\CN[\rho] - \ol{\rho}}_1 
        \end{align}
        Then, we may substitute the Pauli form of $\CN$ to evaluate
        \begin{align}
            \varepsilon_{\rm enc} &\geq \sup_{\rho} \frac{1}{2}\norm{- \delta \ol{\rho} + \sum_{P \in \PauliSetNoSign\setminus\{\Id\}} \chi_{P,P} \ol{P} \ol{\rho} \ol{P}}_1 
            \geq \EV_{\ket{\psi} \sim \mathrm{Haar}} \frac{1}{2}\norm{- \delta \ketbra{\ol{\psi}}+ \sum_{P \in \PauliSetNoSign\setminus\{ \Id\}} \chi_{P,P} \ol{P} \ketbra{\ol{\psi}} \ol{P}}_1 \,
    \end{align}
where expectation value denotes drawing $\ket{\psi}$ uniformly from the Haar measure over $n$-qubit states. 
Generally, for any Hermitian operator $M$ and any pure state $\ket{\xi}$, we have $\nrm{M}_1 \geq |\bra{\xi} M \ket{\xi}|$. Thus, we may write
\begin{align}
    \varepsilon_{\rm enc} &\geq \EV_{\ket{\psi} \sim \mathrm{Haar}} \frac{1}{2}\left\lvert \bra{\ol{\psi}}\Big(- \delta \ketbra{\ol{\psi}}+ \sum_{P \in \PauliSetNoSign\setminus\{ \Id\}} \chi_{P,P} \ol{P} \ketbra{\ol{\psi}} \ol{P}\Big)\ket{\ol{\psi}}\right \rvert \\
    &\geq \frac{1}{2}\delta - \sum_{P \in \PauliSetNoSign\setminus\{ \Id\}}\frac{\chi_{P,P}}{2}\EV_{\ket{\psi} \sim \mathrm{Haar}} \Tr\left((\ketbra{\ol{\psi}} \otimes \ketbra{\ol{\psi}})(\ol{P} \otimes \ol{P}) \right) \\
    &= \frac{1}{2}\delta - \sum_{P \in \PauliSetNoSign\setminus\{ \Id\}}\frac{\chi_{P,P}}{2} \frac{1}{2^n+1} = \frac{1}{2}\delta - \frac{1}{2(2^n+1)}\delta = \frac{2^n}{2^n+1}\frac{\delta}{2} \geq \frac{\delta}{3} 
\end{align}
where we have used that the expectation over the Haar measure of $\ketbra{\ol{\psi}}^{\otimes 2}$ is $(\ol{\Id} \otimes \ol{\Id} + \ol{\Swap})/(d(d+1))$, with $\Swap$ the swap operator and $d = 2^n$ \cite{Mele2024introductionToHaar}. Then, we have evaluated $\Tr(\ol{P} \otimes \ol{P}) = 0$ and $\Tr(\ol{\Swap}(\ol{P} \otimes \ol{P})) = \Tr(\ol{\Id}) = d $. This confirms that $\delta \leq 3 \varepsilon_{\rm enc}$.

To conclude, we combine the fact that $\delta \leq 3 \varepsilon_{\rm enc}$ with \cref{eq:bound_on_Fidelity_with_correction} and \cref{eq:fidelity_in_terms_of_delta} to say
\begin{align}
    \bra{\ol{\Psi(g)}} \ol{\phi(g)} \ket{\ol{\Psi(g)}}  \geq (1-p)^n\Big((1-3\varepsilon_{\rm enc})F(g)_{\rm phys} - 2\Gamma(\CE)\Big)  
\end{align}
which implies the proposition statement. 
\end{proof}

\encodingErrorFT*

\begin{proof}
We will use the fault-tolerant quantum input-output framework introduced by \textcite{christandl2024faultTolerantQuantumInputOutput}, which tells us that, intuitively, a quantum circuit can be implemented fault-tolerantly on a noisy device except at the input and output layers of the circuit, during which correlated errors can occur with a strength that depends on the underlying circuit-level noise model.
Specifically, \refcite[Theorem 59]{christandl2024faultTolerantQuantumInputOutput} states that, given an ideal quantum circuit of size $|\CE|$ that implements map $\CE$, we can efficiently find a fault-tolerant circuit $\CE_\mathrm{FT}$ of size $|\CE|\cdot \poly(k)$ with the same input and output systems, such that its noisy implementation $\tCE_\mathrm{FT}$ subject to circuit-level stochastic noise with parameter $p$ (\cref{def:stochastic-noise}) satisfies
\begin{equation}
\norm{\tCE_\mathrm{FT}- \left(\bigotimes_{j=1}^{n'} \CW_j\right) \circ (\CE \otimes \CF) \circ \left(\bigotimes_{i=1}^{n} \mathcal{N}_i\right)}_{\diamond} \leq 4|\CE|(cp)^k,
\label{eq:christandl_ineq}
\end{equation}
for some absolute constants $c > 0$, integer $k \geq 1$, and a suitable environment channel $\CF$, and where $n$ and $n'$ are the number of qubits at the input and output of $\CE$, respectively.
Here, each $\mathcal{N}_i$ and $\CW_j$ informally represent local noise channels acting on the input and output qubits, respectively, each of which introduces an error rate bounded by $\varepsilon_0 \leq 2c p$ for some absolute constant $c$, in the sense specified by \cref{eq:Ni} and \cref{eq:Wi} below.
Therefore, we have
\begin{align}
    \varepsilon_{\rm enc} =\ & \frac{1}{2}\sup_{\rho}\norm{\CQ \circ\tCQ \circ\tCE_\mathrm{FT}[\rho] - \CE[\rho]}_1\\
    \leq\ & \frac{1}{2}\sup_{\rho} \norm{\CQ\circ\tCQ\circ\tCE_\mathrm{FT}[\rho] - \CQ\circ \tCQ \circ\left(\bigotimes_{j=1}^{n'} \CW_j\right) \circ (\CE \otimes \CF) \circ \left(\bigotimes_{i=1}^{n} \mathcal{N}_i\right)[\rho]}_1\nonumber \\
    & + \frac{1}{2}\sup_{\rho}\norm{\CQ\circ\tCQ \circ \left(\bigotimes_{j=1}^{n'} \CW_j\right) \circ (\CE \otimes \CF) \circ \left(\bigotimes_{i=1}^{n} \mathcal{N}_i\right)[\rho] - \CE[\rho]}_1\\
    \leq\ & \frac{1}{2}\sup_{\rho} \norm{\tCE_\mathrm{FT}[\rho] - \left(\bigotimes_{j=1}^{n'} \CW_j\right) \circ (\CE \otimes \CF) \circ \left(\bigotimes_{i=1}^{n} \mathcal{N}_i\right)[\rho]}_1\nonumber \\
    & + \frac{1}{2}\sup_{\rho}\norm{\CQ\circ\tCQ \circ\left(\bigotimes_{j=1}^{n'} \CW_j\right) \circ (\CE \otimes \CF) \circ \left(\bigotimes_{i=1}^{n} \mathcal{N}_i\right)[\rho] - \CE[\rho]}_1\\
    \leq\ &  2|\CE|(cp)^k + \frac{1}{2}\sup_{\rho}\norm{\CQ\circ\tCQ \circ\left(\bigotimes_{j=1}^{n'} \CW_j\right) \circ (\CE \otimes \CF) \circ \left(\bigotimes_{i=1}^{n} \mathcal{N}_i\right)[\rho] - \CE[\rho]}_1 ,
    \label{eq:enc2}
\end{align}
where first we have used the triangle inequality, then the data processing inequality, and lastly the fact that the $1$-norm is not greater than the diamond norm combined with the inequality in \cref{eq:christandl_ineq}.
In what follows, we bound the second term by getting rid of the channels $\mathcal{N}_i$ and invoking the QEC guarantee on the code $\CE$ to handle the channels $\CW_j$.

From \refcite[Theorem 59]{christandl2024faultTolerantQuantumInputOutput}, we have that each $\mathcal{N}_i$ acts on the $i$-th input qubit and
$\mathcal{N}_i:\CB(A_i) \rightarrow \CB(A_i \otimes F_i)$, where $A_i$ and $F_i$ are Hilbert spaces associated with the $i$-th input qubit and the environment, respectively, and
$\CB(\cdot)$ denotes linear operators on the Hilbert space. Note that in this notation, $\mathcal E: \CB(\bigotimes_{i=1}^n A_i) \rightarrow \CB(\bigotimes_{i=1}^{n'} A_i)$ and $\mathcal F: \CB(\bigotimes_{i=1}^n F_i) \rightarrow \CB(\bigotimes_{i=1}^{n'} F_i)$.
Furthermore, it holds that
\begin{equation}
    \mathrm{Tr}_{F_i} \circ\mathcal{N}_i=
    (1-\varepsilon_0)\mathcal{I}_{A_i}+\varepsilon_0\mathcal{Z}_{A_i},
    \label{eq:Ni}
\end{equation}
where $\varepsilon_0 \leq 2c p$ for some absolute constant $c$, $\mathcal{Z}_{A_i}$ is a CPTP map, and $\CI_{X}$ denotes the identity map on Hilbert space $X$.
Intuitively, this equation tells us that the channel $\mathcal{N}_i$ is close to the identity channel when restricted to the $i$-th qubit, meaning that its action on the environment qubits must be close to a channel that appends a fixed reference state $\tau$.
This can be formalized with the continuity theorem of Stinespring's dilation~\cite[Theorem 1]{kretschmann2008information} as follows.
First, consider the Stinespring dilation
$\mathcal{N}_i[\cdot] = \Tr_{G_i} \ (V_{A_i \rightarrow A_iF_iG_i}[\cdot] V_{A_i \rightarrow A_iF_iG_i}^\dagger)$, where $V_{A_i \rightarrow A_iF_iG_i}$ is an isometry.
The continuity theorem states that for two quantum channels $\mathcal{T}_1$, $\mathcal{T}_2: \CB(A) \rightarrow \CB(B)$, if $V_1$ and $V_2$ are the Stinespring dilating isometries with the same dilating Hilbert space $E$, then
\begin{align}
    \inf_{U} \norm{(\Id_B \otimes U_E) V_1 - V_2}^2 \leq  \norm{\mathcal{T}_1 -\mathcal{T}_2}_\diamond \leq 2 \inf_{U} \norm{(\Id_B \otimes U_E) V_1 - V_2},
\end{align}
where $\Id_X$ is the identity operator on Hilbert space $X$, and $\|\cdot\|$ denotes the operator norm,
and the infimum is taken over unitaries $U$ in $\CB(E)$.
We apply this theorem by taking $A=B=A_i$, $E=F_iG_i$,
$\mathcal{T}_1 = \Tr_{F_i} \circ \ \mathcal{N}_i$, $\mathcal{T}_2 = \mathcal I_{A_i}$, $V_1 = V_{A_i\rightarrow A_iF_iG_i}$ and $V_2$ to be an isometry that for all $\rho_{A_i}$ satisfies $V_2  \rho_{A_i} V_2^\dagger= \rho_{A_i} \otimes \tau_{F_iG_i}$ for some fixed state $\tau_{F_iG_i}$. 
Thus, the continuity theorem implies the existence of a unitary $U_{F_iG_i}$, such that
\begin{align}
    &\norm{(\Id_{A_i} \otimes U_{F_iG_i}) V_{A_i\rightarrow A_iF_iG_i} - V_2} \leq \| \Tr_{F_i} \circ \ \mathcal{N}_i -  \mathcal{I}_{A_i}\|_{\diamond}^{1/2} =  \sqrt{2\varepsilon_0}.
\end{align}

It follows that for each $i$ we can approximate $\mathcal{N}_i$ with a channel $\CR_i$ defined by $\CR_i[\rho_{A_i}] = \rho_{A_i} \otimes \sigma_{F_i}$, where $\sigma_{F_i} = \Tr_{G_i} (U_{F_iG_i}^\dagger \tau_{F_iG_i} U_{F_iG_i})$ and satisfies
\begin{align}
    \| \mathcal{N}_i - \CR_i  \|_\diamond \leq  2 \|V_{A_i \rightarrow A_iF_iG_i} -  (\Id_{A_i} \otimes U^\dagger_{F_iG_i}) V_2 \| \leq 2\sqrt{2\varepsilon_0} \leq 4\sqrt{cp}.
\end{align}
Thus, we obtain
\begin{align}
    \norm{\bigotimes_{i=1}^{n}\mathcal{N}_i - \bigotimes_{i=1}^{n}\CR_i}_{\diamond}
    = \norm{\sum_{i=1}^{n} \CR_1 \otimes \dots \otimes  \CR_{i-1} \otimes(\mathcal{N}_i-\CR_i) \otimes\mathcal{N}_{i+1}\otimes\dots\otimes\mathcal{N}_n}_{\diamond}
    \leq \sum_{i=1}^{n}\norm{\mathcal{N}_i-\CR_i}_{\diamond} \leq 4n\sqrt{cp},
\end{align}
allowing us to get rid of the channels $\mathcal{N}_i$ in \cref{eq:enc2}.
Namely, by invoking the triangle inequality and the fact that the 1-norm is not greater than the diamond norm we obtain
\begin{align}
    \varepsilon_{\rm enc} &\leq \frac{1}{2}\sup_{\rho}\norm{\CQ\circ\tCQ \circ\left(\bigotimes_{j=1}^{n'} \CW_j\right) \circ (\CE \otimes \CF) \circ \bigotimes_{i=1}^{n}\CR_i [\rho] - \CE[\rho]}_1 + 2 n\sqrt{cp} + 2|\CE|(cp)^k, \\
    &= \frac{1}{2}\sup_{\rho}\norm{\CQ\circ\tCQ \circ\left(\bigotimes_{j=1}^{n'} \CW_j\right) [\CE [\rho] \otimes \sigma] - \CE[\rho]}_1 + 2n\sqrt{cp} + 2|\CE|(cp)^k.
    \label{eq:enc3}
\end{align}
where we defined a state $\sigma = \CF(\otimes_{i=1}^n \sigma_{F_i})$ on the Hilbert space $\otimes_{j=1}^{n'} F_j$.

We now turn our attention to the channels $\CW_j$.
From \refcite[Theorem (59)]{christandl2024faultTolerantQuantumInputOutput}, we know that each $\CW_j$ acts on the $j$-th output qubit and 
$\CW_j:\CB(A_j \otimes F_j) \rightarrow \CB(A_j)$, where $A_j$ and $F_j$ are Hilbert spaces associated with the $j$-th output qubit and the environment, respectively.
Furthermore, we have
\begin{equation}
    \CW_j=(1-\varepsilon_0)\mathcal{I}_{A_j}\otimes\mathrm{Tr}_{F_j}+\varepsilon_0 \CY_{A_j},
\label{eq:Wi}
\end{equation}
where $\varepsilon_0 \leq 2cp$ and $c$ is the same absolute constant from~\cref{eq:Ni} and $\CY_{A_j}$ is a CPTP map.
We would like to use the above form of $\CW_j$ and the stochastic-error decoding guarantee of the code $\CE$ to bound the first term in~\cref{eq:enc3}. We have
\begin{align}
   \CQ\circ\tCQ \circ\left(\bigotimes_{j=1}^{n'} \CW_j\right) [\CE [\rho] \otimes \sigma]
   &=\sum_{ S \subseteq [n']} (1-\varepsilon_0)^{n'-|S|} \varepsilon_0^{|S|} \CQ\circ\tCQ \circ \left( \bigotimes_{j \notin S} (\mathcal{I}_j \otimes\mathrm{Tr}_{F_j})  \bigotimes_{j' \in S} \CY_{A_{j'}} \right) [\CE [\rho] \otimes \sigma] \\
   &=\sum_{ S \subseteq [n']} (1-\varepsilon_0)^{n'-|S|} \varepsilon_0^{|S|} \CQ\circ\tCQ \circ (\mathcal{I}_{S^c} \otimes \CW'_S)  [\CE [\rho]],
\end{align}
where $\CW'_S[\cdot] := \left(   \bigotimes_{j' \in S} \CY_{A_{j'}}\right)  \Big[[\cdot ]\otimes (\bigotimes_{j \notin S}\mathrm{Tr}_{F_j})[\sigma]\Big]  $ is a channel acting on the qubits $S \subseteq [n']$.
We now split the above sum into two parts depending whether or not the subset $S$ is correctable. For a correctable set $S$, we have
\begin{align}
    \CQ\circ\tCQ \circ (\mathcal{I}_{S^c} \otimes \CW'_S) [\CE [\rho]] = \CE [\rho].
\end{align}
On the other hand, according to the decoding guarantee of the fault-tolerant QEC gadget $\CQ_{\rm FT}$ on the code $\CE$, as long as $p$ (and accordingly $\varepsilon_0 \leq 2cp$) is sufficiently below a threshold value, the contribution of the uncorrectable sets is bounded follows
\begin{align}
    p_{\text{uncorrectable}} :=  \sum_{ S \text{ uncorrectable}} (1-\varepsilon_0)^{n'-|S|} \varepsilon_0^{|S|} \leq \Gamma(\CE),
\end{align}
where $\Gamma(\CE)$ is a quantity that can be exponentially driven to zero by choosing larger codes from the QEC family, as discussed in \cref{sec:FTQC} (see Eq.~\eqref{eq:local_stochastic_error}). 
Hence, we have that
\begin{align}
    &\norm{\CQ\circ\tCQ \circ\left(\bigotimes_{j=1}^{n'} \CW_j\right)[\CE [\rho] \otimes \sigma] -  \CE[\rho]}_1 \\
    \leq \ &\norm{\sum_{ S \text{ uncorrectable}} (1-\varepsilon_0)^{n'-|S|} \varepsilon_0^{|S|} \CQ\circ\tCQ \circ (\mathcal{I}_{S^c} \otimes \CW'_S)  [\CE [\rho]] -  p_{\text{uncorrectable}}\cdot \CE[\rho]}_1\\
    \leq \ & 2 \Gamma(\CE).
\end{align}

Collecting terms, we therefore conclude that
\begin{align}
    \varepsilon_\mathrm{enc} &\leq \Gamma(\CE) + 2\sqrt{cp} n  + 2|\CE|(cp)^k,
\end{align}
which completes the proof.
\end{proof}

\section{Delayed proofs for partial Clifford twirling}\label{sec:delayed_twirling_proofs}

First, we state a couple of propositions used in the proofs of the claims in the main text. 

\begin{proposition}\label{prop:off_diagonals_vanish}
Suppose $P, P' \in \PauliSet$ with $P \neq P'$ and $P \neq -P'$. Let $\EV_{C \sim \CliffordSet}$ denote expectation value over drawing $C$ randomly from $\CliffordSet$ as in \cref{def:twirling_set}. Then 
\begin{align}
    \EV_{C \sim \CliffordSet} C P C^\dag \otimes C P' C^\dag = 0
\end{align}
\end{proposition}
\begin{proof}
    Denote $P$ and $P'$ in canonical form as $P = \iUnit^{a \cdot  b}(-1)^s X^b Z^a$, and $P' = \iUnit^{a' \cdot  b'}(-1)^{s'} X^{b'} Z^{a'}$. If $P \neq \pm P'$, then it either must be the case that $a \neq a'$ or $b \neq b'$. Consider a fixed $C \in \CliffordSet$ defined by $(A,B, u,v)$. We define another $C' \in \CliffordSet$ for each case separately. If $a \neq a'$ then $a \oplus a' \neq 0^n$ and we can find a $w$ such that $w \cdot (a \oplus a') = 1$ (note that $w$ is actually independent of $C$). Then we define $C'$ by the choice $(A,B, u \oplus w, v)$. If $b \neq b'$ then we can find a $w$ such that $w \cdot (A^{-1}b \oplus A^{-1}b') =1$ ($w$ depends on $C$ only through $A$), and we define $C'$ by $(A, B, u, v\oplus w )$. In both cases, referring to \cref{eq:CPC^dag} from the proof of \cref{prop:twirling_uniform}, we have $CPC^\dag = \pm C'PC'^\dag$ and $CP'C^\dag = \mp C'P'C'^\dag$, or in other words,
    \begin{align}
        CPC^\dag \otimes CP'C^\dag = -C'PC'^\dag \otimes C'P'C'^\dag\,.
    \end{align}
    However, adding $w$ to either $u$ or to $v$ does not change the random distribution over $C \sim \CliffordSet$. This implies the expectation value is equal to its negation, and therefore it is zero. 
\end{proof}

\begin{proposition}\label{prop:odd_part_vanishes}
    Let $P \in \PauliSet_{\rm odd}$ (from \cref{def:signed_Pauli_strings}). Then, for any $g$, 
    \begin{align}
        P \ketbra{\Psi(g)} P \ket{\Psi(g)} = 0\,.
    \end{align}
\end{proposition}
\begin{proof}
    Write $P$ in canonical form as $\iUnit^{a \cdot  b}(-1)^{s}X^b Z^a$. Since $P \in \PauliSet_{\rm odd}$, we have that $a \cdot b = 1$. Hence, while $P$ is Hermitian, it also satisfies $P = -P^*$, where the ${}^*$ denotes complex conjugation. In other words, all of the matrix elements of $P$ in the computational basis are imaginary. The Hermiticity of $P$ implies that for any state $\ket{\xi}$, we have that $\bra{\xi} P \ket{\xi}$ is real, i.e., $\bra{\xi} P \ket{\xi}= (\bra{\xi} P \ket{\xi})^*$. However, if the entries of $\ket{\xi}$ in the computational basis are real, then 
    \begin{align}
        (\bra{\xi} P \ket{\xi})^* =  \bra{\xi} P^* \ket{\xi} = -\bra{\xi} P \ket{\xi}\,,
    \end{align}
    which implies that $\bra{\xi} P \ket{\xi} = 0$. It suffices to observe that the entries of $\ket{\Psi(g)}$ are real. 
\end{proof}

\subsection{Proof of uniform Pauli spreading}\label{sec:delayed_proof_uniform_spreading}
Here we restate and prove \cref{prop:twirling_uniform} from \cref{sec:partial_Clifford_twirl}. 
\uniformPauliProp*
\begin{proof}
    Recall that choosing $C\sim \CliffordSet$ means generating uniformly random $A, B, u,v$ and defining $C = Z^v Q_B M_A^\dag X^u$.  Let $P \in \PauliSet$ be a fixed Pauli string $P = \iUnit^{a \cdot  b}(-1)^s X^b Z^a $. We examine the effect of conjugation by each component of $C$ sequentially. We leave off the $\iUnit$ until the end since it is simply a constant. 
\begin{enumerate}
    \item Conjugation by $X$ gates: We have $X^u \big[ (-1)^{s}X^bZ^a \big] X^u = (-1)^{s\oplus (u \cdot a)}X^b Z^a $.
    \item  Conjugation by CNOT gates: First, we note that for any computational basis state $\ket{x}$, we have 
    \begin{align}
        M_A^\dag X^b Z^a  M_A\ket{x} &=M_A^\dag X^bZ^a \ket{Ax}  =  (-1)^{Ax \cdot  a}M_A^\dag\ket{Ax \oplus b} = (-1)^{ A^\top a \cdot x}\ket{x \oplus A^{-1}b} \\
        &=(-1)^{A^\top a \cdot x}X^{A^{-1}b}\ket{x} 
        = X^{A^{-1}b} Z^{A^\top a}\ket{x}
    \end{align}
    From this, we have
    \begin{align}
    M_A^\dag \big[(-1)^{s\oplus (u \cdot a)} Z^a X^b\big]M_A = (-1)^{s\oplus (u \cdot a)} X^{A^{-1}b} Z^{A^\top a} 
    \end{align}
    \item Conjugation by $\mathrm{CZ}$: The $Q_B$ gate is diagonal and Hermitian and commutes with Pauli-$Z$ strings. We notice for any $p \in \mathbb{F}_2^n$ and any $\ket{x}$ that
    \begin{align}
        Q_B X^p Q_B\ket{x} &= (-1)^{x^\top B x}Q_B X^p \ket{x} = (-1)^{x^\top B x}Q_B \ket{x \oplus p} = (-1)^{\big[x^\top B x\big] \oplus \big[(x \oplus p)^\top B (x \oplus p)\big]}\ket{x \oplus p} \\
        &= (-1)^{\big[x^\top B x\big] \oplus \big[(x \oplus p)^\top B (x \oplus p)\big]}X^p \ket{x} = (-1)^{\big[(Bp \oplus B^\top p) \cdot x\big] \oplus \big[p^\top B p\big]}X^p \ket{x} \\
        &= (-1)^{p^\top B p}X^p Z^{(Bp \oplus B^\top p)} \ket{x}
    \end{align}
    and thus
    \begin{align}
        Q_B \Big[(-1)^{s\oplus (u \cdot a)} X^{A^{-1}b} Z^{A^\top a}\Big] Q_B = (-1)^{s\oplus \big[u \cdot a\big]\oplus \big[b^\top A^{-1 \top} B A^{-1} b\big]} X^{A^{-1}b} Z^{A^\top a \oplus \big[(B \oplus B^\top) A^{-1}b\big]}
    \end{align}
    \item Conjugation by $Z$: Finally, the $Z^v$ gate is commuted, which only modifies the sign
    \begin{align}
        &Z^v\Big[\iUnit^{a \cdot  b}(-1)^{s\oplus \big[u \cdot a\big]\oplus \big[b^\top A^{-1 \top} B A^{-1} b\big]} X^{A^{-1}b} Z^{A^\top a \oplus \big[(B \oplus B^\top) A^{-1}b\big]}\Big]Z^v \\
        ={}& \iUnit^{a \cdot  b}(-1)^{s\oplus \big[u \cdot a\big]
        \oplus \big[b^\top A^{-1 \top} B A^{-1} b\big]
        \oplus \big[ v \cdot A^{-1}b\big]} 
        X^{A^{-1}b} Z^{A^\top a \oplus \big[(B \oplus B^\top) A^{-1}b\big]} \label{eq:CPC^dag}
        \\
        ={}& CPC^\dag
    \end{align}
\end{enumerate}
    Now, we verify that the claims of the uniformity of the distribution over $Q = C P C^\dag$ when $A,B,u,v$ are chosen uniformly at random. We begin with the easy cases.  
    \begin{itemize}
        \item If $P \in \PauliSet_0$ or $P \in \PauliSet_1$, then the statement is trivial, as the subset has only one element, and substituting $a=b=0^n$ confirms that $CPC^\dag = P$.
        \item If $P \in \PauliSet_Z$, then $b=0^n$ and $a \neq 0^n$. We have 
        \begin{align}
            CPC^\dag = \iUnit^{a \cdot  b}(-1)^{s \oplus \big[u \cdot a\big] }Z^{A^\top a }
        \end{align}
        For $A$ a uniformly random invertible binary matrix, $A^\top a$ is a uniformly random nonzero element of $\mathbb{F}_2^n$. Moreover, the uniformly random choice of $u \in \mathbb{F}_2^n$ ensures that the sign is uniformly random. Together, these facts confirm that $CPC^\dag$ is uniformly random over $\PauliSet_Z$. 
    \end{itemize} 
    To handle the final two cases, we now examine the term $(B \oplus B^\top)A^{-1}b$. We have defined $B$ as a random upper triangular matrix. Note, however, that for the matrix $B \oplus B^\top$, one obtains the same distribution regardless of whether $B$ is upper triangular or arbitrary, so long as the entries of $B$ are chosen uniformly at random from $\mathbb{F}_2$. Furthermore, if $V$ is an invertible matrix, then $B \mapsto V^\top B V$ is a bijective map (its inverse is $B \mapsto V^{\top -1} B V^{-1}$). Thus, if $B$ is chosen uniformly at random from all $n \times n$ matrices in $\mathbb{F}_2$  then the distribution over $V^\top B V$ is also the uniform distribution over all matrices in $\mathbb{F}_2$, so we may replace $B$ with $V^\top BV$ and not change the distribution. 
    For fixed $A$ and nonzero $b$, we can find an invertible $V$ such that $VA^{-1}b = e$, where $e = (1,0,0,\ldots,0) \in \mathbb{F}_2^n$ is the string with a single 1 in its first entry. For this choice of $V$, we have
    \begin{align}\label{eq:VtopBplusBtope}
        (V^\top B V \oplus V^\top B^\top V)A^{-1}b = V^\top (B \oplus B^\top) e
    \end{align}
    We observe that $(B \oplus B^\top)e$ has a zero in its first entry, regardless of the choice of $B$, but the other $n-1$ entires are uniformly random when $B$ is uniformly random, that is, $(B \oplus B^\top)e$ is distributed uniformly over the subset of $\mathbb{F}_2^n$ orthogonal to $e$ (in the sense of having dot product zero). Letting $c= (B \oplus B^\top)e$, we have $0 =  c^\top e = (V^\top c)^\top V^{-1}e = (V^\top c)^\top A^{-1}b$, so $A^{-1}b$ is orthogonal to $V^\top c$. In fact, since $c$ is distributed uniformly over all vectors orthogonal to $e$ and $V$ is bijective, $V^\top c$ is distributed uniformly over all vectors orthogonal to $A^{-1}b$. As noted above, since $V$ is a bijective map, the distribution of $V^\top B V$ is the same as that of $B$ itself, so we can replace $V^\top B V$ with $B$ in \cref{eq:VtopBplusBtope}, and hence we see that $(B\oplus B^\top)A^{-1}b=V^\top c$. Thus, the quantity $(B \oplus B^\top)A^{-1}b$ is also distributed uniformly over all vectors orthogonal to $A^{-1}b$. 
    We now continue with the final two cases. 
    \begin{itemize}
        \item If $P \in \PauliSet_{\rm odd}$, then $a,b \neq 0^n$ (otherwise $a\cdot b = 1$ would be impossible). The quantity $CPC^\dag$ can be written in canonical form as $\iUnit^{a' \cdot  b'}(-1)^{s'} X^{b'}Z^{a'}$, where $s'$, $b'$, $a'$ are given by \cref{eq:CPC^dag}. First, we immediately observe that randomizing over $u \in \mathbb{F}^n$ leads $u \cdot a$ to be uniformly random, and thus $s'$ is uniformly random over $\mathbb{F}_2$, independent of $a'$ and $b'$. Since $A$ is distributed randomly over invertible matrices, $b' = A^{-1}b$ is distributed uniformly at random over  $\mathbb{F}_2^n \setminus \{0^n\}$. Next, for fixed $A$ (and thus fixed $b'$), the above reasoning establishes that,
        \begin{align}
            a' = A^\top a \oplus \big[(B \oplus B^\top) A^{-1}b\big] = A^\top a \oplus r
        \end{align}
        where $r$ is distributed uniformly at random over vectors orthogonal to $b' = A^{-1}b$. Recall that $(A^\top a) \cdot (A^{-1}b) = a \cdot b = 1$. This means that
        \begin{align}
            a' \cdot b' = \big[(A^\top a) \cdot (A^{-1}b)\big] \oplus \big[r \cdot A^{-1} b \big] = a \cdot b = 1
        \end{align}
        Since $r$ is distributed at random among all vectors for which $b' \cdot r = 0$, the distribution over $a'$ is precisely the uniform distribution over all vectors for which $a' \cdot b' = 1$. This confirms that $CPC^\dag$ is distributed uniformly over $\PauliSet_{\rm odd}$. 
        \item If $P \in \PauliSet_{\rm even}$, then $b \neq 0^n$ and $a \cdot b = 0$. As before, we have $CPC^\dag = \iUnit^{a' \cdot  b'}(-1)^{s'}X^{b'}Z^{a'}$, and $b' = A^{-1} b$ is distributed uniformly at random over $\mathbb{F}_2^n \setminus \{0^n\}$. Furthermore, since $v \in \mathbb{F}_2^n$ is uniformly and indepenently chosen, the quantity $v \cdot (A^{-1}b)$ is uniformly random, meaning $s'$ is uniformly random and independent of $a'$ and $b'$. Next, we can again write
        \begin{align}
            a' = A^\top a \oplus r
        \end{align}
        where $r$ is orthogonal to $b'$, and $A^\top a$ is a fixed vector for which $b' \cdot (A^\top a) = b \cdot a = 0$. Thus, the distribution over $a'$ is uniform over vectors orthogonal to $b'$, confirming that $C PC^\dag$ is uniform over the set $\PauliSet_{\rm even}$. 
    \end{itemize}
\end{proof}

\subsection{Proof of correct top eigenvector}\label{sec:delayed_proof_correct_top_eigenvector}

Here we restate and prove \cref{prop:correct_top_eigenvector_after_twirling} from \cref{sec:partial_Clifford_twirl}. 
\correctTopEigenvector*
\begin{proof}
    From the definition of dataset-independent noise (\cref{def:dataset-independent_noise}) and the definiton of $\ol{\phi(g)}$ in \cref{eq:noisy_logical_resource_state}, we have that 
    \begin{align}
        \ol{\phi(g)} = \CQ \circ \tCQ_{\rm FT} \circ \tCE_{\rm FT} \circ \CN_2 \circ \CV(g) \circ \CN_1[\ketbra{+}^{\otimes n}]\,.
    \end{align}
    Since the ideal encdoing map $\CE$ is injective and $\CQ$ projects onto its image (i.e., the codespace), we may define the channel $\CN_{\rm enc} = \CE^{-1} \circ \CQ \circ \tCQ_{\rm FT} \circ \tCE_{\rm FT} $ (where $\CE^{-1}$ is well defined on inputs in the codespace) and note that
    \begin{align}
        \CQ \circ \tCQ_{\rm FT} \circ \tCE_{\rm FT} = \CE \circ \CN_{\rm enc}\,.
    \end{align}
    Thus, we have
    \begin{align}
    \ol{\phi(g)} = \CE\Big[\CN_{\rm enc}\Big[\CN_2\Big[V(g) \, \CN_1\Big[\ketbra{+}^{\otimes n} \Big] V(g) \Big ]\Big]\Big]
    \end{align}
    The $2^n$ states $\{Z^u \ket{+}, u \in \{0,1\}^n\}$ form an orthonormal basis, and we may write the state $\CN_1\Big[\ketbra{+}^{\otimes n}\Big]$ as a fixed $g$-independent density matrix in this basis
    \begin{align}
        \CN_1\Big[\ketbra{+}^{\otimes n}\Big] = \sum_{u,v \in \{0,1\}^n} M_{uv} Z^u \ketbra{+}^{\otimes n} Z^v
    \end{align}
    The operation $V(g)$ is diagonal and commutes with $Z^u$ and $Z^v$, so we have 
        \begin{align}
        V(g)\Big[ \CN_1\Big[\ketbra{+}^{\otimes n}\Big]\Big] V(g) = \sum_{u,v \in \{0,1\}^n} M_{uv} Z^u \ketbra{\Psi(g)} Z^v
    \end{align}
    The channel $\CN_{\rm enc} \circ \CN_2$ in general has a Pauli decomposition
    \begin{align}
        \CN_{\rm enc} \circ \CN_{2} [\sigma] = \sum_{P, P' \in \PauliSet} \chi_{PP'} P \sigma P' 
    \end{align}
    where $\chi$ is some noise matrix, and $\PauliSet$ denotes the set of signed Pauli strings, as in \cref{def:signed_Pauli_strings}. Moreover, $\CE$ is the physical-to-logical mapping which is equivalent to putting an overline on all the objects. 
    Thus, we have
    \begin{align}
        \ol{\phi(g)} = \sum_{P,P' \in \PauliSet}\sum_{u,v \in \{0,1\}^n} \chi_{PP'}M_{uv}  \ol{P}\,  \ol{Z}^{ u} \ketbra{\ol{\Psi(g)}} \ol{Z}^{v}\, \ol{P}'\,.
    \end{align}
    Now we compute the expectation over the twirled state, recalling the equality from \cref{prop:twirl_consistency}.
    \begin{align}
        \ol{\phi(g)}_{\rm twirl} &= \EV_{C \sim \CliffordSet} \sum_{P,P' \in \PauliSet}\sum_{u,v \in \{0,1\}^n}\chi_{PP'}M_{uv}  \ol{C}\, \ol{P}\,  \ol{Z}^{ u} \ketbra{\ol{\Psi(g_C)}} \ol{Z}^{v}\, \ol{P}' \, \ol{C}^\dag  \\
        &= \EV_{C \sim \CliffordSet} \sum_{P,P' \in \PauliSet}\sum_{u,v \in \{0,1\}^n}\chi_{PP'}M_{uv}  \ol{C}\, \ol{P}\,  \ol{Z}^{ u}\, \ol{C}^\dag \ketbra{\ol{\Psi(g)}} \ol{C} \, \ol{Z}^{v}\, \ol{P}' \, \ol{C}^\dag \\
        &= \EV_{C \sim \CliffordSet} \sum_{P,P' \in \PauliSet} \chi'_{PP'}  \ol{C}\, \ol{P}\, \ol{C}^\dag \ketbra{\ol{\Psi(g)}} \ol{C} \, \ol{P}' \, \ol{C}^\dag
    \end{align}
    where we have absorbed the $\ol{Z}^u$ and $\ol{Z}^v$ into $\ol{P}$ and $\ol{P'}$, leading to a redefinition of $\chi_{PP'}$ to $\chi'_{PP'}$. We may move the expectation value inside the sum. The result of \cref{prop:off_diagonals_vanish} ensures that if $P \neq P'$, the expectation value vanishes, leaving only diagonal terms where $P = P'$ (if $P = -P'$ then the term does not immediately vanish, but without loss of generality we may take $\chi'_{-P,P} = 0$ for all $P$ by appropriately redefining the value of $\chi'_{P,P}$). This gives
    \begin{align}
        \ol{\phi(g)}_{\rm twirl} =  \sum_{P \in \PauliSet} \chi'_{PP}  \EV_{C \sim \CliffordSet} \ol{C}\, \ol{P}\, \ol{C}^\dag \ketbra{\ol{\Psi(g)}} \ol{C} \, \ol{P} \, \ol{C}^\dag\,.
    \end{align}
    Recall from \cref{def:signed_Pauli_strings} that we can partition the set $\PauliSet$ into sets $\PauliSet_0$, $\PauliSet_1$, $\PauliSet_{Z}$, $\PauliSet_{\rm even}$, and $\PauliSet_{\rm odd}$. Note that $\PauliSet_0$ and $\PauliSet_1$ each contain a single element, which differs by a sign that cancels out in the expression above. Without loss of generality, we may assume that $\chi'_{PP} = 0$ for $P = -\Id$ and ignore the set $\PauliSet_1$ henceforth. The fact that $\ol{\phi(g)}_{\rm twirl}$ has trace 1 means that $\sum_{P \in \PauliSet} \chi'_{PP} = 1$; we define $\chi'_{0}$, $\chi'_Z$, $\chi'_{\rm even}$, and $\chi'_{\rm odd}$ as the sum of $\chi'_{PP}$ over $P$ from the relevant set. 
    
    The result of \cref{prop:twirling_uniform} states that for each $P$, $CPC^\dag$ is distributed uniformly over the subset of $\PauliSet$ that contains $P$. Thus, for all $P$ in the same subset the quantity $\EV_{C \sim \CliffordSet} \ol{C}\, \ol{P}\, \ol{C}^\dag \ketbra{\ol{\Psi(g)}} \ol{C} \, \ol{P} \, \ol{C}^\dag$ is the same.  We may thus define
    \begin{align}
        \ol{\rho(g)}_{Z} &= \EV_{P \sim \PauliSet_Z} \ol{P} \ketbra{\ol{\Psi(g)}} \ol{P} \\
        \ol{\rho(g)}_{\rm even} &= \EV_{P \sim \PauliSet_{\rm even}} \ol{P} \ketbra{\ol{\Psi(g)}} \ol{P} \label{eq:rho_even}\\
        \ol{\rho(g)}_{\rm odd} &= \EV_{P \sim \PauliSet_{\rm odd}} \ol{P} \ketbra{\ol{\Psi(g)}} \ol{P}\label{eq:rho_odd}
    \end{align}
    and write 
    \begin{align}\label{eq:decomp_state_into_subsets}
        \ol{\phi(g)}_{\rm twirl} &=   \chi'_{0}   \ketbra{\ol{\Psi(g)}}  +   \big[\chi'_Z \ol{\rho(g)}_Z\big]  +   \big[\chi'_{\rm even}  \ol{\rho(g)}_{\rm even}\big] +   \big[\chi'_{\rm odd} \ol{\rho(g)}_{\rm odd}\big]
    \end{align}
    Next, we explicitly compute the term of $Z$-like Paulis. \begin{align}\label{eq:expectation_Ztype_Paulis}
        \ol{\rho(g)}_{Z} = \frac{1}{2^n-1} \sum_{0^n \neq u \in \{0,1\}^n} \ol{Z}^u \ketbra{\ol{\Psi(g)}} \ol{Z}^u = \frac{\ol{\Id}}{2^n-1} - \frac{\ketbra{\ol{\Psi(g)}}}{2^n-1}\,,
    \end{align} 
    which follows since the set of $\ol{Z}^u\ket{\ol{\Psi(g)}}$ form an orthonormal basis for $u \in \{0,1\}^n$, and sampling a random vector from an orthonormal basis yields the maximally mixed state. 
    
    Additionally, we can use the fact that the Pauli matrices form a 1-design \cite{Mele2024introductionToHaar} to say that if we draw a $P$ uniformly at random from all of $\PauliSet$, conjugating an arbitrary state $\ketbra{\ol{\xi}}$ by $P$ yields the maximally mixed state after averaging over $P$. Weighting the different subsets by their sizes (e.g., a $(2^n-1)/2^{2n}$ fraction of all signed Pauli strings are in $\PauliSet_{Z}$), this fact is equivalent to
    \begin{align}
        \frac{1}{2^{2n}}\ketbra{\ol{\xi}} + \frac{2^n-1}{2^{2n}}\EV_{P \in \PauliSet_Z} \ol{P} \ketbra{\ol{\xi}} \ol{P} + \frac{2^n-1}{2^{n+1}}\EV_{P \in \PauliSet_{\rm even}} \ol{P} \ketbra{\ol{\xi}} \ol{P} + \frac{2^n-1}{2^{n+1}}\EV_{P \in \PauliSet_{\rm odd}} \ol{P} \ketbra{\ol{\xi}} \ol{P} = \frac{\ol{\Id}}{2^n}
    \end{align}
    We apply this facts to the state $\ketbra{\ol{\Psi(g)}}$, substitute \cref{eq:expectation_Ztype_Paulis}, and rearrange to write
    \begin{align}
        \ol{\rho(g)}_{\rm even} = 2 \frac{\ol{\Id}}{2^n} - \ol{\rho(g)}_{\rm odd} \label{eq:rho_even_rho_odd}
    \end{align}
    Next, we plug this into \cref{eq:decomp_state_into_subsets} to get rid of $\rho(g)_{\rm even}$ and say
    \begin{align}\label{eq:phi_twirl_expanded}
        \ol{\phi(g)}_{\rm twirl} &=   (\chi'_{0} - \frac{\chi'_Z}{2^n-1})   \ketbra{\ol{\Psi(g)}}   +   \big[\big(\frac{2^n}{2^n-1} \chi'_Z + 2\chi'_{\rm even}\big)  \frac{\Id}{2^n}\big] +   \big[(\chi'_{\rm odd}-\chi'_{\rm even}) \ol{\rho(g)}_{\rm odd}\big]
    \end{align}
    
    We are now ready to conclude. A consequence of \cref{prop:odd_part_vanishes} is that $\ol{\rho(g)}_{\rm odd} \ket{\ol{\Psi(g)}} = 0$. Thus, we have
    \begin{align}
        \ol{\phi(g)}_{\rm twirl}\ket{\ol{\Psi(g)}} = \lambda_{\rm twirl} \ket{\ol{\Psi(g)}}
    \end{align}
    with  $\lambda_{\rm twirl} = \chi_0' + 2^{-n+1} \chi'_{\rm even}$. From the definition of $\ol{\phi(g)}_{\rm twirl}$ and the relation in \cref{prop:twirl_consistency}, we have
    \begin{align}
        \lambda_{\rm twirl} =  \bra{\ol{\Psi(g)}}\ol{\phi(g)}_{\rm twirl}\ket{\ol{\Psi(g)}} = \EV_{C \sim \CliffordSet} \bra{\ol{\Psi(g)}}\ol{C} \,  \ol{\phi(g_C)} \, \ol{C}^\dag\ket{\ol{\Psi(g)}} = \EV_{C \sim \CliffordSet} \bra{\ol{\Psi(g_C)}}  \ol{\phi(g_C)} ^\dag\ket{\ol{\Psi(g_C)}} \geq F_{\rm min}\,.
    \end{align}
    Finally, we reason about the other eigenvalues of $\ol{\phi(g)}_{\rm twirl}$. We may note from their definitions in \cref{eq:rho_even} and \cref{eq:rho_odd} that $\ol{\rho(g)}_{\rm even}$ and $\ol{\rho(g)}_{\rm odd}$ are positive semidefinite operators. From \cref{eq:rho_even_rho_odd}, this means the eigenvalues of $\ol{\rho(g)}_{\rm odd}$ cannot exceed $2^{-n+1}$. We now turn back to \cref{eq:phi_twirl_expanded} to bound any other eigenvalue $\lambda_{\rm other}$ of an eigenvector $\ket{\ol{\lambda_{\rm other}}}$ of $\ol{\phi(g)}_{\rm twirl}$. We have
    \begin{equation}
        \ol{\phi(g)}_{\rm twirl} \ket{
        \ol{\lambda_{\rm other}}} = \big[\big(\frac{2^n}{2^n-1} \chi'_Z + 2\chi'_{\rm even}\big)  \frac{1}{2^n}\big]\ket{\ol{\lambda_{\rm other}}} +   \big[(\chi'_{\rm odd}-\chi'_{\rm even}) \ol{\rho(g)}_{\rm odd}\big]\ket{\ol{\lambda_{\rm other}}}\,.
    \end{equation}
    Thus, $\ket{\ol{\lambda_{\rm other}}}$ must also be an eigenvector of $\ol{\rho(g)}_{\rm odd}$; assume the associated eigenvalue is $C\geq 0$. We have, 
    \begin{align}
        \lambda_{\rm other} &= \left(\frac{2^n}{2^n-1} \chi'_Z + 2\chi'_{\rm even}\right)  \frac{1}{2^n} +   (\chi'_{\rm odd}-\chi'_{\rm even})C\label{eq:lambda_other_difference}\\
        &=\frac{1}{2^n-1} \chi'_Z + \chi'_{\rm even}\left(\frac{1}{2^{n-1}} -C\right)+\chi'_{\rm odd}C
        \label{eq:lambda_other}
    \end{align}
    Suppose $\chi'_{\rm even} \geq \chi'_{\rm odd}$. In that case, the last term in \cref{eq:lambda_other_difference} is negative (or zero), so we may bound
    \begin{align}
        \lambda_{\rm other} \leq \frac{1}{2^n-1}\chi'_Z + \frac{2}{2^n} \chi'_{\rm even}.
    \end{align}
    On the other hand, suppose $\chi'_{\rm even} < \chi'_{\rm odd}$. In that case, we have that
     \begin{align}
        \lambda_{\rm other} &\leq \frac{1}{2^n-1} \chi'_Z + \chi'_{\rm odd}\left(\frac{1}{2^{n-1}} -C\right)+C\chi'_{\rm odd}\\
        &= \frac{1}{2^n-1} \chi'_Z + \chi'_{\rm odd}\frac{1}{2^{n-1}}.
    \end{align}
    Thus, we can combine the expressions and conclude that
    \begin{equation}
        \lambda_{\rm other} \leq \frac{1}{2^n-1} \chi'_Z + \frac{1}{2^{n-1}}\max (\chi'_{\rm even},\chi'_{\rm odd}).
    \end{equation}
    Finally, since $\chi'_0 +\chi'_Z+\chi'_{\rm even} + \chi'_{\rm odd} = 1$ with each term in the sum non-negative, it immediately follows that no eigenvalue $\lambda_{\rm other}$ can be larger than $2^{-n+1}$. 
\end{proof}

\section{QRAM is in the Clifford hierarchy}\label{app:QRAM_in_Clifford_hierarchy}

Here we show that for any $f$, the $n$-qubit QRAM gate $\ol{V(f)}$ from \cref{eq:QRAM_intro} is in the $n$-th level of the logical Clifford hierarchy, $\CC_n$ from \cref{eq:Clifford_hierarchy}, and can therefore be teleported via the strategy sketched in \cref{sec:teleportable_gates}. This can be seen directly as a consequence of \refcite[Theorem 3]{cui2017diagonalGatesCliffordHierarchy}, but here we give a self-contained proof. 

We view the classical dataset $f$ as an $n$-bit Boolean function $f \colon \{0,1\}^n \rightarrow \{0,1\}$, which we expand as a polynomal of its input bits over the field $\mathbb{F}_2$ (i.e., addition and multiplication are done modulo 2). We consider all possible $2^n$ monomials, denoted by $x^e$, where $e \in \{0,1\}^n$ dictates the exponents of each bit
\begin{align}
    x^e  = x_1^{e_1}x_2^{e_2}\cdots x_n^{e_n}\,.
\end{align}
Then, we may uniquely write
\begin{align}\label{eq:Boolean_function_as_polynomial}
    f(x_1,\ldots,x_n) = \bigoplus_{e \in \{0,1\}^n} c_e x^e
\end{align}
where $c_e \in \{0,1\}$ is the binary coefficient for the monomial $x^e$, and $\oplus$ denotes addition modulo 2. The degree of $f$ is then given by the maximum degree of any monomial that appears in this expansion, that is,
\begin{align}
    \deg(f) = \max_{e \in \{0,1\}^n} c_e |e|\,,
\end{align}
where $|e|=\sum_{i=1}^n e_i$ is the Hamming weight of $e$. 

We will prove by induction that $\ol{V(f)} \in \CC_{\deg(f)}$. As the base case, we observe that if $\deg(f) =1$, then there is a bit string $u \in \{0,1\}^n$ for which $f(x) = \bigoplus_{i=1}^n u_i x_i$. Thus, we have that $\ol{V(f)} = \ol{Z}^{u}$, where
$\ol{Z}^u$ denotes the (logical) Pauli operator with $\ol{Z}$ in positions where $u_i=1$ and $\ol{\Id}$ in other positions. That is, for degree-1 functions $f$, the QRAM operation is a Pauli operator,  $\ol{V(f)} \in \CC_1$.

Next, we assume for induction that for any degree-$(d-1)$ function $g$, the gate $\ol{V(g)} \in \CC_{d-1}$. We consider a degree-$d$ function $f$, and we would like to show that this implies $\ol{V(f)} \in \CC_d$. Note that $\ol{V(f)} = \ol{V(f)}^\dag$, and that the family of QRAM operations obeys the general composition rule
\begin{align}
    \ol{V(f)} \, \ol{V(h)} = \ol{V(f \oplus h)}
\end{align}
for any pair of $n$-bit Boolean functions $f,h$, where $f\oplus h$ is the function for which $(f \oplus h)(x) = f(x) \oplus h(x)$. Thus, based on the decomposition in \cref{eq:Boolean_function_as_polynomial}, we have
\begin{align}\label{eq:QRAM_unitary_decomp}
    \ol{V(f)} = \prod_{e \in \{0,1\}^n} \ol{V(x^e)}^{c_e}
\end{align}
Consider a factor $\ol{V(x^e)}$ associated with a degree-$|e|$ monomial $x^e$. For any $X$-type Pauli string $X^m$, we compute 
\begin{align}
    \ol{X}^m \ol{V(x^e)} = \ol{V((x\oplus m)^{e})}\,  \ol{X}^m\,,
\end{align}
and thus
\begin{align}
    \ol{V(x^e)}\, \ol{X}^m \ol{V(x^e)} = \ol{V(x^e)}\,\ol{V((x\oplus m)^{e })} \, \ol{X}^m = \ol{V(x^e \oplus (x\oplus m)^{e})}\, \ol{X}^m\,.
\end{align}
The crucial observation is that for any fixed $m$,  the Boolean function $x^e \oplus (x\oplus m)^{e}$, viewed as a function of $x_1,\ldots, x_n$, has degree at most $|e|-1$. We see this by inspection of the expression
\begin{align}\label{eq:cancellation_of_monomials}
    x^e \oplus (x\oplus m)^{e} = \big[x_1^{e_1}\cdots x_n^{e_n}\big] \oplus \big[(x_1\oplus m_1)^{e_1}(x_2 \oplus m_2)^{e_2}\cdots (x_n \oplus m_n)^{e_n}\big] \,,
\end{align}
noting that distributing the multiplication of the second term gives $2^{|e|}$ terms, one of which has degree $|e|$ and will cancel the first term. The rest of the terms have degree $|e|-1$ or lower. 
Since $f$ is a sum of monomials as in \cref{eq:Boolean_function_as_polynomial} and $V(f)$ decomposes into commuting diagonal factors as $V(f) = \prod_{e \in \{0,1\}^n} \ol{V(x^e)}^{c_e}$ from \cref{eq:QRAM_unitary_decomp}, we have
\begin{align}
    \ol{V(f)}\, \ol{X}^m \, \ol{V(f)} = \left[\prod_{e \in \{0,1\}^n} \ol{V(x^e \oplus (x\oplus m)^{e})}^{c_e}\right] \ol{X}^m = \ol{V(g)} \, \ol{X}^m\,,
\end{align}
where $g$ is a Boolean function satisfying $g(x) = \bigoplus_{e \in \{0,1\}^n} c_e(x^e \oplus (x\oplus m)^{e})$. Since $f$ has degree $d$, $g$ is the sum of degree $d-1$ functions, and thus has degree at most $d-1$. By the inductive assumption $\ol{V(g)} \in \CC_{d-1}$, which implies that $\ol{V(g)}\, \ol{X}^m \in \CC_{d-1}$ as well. Furthermore, since $\ol{V(f)}$ is diagonal, it is straightforward to see that $\ol{V(f)}\, \ol{Z}^m\,  \ol{V(f)} = \ol{Z}^m \in \CC_{d-1}$ for any $m$. Finally, since any Pauli string $\ol{P}$ is proportional to $\ol{X}^a \, \ol{Z}^b$ for some $a, b \in \{0,1\}^n$, these facts together imply $\ol{V(f)} \, \ol{X}^a \, \ol{Z}^b \, \ol{V(f)} = \ol{V(g)}\,\ol{X}^m\,\ol{Z}^b \in \CC_{d-1}$; thus, $\ol{V(f)}\, \ol{P}\, \ol{V(f)} \in \CC_{d-1}$ for any Pauli string $\ol{P}$. By induction, we conclude that $\ol{V(f)} \in \CC_{\deg(f)}$ for all $f$. Since the maximum degree of any $f$ is $n$, we have $\ol{V(f)} \in \CC_n$. 

\addcontentsline{toc}{section}{\refname}
\printbibliography
\end{document}